\documentclass[12pt]{article}
\newif\ifblind
\blindfalse

\newif\ifarxiv
\arxivtrue

\usepackage[affil-it]{authblk} 
\usepackage{amssymb,amsbsy,amsfonts,amsmath,amsthm,xspace}
\usepackage{lscape}
\usepackage{multibib}
\usepackage[round]{natbib}
\usepackage{graphicx}
\usepackage{subcaption}
\graphicspath{{figures/}}
\usepackage{color}
\usepackage[hidelinks]{hyperref}
\usepackage{bm}
\usepackage{comment}
\usepackage{tikz}
\usepackage{booktabs}

\newcites{body}{References}
\newcites{supp}{References}

\def\twoImages#1#2#3#4#5#6
{
\centerline{\hfill
\includegraphics[width=#2]{#1}
\hfill
\includegraphics[width=#5]{#4}
\hfill
}
}

\def\fourImages#1#2#3#4#5#6
{
\centerline{\hfill
\includegraphics[width=#5]{#1}
\hfill
\includegraphics[width=#6]{#2}
\hfill
}
\centerline{\hfill
\includegraphics[width=#5]{#3}
\hfill
\includegraphics[width=#6]{#4}
\hfill
}
}

\newcommand{\indist}{\overset{d}{\rightarrow}}
\newcommand{\inprob}{\overset{P}{\rightarrow}}

\newcommand{\tp}{\intercal}

\newcommand{\pd}[2]{\frac{\partial #1}{\partial #2}}

\newcommand{\expect}{\mathbb{E}}
\newcommand{\prob}{\mathbb{P}}
\newcommand{\var}{\operatorname{Var}}
\newcommand{\cov}{\operatorname{Cov}}
\newcommand{\acov}{\operatorname{aCov}}

\newcommand{\real}{\mathbb{R}}

\newcommand{\realp}{\mathbb{R}^p}

\newcommand{\Anet}{\bm{A}}
\newcommand{\Ynet}{\bm{Y}}
\newcommand{\Anetk}[1]{\bm{A}^{(#1)}}
\newcommand{\Ynetk}[1]{\bm{Y}^{(#1)}}

\newcommand{\Yvec}{\vec{\bm{Y}}}

\newcommand{\subtri}[1]{
  #1_{\tikz[baseline=-0.5ex, scale=0.3]{
    \coordinate (A) at (0,0);
    \coordinate (B) at (1,0);
    \coordinate (C) at (0.5,0.866); 
    \draw[line width=0.5pt] (A) -- (B) -- (C) -- (A);
  }}
}

\newcommand{\subtstar}[1]{
  #1_{\tikz[baseline=-0.5ex, scale=0.3]{
    \coordinate (A) at (0,0);
    \coordinate (B) at (1,0);
    \coordinate (C) at (0.5,0.866);
    \draw[line width=0.5pt] (B) -- (C) -- (A);
  }}
}

\newcommand{\subtchain}[1]{
  #1_{\tikz[baseline=-0.5ex, scale=0.3]{
    \coordinate (A) at (0,0);
    \coordinate (B) at (1,0);
    \coordinate (C) at (0.5,0.866);
    \draw[line width=0.5pt] (A) -- (B) -- (C);
    \draw[dotted, line width=0.5pt] (C) -- (A);
  }}
}

\newcommand{\subhchain}[1]{
  #1_{\tikz[baseline=-0.5ex, scale=0.3]{
    \coordinate (A) at (0,0);
    \coordinate (B) at (1,0);
    \coordinate (C) at (0.5,0.866);
    \draw[line width=0.5pt] (B) -- (C);
    \draw[dotted, line width=0.5pt] (C) -- (A);
  }}
}

\theoremstyle{definition}

\newtheorem{proposition}{Proposition}
\newtheorem{corollary}{Corollary}
\newtheorem{lemma}{Lemma}
\newtheorem{remark}{Remark}

\newtheorem{assumption}{Assumption}

\title{Inference for subgraph densities in noisy dynamic networks}
\ifblind
\author{\empty}
\else
\author[1]{Peter W. MacDonald}
\affil[1]{\footnotesize University of Waterloo, \ifarxiv \else Waterloo ON Canada N2L 3G1, \fi {\tt pwmacdonald@uwaterloo.ca}} 
\author[2]{Eric D. Kolaczyk}
\affil[2]{\footnotesize McGill University, \ifarxiv \else Montr\'eal QC Canada H3A 0B9, \fi {\tt eric.kolaczyk@mcgill.ca}}
\fi
\date{August 5, 2026}

\begin{document}

\maketitle

\begin{abstract}
  In this work we develop statistical methodology to estimate and perform inference on subgraph densities using time-indexed, or dynamic network sequences.
  These estimates explicitly adjust for observation errors for the network edges, and have good theoretical properties as the size of the network grows.
  By specifying a stochastically evolving hidden Markov network model, we address two important directions for further investigation identified by \citetbody{chang22estimation}: robustness to non-identical network replicates, and efficient aggregation of multiple available network snapshots.
  These new methods vastly expand the analysis of noisy networks to new data settings, as network replicates are commonly observed dynamically.
  The methodology is also extended to consider joint inference for subgraph densities at multiple time points, to facilitate formal statistical comparison of dynamic network snapshots.
\end{abstract}

\noindent
{\bf Keywords:} measurement error; network data; subgraph inference; temporal networks

\pagebreak

\section{Introduction} \label{sec:intro}

Network data has become prevalent in many statistical applications, including neuroscience, economics, and sociology.
However, network data is not always made up of ground truth edges, but instead proxies constructed from other data sources.
This can lead to non-reporting biases, linkage errors, missingness, and other forms of measurement error on network edges, resulting in what we refer to as {\em network noise} (and a {\em noisy network} as the resulting data object).
Unaccounted-for network noise can lead to systematic bias in downstream inference tasks.
{\em In this work our focus is on noise-induced bias in network summaries, specifically subgraph densities.} 
But network noise can affect other tasks like link prediction, or node ranking.
In past work, \citebody{priebe15statistical} studied the effect of measurement error on estimation of the parameters of a stochastic block model.
\citebody{balachandran17propagation} studied the exact distributions of subgraph count discrepancies between ground truth and noisy networks.
\citebody{li2022estimation} studied estimation of branching factors.
Finally, \citebody{chang22estimation} considered estimation and inference for subgraph densities under network noise.
Their fundamental model has also been used by \citebody{jiang23autoregressive} to model autoregressive dynamic networks, and as a mechanism to introduce noise in network differential privacy \citepbody{chang24edge}.

Non-identifiability is a fundamental challenge of this noisy network model. \citebody{chang22estimation}, Theorem 1 states that even with constant but unknown type I and type II error rates, one (noisy) observation of the network is insufficient to identify the subgraph densities of the underlying network. As a result, they develop an approach that requires
three independent and identically distributed (iid) replicates. In some real-world settings, it is reasonable to expect network replicates.
For instance gene co-expression networks are typically built by aggregating replicated samples.
These samples can instead be aggregated into several smaller groups, which can be reasonably assumed to be iid, at the cost of statistical power to detect edges.
In neuroimaging, networks may similarly be constructed by pooling several subjects from a larger sample, although this can cause issues if the underlying brain networks are heterogeneous.
However, in social network analysis or in neuroimaging settings where we often want to infer about one social group or one individual, we may intrinsically have a sample size of one.

Our work is based on the following simple but powerful premise: {\em in the subgraph density estimation problem, dynamic network data can provide sufficient replicability for identification.} This opens up a vastly larger realm of applications. For example, we frequently can observe replicates of a social group over several days, or of a brain network longitudinally across several sessions in a scanner.
Subgraph density estimators can be used to globally summarize the network sequence, but exploiting the above premise is nontrivial. While the dynamic setting may lead to cases with more ``replicates,'' perhaps many more than three, these are unlikely to be exact replicates; both network dependence and heterogeneity across 
such ``replicates'' requires
development of appropriate, novel methodology for inference, as we do here.

In addition to improved, robust methods for the usual subgraph density estimation problem, new and critical inferential questions are suggested when network replicates are not exchangeable, and when the data is collected with dynamics in mind. 
For instance, in Section~\ref{sec:real_data}, we consider neuroimaging data originally collected to understand changes in brain activation from the beginning to end of the study period.

Dynamic network analysis has been a popular area of study in statistical network analysis and related fields, especially over the last decade;
see for instance the reviews by \citebody{kim2018review} and \citebody{kazemi2020representation}.
Classical statistical models for static networks are generalized to dynamic networks (observed in discrete time) in \citebody{krivitsky2014}, \citebody{sewell2015latent}, \citebody{matias2017}, and \citebody{jiang23autoregressive}, among others.
All of the examples above are fundamentally based on an {\em evolutionary} mechanism for binary edges, to efficiently parameterize the distributions of the dynamically observed network edges, and induce dependence over time.
Our work represents the same critical foundation for sequences of noisy networks, and we focus on the fundamental phenomena introduced by the interplay between stochastic evolution and error models.
Our model still allows flexibility to fit real data though an arbitrary underlying network or networks (cf.~Sections~\ref{sec:model} and \ref{sec:comparison}).
In specific applications where additional heterogeneity or exogenous covariates are present, it is natural to define more complex stochastic evolution and error models, see for instance \citebody{chang2026autoregressive}, which generalizes the autoregressive network model of \citebody{jiang23autoregressive}. 

In this paper we make many contributions to the estimation and inference of subgraph densities based on dynamically observed network snapshots, represented by their binary adjacency matrices
\[
  \{\Ynetk{k}\}_{k=1}^K \subset \{0,1\}^{n \times n},
\]
where each $\Ynetk{k}$ is a noisy version of an underlying network $\Anetk{k}$.
We consider two main settings: for global summarization, inference for subgraph densities at a fixed time (without loss of generality $k=1$); and for dynamic comparison, joint inference for subgraph densities at two time points (without loss of generality $k=1$ and $k=K$). 
In Section~\ref{subsec:comparison_problems}, we discuss analogous methodology for more complicated settings.

We distinguish new methodology for two distinct types of unknown parameters: edge density, and higher-order subgraph densities.
We will see in Section~\ref{sec:ho} that inference in the higher-order setting can utilize a special form of bootstrap resampling.
We defer some details of higher-order subgraph density with unknown error and evolution rate parameters to the \ifarxiv
appendices.
\fi
The locations in the manuscript of these methodological contributions are summarized in Table~\ref{tab:roadmap}.
We evaluate our new methodology on synthetic and noisy network sequences in Sections~\ref{sec:simulation}, apply it to neuroimaging data in Section~\ref{sec:real_data}, and discuss conclusions and directions for future work in Section~\ref{sec:conclusion}.
Technical proofs and additional detailed derivations are provided in the 
\ifarxiv
appendices.
\fi

\begin{table}[ht]
  \begin{center}
\begin{tabular}{|l|c|c|}
\hline
\multicolumn{1}{|l|}{\begin{tabular}[c]{@{}l@{}}Target parameter\end{tabular}} &
\multicolumn{1}{l|}{\begin{tabular}[c]{@{}l@{}}Edge density\end{tabular}} &
\multicolumn{1}{l|}{\begin{tabular}[c]{@{}l@{}}Higher-order subgraph density \\ (additional details) \end{tabular}} \\ \hline
$\Anetk{1}$ &
Sec.~\ref{sec:estimation} &
Sec.~\ref{sec:ho} (\ifarxiv App.\fi~\ref{app:ho_theory}) \\ \hline
$\Anetk{1}, \Anetk{K}$ &
Sec.~\ref{subsec:gmm_comparison} &
Sec.~\ref{subsec:ho_comparison} (\ifarxiv App.\fi~\ref{app:comparison_theory}) \\ \hline
\end{tabular}
\end{center}
\caption{Road map of methodology developed in this manuscript. \label{tab:roadmap}}
\end{table}

\section{Noisy dynamic network model} \label{sec:model}

For a single binary, undirected network on $n$ nodes, the noisy network model of \citebody{chang22estimation} specifies the distribution of a noisy observation $\Ynet \in \{0,1\}^{n \times n}$ of an underlying (deterministic) network $\Anet \in \{0,1\}^{n \times n}$, represented by their adjacency matrices.
They assume a two-parameter error model (independent bit flipping) for each edge, where
\begin{align*}
  \mathbb{P}(\Ynet_{ij}=1 \vert \Anet_{ij}=0) = \alpha, \quad 
  \mathbb{P}(\Ynet_{ij}=0 \vert \Anet_{ij}=1) &= \beta, \quad \alpha, \beta \in (0,1),
\end{align*}
independently over $i < j$, describing type I and type II edge error rates.
Model identifiability requires at least three iid replicates denoted by $(\Ynet,\Ynet^*,\Ynet^{**})$, where $\Ynet^* \in \{0,1\}^{n \times n}$ and $\Ynet^{**} \in \{0,1\}^{n \times n}$ are independent and have the same distribution as $\Ynet$ conditional on $\Anet$.

The development in \citebody{chang22estimation} assumes the three replicates are noisy observations of the exact same underlying network $\Anet$, although their results will remain valid if instead $\Ynet^*$ is a noisy observation of $\Anet^*$, and $\Ynet^{**}$ is a noisy observation of $\Anet^{**}$, with
\begin{equation} \label{c22_cond}
  \max\left\{ \binom{n}{2}^{-1} \sum_{i < j} \lvert \Anet_{ij} - \Anet^*_{ij} \rvert, \binom{n}{2}^{-1} \sum_{i < j} \lvert \Anet_{ij} - \Anet^{**}_{ij} \rvert \right\} = o\left( \frac{1}{n} \right).
\end{equation}
In this work, we are interested in relaxing this assumption to tap the rich discrete time dynamic network setting.

Assume now that we observe an ordered sequence of noisy network snapshots $\{\Ynetk{k}\}_{k=1}^K$.
Each snapshot represents a network on a common set of $n$ nodes by a binary, symmetric adjacency matrix, with no self loops.
For $k=1,\ldots,K$, suppose that $\Ynetk{k}$ is a noisy observation of an underlying network $\Anetk{k}$ through an {\em error} model,
\begin{align*}
  \mathbb{P}(\Ynetk{k}_{ij}=1 \vert \Anetk{k}_{ij}=0) = \alpha, \quad
  \mathbb{P}(\Ynetk{k}_{ij}=0 \vert \Anetk{k}_{ij}=1) &= \beta, \quad \alpha, \beta \in (0,1),
\end{align*}
independently over $i < j$, and $k=1,\ldots,K$.
Furthermore, assume a stochastic {\em evolution} model on the sequence of underlying networks $\{\Anetk{k}\}_{k=1}^K$:
\begin{align*}
  \mathbb{P}(\Anetk{k+1}_{ij}=1 \vert \Anetk{k}_{ij}=0) = \lambda, \quad
  \mathbb{P}(\Anetk{k+1}_{ij}=0 \vert \Anetk{k}_{ij}=1) &= \mu, \quad \lambda, \mu \in (0,1),
\end{align*}
for $k=1,\ldots,K-1$. 
Under this evolution model, we will have 
$$
  \expect\left\{ \binom{n}{2}^{-1} \sum_{i < j} \lvert \Anetk{k}_{ij} - \Anetk{1}_{ij} \rvert ~\bigg\vert~ \Anetk{1} \right\} \geq \min\{\lambda,\mu\},
$$
for all $k > 1$, and we will not get the required rate in \eqref{c22_cond} for inference to remain valid under methods which assume iid replicates.

In summary, we define a simple but practical model for the joint distribution of the noisy network snapshots $\{\Ynetk{k}\}_{k=1}^K$, treating the initial underlying network $\Anetk{1}$ as a high-dimensional unknown parameter.
Notice that if we treated all underlying networks as unknown parameters, we would be back in the {\em impossibility} setting of \citebody{chang22estimation}, Theorem 1. By replacing an iid (or near-iid) assumption with a simple stochastic evolution model, inference of subgraph densities (as well as both error and evolution rates) again becomes possible. Natural modifications and extensions include treating additional underlying snapshots as model parameters (cf. Section~\ref{sec:comparison}); relaxing the edgewise or temporal independence; or allowing the unknown error and evolution parameter vector $\theta = (\alpha,\beta,\lambda,\mu) \in (0,1)^4$ to depend on $i$, $j$, and/or $k$.

\subsection{Model identifiability} \label{subsec:identifiability}

In Section~\ref{sec:estimation}, we will always be concerned with identifiability of the $5$-dimensional parameter of interest $(\delta_1,\alpha,\beta,\lambda,\mu)$,
where $\delta_1$ is the edge density of $\Anetk{1}$,
$$
  \delta_1 = \binom{n}{2}^{-1} \sum_{i < j} \Anetk{1}_{ij};
$$
the configuration of the individual network edges remain as nuisance parameters.
In Section~\ref{sec:ho}, we may additionally add a higher-order subgraph density as an unknown parameter.
For this reason, each entry of $\Anetk{1}$ may not be fully identified.
For edge density, the parameters will be invariant to arbitrary relabelling of the node pairs.
For higher-order subgraph densities, the parameters will be invariant to arbitrary relabelling of the nodes, and possible restricted relabellings of the edges, although these will depend on the subgraph density of interest, and the network structure.

We require further restrictions on the parameter space to identify the error and evolution rate parameters.
First, note that $\alpha + \beta = 1$ will give a degenerate model in which all observed edges are identically distributed, regardless of their underlying status.
Moreover, as noted by \citebody{zhu23distinguishing}, for any $(\delta_1,\alpha,\beta,\lambda,\mu)$, there is a ``mirrored'' model which produces equivalent marginal probabilities of observed edges with parameter $(1-\delta_1,1-\alpha,1-\beta,\mu,\lambda)$,
which swaps the roles of the underlying edge labels but then reverses the labelling in the error process.
Based on these observations, we restrict $\alpha + \beta < 1$.

A similar phenomenon occurs when $\lambda + \mu = 1$; in this case the evolution process on edges has no memory, and additional snapshots provide no information about the initial status of underlying edges.
When $\lambda + \mu > 1$, edges are more likely to change status than maintain status in consecutive snapshots, leading to alternating behavior.
To preserve convexity of the parameter space, we restrict $\lambda + \mu < 1$.
This will produce models with realistic non-alternating evolution, which can be modelled coherently at different discrete time scales.

In the remainder of the paper, we define our restricted parameter space by
$$
  \Psi = \{(\delta_1,\alpha,\beta,\lambda,\mu) \in (0,1)^5 : \alpha + \beta < 1,~\lambda + \mu < 1\}.
$$
For technical reasons, we may also consider the compact restriction of this parameter space
$$
  \bar{\Psi} = \bar{\Psi}(\xi) = \{(\delta_1,\alpha,\beta,\lambda,\mu) \in [\xi,1-\xi]^5 : \alpha + \beta \leq 1 - \xi,~\lambda+\mu \leq 1 - \xi\}.
$$
for a small constant $\xi > 0$.

In addition to these restrictions on the parameter space, we require a sufficient number of snapshots $K$ to perform inference on the unknown model parameters.
In Section~\ref{sec:estimation}, we find that in order to estimate and do inference on $(\delta_1,\alpha,\beta,\lambda,\mu)$, we require $K \geq 3$.
In Section~\ref{subsec:gmm_comparison}, we find that in order to estimate and do inference under a dynamic comparison model, we require $K \geq 5$.

\section{Generalized method of moments estimation of edge density} \label{sec:estimation}

In this section, we detail a generalized method of moments (GMM) approach to estimation of edge density, and the nuisance parameter $\theta = (\alpha,\beta,\lambda,\mu)$ which governs the error and evolution of the noisy network.
In general, our adoption of moment-based estimation is based on the following observation for so-called {\em edge-averaged} moments.
For a mapping $m: \{0,1\}^K \rightarrow \realp$, define the corresponding edge-averaged empirical moment
$$
  \hat{\bm{m}} = \binom{n}{2}^{-1} \sum_{i < j} m(\Yvec_{ij}),
$$
where in general we use the notation $\vec{\cdot}$ to denote the length $K$-sequence observed over time, e.g. $\Yvec_{ij} = (\Ynetk{1}_{ij},\ldots,\Ynetk{K}_{ij})^{\tp} \in \{0,1\}^K$.
Also denote the edge-averaged expectation
\begin{equation*} \label{edge_averaged}
  \bm{m}(\delta_1,\theta) = \expect_{\theta}( \hat{\bm{m}} ~\vert~ \Anetk{1})
  = \delta_1 \bm{m}_1(\theta) + (1 - \delta_1) \bm{m}_0(\theta),
\end{equation*}
where
$$
  \bm{m}_s(\theta) = \expect_{\theta} \left\{ m\left( \Yvec_{12} \right) ~\big\vert~ \Anetk{1}_{12}=s \right\} \in \realp,
$$
and the node pair $(1,2)$ is chosen without loss of generality.
Moment-based estimation will seek estimators $\hat{\delta}_1$ and $\hat{\theta}$ such that $\hat{\bm{m}} \approx \bm{m}(\hat{\delta}_1,\hat{\theta})$.
For efficient evaluation without estimating the individual entries of $\Anetk{1}$, note that although $\Anetk{1}$ is an $O(n^2)$-dimensional unknown parameter, $\bm{m}(\theta,\delta_1)$ depends on it only through $\delta_1$, for any choice of moments, and can be evaluated based on the joint probabilities from a finite state space hidden Markov model \citepbody{jaeger00observable}.
 
Following the classical GMM approach to parametric estimation \citepbody{newey94large,hall03generalized}, we define the $5$-dimensional estimator $(\hat{\delta}_1,\hat{\theta})$ as the minimizer of
\begin{equation} \label{gmm_objective}
  f(\delta_1,\theta) = \left\{ \hat{\bm{m}} - \bm{m}(\delta_1,\theta) \right\}^{\tp} \bm{W} \left\{ \hat{\bm{m}} - \bm{m}(\delta_1,\theta) \right\} 
\end{equation}
where $\bm{W}$ is a positive-definite weight matrix.
Note that
\begin{align*} \label{cov_edge_averaged}
  \cov_{\theta}( \hat{\bm{m}} ~\vert~ \Anetk{1}) 
  &= \binom{n}{2}^{-1} \left\{ \delta_1 \bm{\Sigma}_1(\theta) + (1 - \delta_1) \bm{\Sigma}_0(\theta) \right\},
\end{align*}
where
\begin{equation} \label{sigma_s}
  \bm{\Sigma}_s(\theta) = \cov_{\theta} \left\{ m\left( \Yvec_{12} \right) ~\big\vert~ \Anetk{1}_{12}=s \right\} \in \real^{p \times p},
\end{equation}
for $s \in \{0,1\}$, facilitating easy evaluation of covariance for arbitararily large $n$, again without estimating the individual entries of $\Anetk{1}$.
To implement GMM estimation, all that remains is to specify $m$.

\subsection{Estimation of edge density}
\label{subsec:theta_unknown}

In order to estimate $\delta_1$ and $\theta$ jointly, the univariate components of $m$ will come from one of two types.
First the {\em local densities}, defined as $D_k(y_1,\ldots,y_K) = y_k$ 
for $k=1,\ldots,K$, and second {\em time-averaged triples}
$$
  T_{e_0e_1e_2}(y_1,\ldots,y_K) = \sum_{k=3}^K \mathbb{I}\left( y_k = e_0, y_{k-1} = e_1, y_{k-2} = e_2 \right)
$$
for $(e_0,e_1,e_2) \in \{0,1\}^3$, which are well-defined for $K \geq 3$.
The three binary indices $e_0$, $e_1$, and $e_2$ encode the value of an edge at lag $0$, $1$, and $2$ respectively; we require at least two lags to identify all the model parameters.
The frequencies of these edge subsequences provide information to disentangle the error and evolution parameters in the noisy dynamic model.
Similar relative frequencies of pairs and triples have been used to estimate the transition and error distributions of general finite-state hidden Markov models \citepbody{hsu12spectral}.

Our final estimation routine is based on a subset of these local densities and time-averaged triples, with some dropped due to linear dependencies:
$$
  m^* = (D_1,\ldots,D_{K-3},T_{000},T_{001},T_{010},T_{100},T_{011},T_{101},T_{110}) : \{0,1\}^{K} \rightarrow \real^{K+4}.
$$
The empirical edge-averaged moments in $m^*$ and their population counterparts are well-defined for $K \geq 3$.
Thus, even though we introduce two additional evolution parameters relative to the exact replicate model in \citebody{chang22estimation}, we still require only $3$ (or more) snapshots to estimate the unknown parameters.

In order to estimate the unknown parameters efficiently, we implement a modification of a usual iterative approch to adaptively select an efficient $\bm{W}$ in \eqref{gmm_objective}.
To initialize, we first evaluate
$$
  (\hat{\delta}_1^{(\mathrm{init})},\hat{\theta}^{(\mathrm{init})}) = \operatorname{argmin}_{(\delta_1,\theta) \in \Psi} \left\{ \left\lVert \hat{\bm{m}}^{(\mathrm{init})} -\bm{m}^{(\mathrm{init})}(\delta_1,\theta) \right\rVert_2^2 \right\}
$$
where $\hat{\bm{m}}^{(\mathrm{init})}$ and $\bm{m}^{(\mathrm{init})}(\delta_1,\theta)$ are the empirical and theoretical values for the edge-averaged moments
$$
  m^{(\mathrm{init})} = (D_1,T_{000},T_{001},T_{010},T_{100},T_{011},T_{101},T_{110}) : \{0,1\}^{K} \rightarrow \real^8.
$$
for $K \geq 4$. If $K=3$ we drop $D_1$ to avoid linear dependence in the moments.
We then use the initial estimate to choose an efficient weighting matrix in \eqref{gmm_objective} using the edge-averaged moments $m^*$.
In particular, we set
\begin{equation} \label{W_adaptive}
  \widehat{\bm{W}}^{-1} = \hat{\delta}_1^{(\mathrm{init})} \bm{\Sigma}^*_1(\hat{\theta}^{(\mathrm{init})}) + \left( 1 - \hat{\delta}_1^{(\mathrm{init})} \right) \bm{\Sigma}^*_0(\hat{\theta}^{(\mathrm{init})}),
\end{equation}
where $\bm{\Sigma}^*_s(\theta)$ is the covariance matrix of $m^*(\Yvec_{ij})$ conditional on $\Anetk{1}_{ij}=s$.
More iterations can be used, but one is sufficient to prove asymptotic efficiency.

Implementation of this estimation procedure requires evaluation of
$$
  \bm{m}^{(\mathrm{init})}(\delta_1,\theta) \in \real^8, \quad \bm{m}^*(\delta_1,\theta) \in \real^{K+4}, \quad \bm{\Sigma}^*_s(\theta) \in \real^{(K+4) \times (K+4)}
$$
for arbitrary $\delta_1$, $\theta$ and $s=0,1$. Additional details are provided in  
\ifarxiv
Appendix~\ref{app:gmm_computation}.
\fi

Towards proving consistency and asymptotic normality of our GMM estimators, we assume that both $m^*$ and $m^{(\mathrm{init})}$ contain enough information to identify all the model parameters.
\begin{assumption} \label{assump:global_identification}
The mapping $\bm{m}^{(\mathrm{init})}(\delta_1,\theta)$ is an injective function on $\Psi$.
\end{assumption}
Since the coordinate functions in $m^{(\mathrm{init})}$ are also contained in $m^*$, Assumption~\ref{assump:global_identification} implies that $\bm{m}^*(\delta_1,\theta)$ is also injective.
It also implies that the unknown parameters $(\delta_1,\theta)$ are generally identifiable.
In order to prove asymptotic normality and asymptotic efficiency of our GMM estimator, we require three more technical assumptions.
In this setting, we work in the compact parameter space $\bar{\Psi} = \bar{\Psi}(\xi)$ defined in Section~\ref{subsec:identifiability}, for a small constant $\xi > 0$.

\begin{assumption} \label{assump:delta_seq_compact}
  The true parameters satisfy $(\delta_1(n),\alpha,\beta,\lambda,\mu) \in \bar{\Psi}$ for all $n$, and the sequence converges, that is $\delta_1(n) \rightarrow \delta_1(\infty) \in (0,1)$.
\end{assumption}
Throughout, we will explicitly note when the limiting value $\delta_1(\infty)$ is referred to.
When $\delta_1$ is written, it should be interpreted as $\delta_1(n)$ for the appropriate choice of $n$.

\begin{assumption} \label{assump:local_identification_body}
  The derivative matrix
  $$
    \bm{Dm}^*(\delta_1,\theta) = \begin{pmatrix}
      \pd{\bm{m}^*}{\delta_1} & \pd{\bm{m}^*}{\theta}
  \end{pmatrix} \in \real^{(K+4) \times 5}
  $$
  has full column rank for all $(\delta_1,\theta) \in \bar{\Psi}$.
\end{assumption}

\begin{assumption} \label{assump:cov_spectrum}
  The covariance matrix $\bm{\Sigma}^*(\delta_1,\theta) = \delta_1 \bm{\Sigma}_1^*(\theta) + (1-\delta_1) \bm{\Sigma}_0^*(\theta)$
  satisfies
  $
    0 < c \leq \lambda_{\min}\left\{ \bm{\Sigma}^*(\delta_1,\theta) \right\} \leq \lambda_{\max}\left\{ \bm{\Sigma}^*(\delta_1,\theta) \right\} \leq C < \infty
  $
  uniformly over all $(\delta_1,\theta) \in \bar{\Psi}$ for constants $c$ and $C$ (which may depend on $\xi$).
\end{assumption}

Assumption~\ref{assump:local_identification_body} assumes that the parameters are locally identified, and thus that the estimating equations have non-degenerate covariance, and Assumption~\ref{assump:cov_spectrum} controls the conditioning of the adaptive weighting matrices $\widehat{\bm{W}}$ as $n \rightarrow \infty$.

\begin{remark} \label{rem:identifiability}
  Although we do not show analytically that Assumptions~\ref{assump:global_identification}, \ref{assump:local_identification_body}, and \ref{assump:cov_spectrum} hold for our model, we can verify these assumptions numerically for $K \geq 3$.
  For Assumptions~\ref{assump:global_identification} and \ref{assump:local_identification_body}, we verify that $\bm{Dm}^{(\mathrm{init})}(\delta_1,\theta)^{\tp} \bm{Dm}^{(\mathrm{init})}(\delta_1,\theta)$
  and $\bm{Dm}^*(\delta_1,\theta)^{\tp} \bm{Dm}^*(\delta_1,\theta)$ are invertible,
  except at or very near the boundary of the full parameter space $\Psi$.
  Note that under regularity conditions which hold for our noisy dynamic network model, local identification is a sufficient condition for global identification.
  Hence it is reasonable to assume the assumptions hold on the compact restricted parameter space $\bar{\Psi}(\xi)$.
  A similar numerical verification is used to check Assumption~\ref{assump:cov_spectrum}: $\bm{\Sigma}^*(\delta_1,\theta)$ is invertible away from the boundary of $\Psi$ for $K \geq 3$.
\end{remark}

We are now ready to state the main result of this section, the asymptotic normality and efficiency of our adaptive GMM estimator.

\begin{proposition} \label{prop:gmm_adaptive_clt}
  Suppose Assumptions~\ref{assump:global_identification}---\ref{assump:cov_spectrum} hold.
  Define estimators
  $$
    (\hat{\delta}_1,\hat{\theta}) = \operatorname{argmin}_{(\delta_1,\theta) \in \bar{\Psi}} \left\{ \hat{\bm{m}}^* - \bm{m}^*(\delta_1,\theta) \right\}^{\tp} \widehat{\bm{W}} \left\{ \hat{\bm{m}}^* - \bm{m}^*(\delta_1,\theta) \right\},
  $$
  where $\widehat{\bm{W}}$ is defined in \eqref{W_adaptive}.
  Then
  $$
    \binom{n}{2}^{1/2} \left\{ \begin{pmatrix}
      \hat{\delta}_1 \\ \hat{\theta}
    \end{pmatrix} - \begin{pmatrix}
      \delta_1 \\ \theta
    \end{pmatrix} \right\} \indist \mathcal{N}\left(\bm{0} , \bm{\Sigma}^{(\mathrm{GMM})}(\delta_1(\infty),\theta) \right),
  $$
  where
  $$
    \bm{\Sigma}^{(\mathrm{GMM})}(\delta_1(\infty),\theta) = \left\{ \bm{Dm}^*(\delta_1(\infty),\theta)^{\tp} \left\{ \bm{\Sigma}^*(\delta_1(\infty),\theta) \right\}^{-1} \bm{Dm}^*(\delta_1(\infty),\theta) \right\}^{-1}.
  $$
\end{proposition}
This estimator achieves the smallest possible asymptotic covariance among GMM estimators using the moments in $m^*$.

\begin{remark} \label{rem:thetaknown}

After estimating $\delta_1$ and $\theta$ using the GMM approach of Section~\ref{subsec:theta_unknown}, the methodology described in
\ifarxiv
Appendix~\ref{subsec:theta_known}
\fi
(assuming known $\theta$) can be used to find an alternate estimator of the form
$\bm{w}_*(\hat{\delta}_1,\hat{\theta})^{\tp} \hat{\bm{\delta}}_{1}$, where
$$
    \hat{\bm{\delta}}_{1} = \begin{pmatrix}
        \hat{\delta}_1^{(1)} & \cdots & \hat{\delta}_1^{(K)} 
    \end{pmatrix}^{\tp}, \quad \hat{\delta}_1^{(k)} = \binom{n}{2}^{-1} \sum_{i < j} \frac{\Ynetk{k}_{ij} - \expect_{\hat\theta}(\Ynetk{k}_{ij} ~\vert~ \Anetk{1}_{ij}=0)}{\expect_{\hat\theta}(\Ynetk{k}_{ij} ~\vert~ \Anetk{1}_{ij}=1) - \expect_{\hat\theta}(\Ynetk{k}_{ij} ~\vert~ \Anetk{1}_{ij}=0)}
$$
is a vector of snapshot-specific adjusted edge densities, and $\bm{w}_*(\hat{\delta}_1,\hat{\theta}) \in \real^K$ is a non-negative weight vector which sums to $1$, chosen to minimize the estimator's asymptotic variance.
We find that this alternate estimator is empirically nearly identical to the one found using two-stage GMM. Moreover $\bm{w}_*(\hat{\delta}_1,\hat{\theta})$ is practically useful to interpret the contribution of each network snapshot to the final estimate of $\delta_1$, as in Section~\ref{sec:simulation}.
\end{remark}

\begin{remark} \label{rem:mle}
The likelihood function for $\Anetk{1}$, assuming known $\theta$, is given by 
\begin{equation*}
  L(\Anetk{1} ; \Ynetk{1}, \ldots, \Ynetk{K}, \theta) = \prod_{i < j} \mathbb{P}_{\theta,\Anetk{1}_{ij}}\left( \Ynetk{1}_{ij}, \ldots , \Ynetk{K}_{ij} \right).
\end{equation*}
Thus, for each node pair $(i,j)$, the maximum likelihood estimator (MLE) of $\Anetk{1}_{ij}$ is given by
\begin{equation}
  \hat{\bm{A}}^{(1)}_{ij} = \operatorname{argmax}_{s \in \{0,1\}} \mathbb{P}_{\theta,s}\left( \Ynetk{1}_{ij}, \ldots , \Ynetk{K}_{ij} \right),
\end{equation}
and by MLE invariance,
\begin{equation}
  \hat{\delta}^{(\mathrm{MLE})}_1 = \binom{n}{2} \sum_{i < j} \hat{\bm{A}}^{(1)}_{ij}.
\end{equation}

The bias of the MLE can be calulated exactly based on the conditional distributions of $\{0,1\}$ sequences as a function of $K$, $\delta_1$ and $\theta$.
Define
$$
  \mathcal{A} = \{ (y_1,\ldots,y_K) \in \{0,1\}^K : \mathbb{P}_{\theta,1}( y_1,\ldots,y_K) > \mathbb{P}_{\theta,0}( y_1,\ldots,y_K) \},
$$
the set of binary length-$K$ edge sequences for which the most likely initial state is $1$.
Then
\begin{equation} \label{mle_bias}
    \expect\left( \hat{\delta}^{(\mathrm{MLE})}_1 \right) = \delta_1 \prob_{\theta,1}(\mathcal{A}) + (1 - \delta_1) \prob_{\theta,0}(\mathcal{A}).
\end{equation}
Note that \eqref{mle_bias} is constant in $n$, and implies that the MLE will only be unbiased (even asymptotically) if $\prob_{\theta,1}(\mathcal{A})=1$ and $\prob_{\theta,0}(\mathcal{A})=0$.
That is, only if it is possible to perfectly classify the initial state of each noisy edge sequence.
\end{remark}

\section{Higher-order subgraph density estimation} \label{sec:ho}

In this section, we introduce new notation to describe network summaries based on arbitrary subgraph densities.
For clarity of presentation we assume that the error and evolution parameters in $\theta$ are known.
In practical applications, we replace the true parameter $\theta$ by the GMM estimator $\hat{\theta}$ derived in Section~\ref{subsec:theta_unknown}; all technical proofs and additional details of inference are developed in 
\ifarxiv
Appendix~\ref{app:ho_theta_unknown}.
\fi

We wish to estimate the density of an arbitrary subgraph $H$ of fixed order $\lvert V_H \rvert \geq 2$.
Let $\mathcal{V}$ denote the set of node pairs that can map onto the prespecified edges of $V_H$, and $\tau_1,\ldots,\tau_L \in \{0,1\}$ denote the prescribed values of those edges.
Note that a given $\bm{v} \in \mathcal{V}$ will contain $L$ distinct node pairs, $\bm{v}=(v_1,\ldots,v_L)$, and although some individual node indices may overlap, the pairs will be distinct.

The subgraph $H$ density in the first underlying network can be written as
\begin{equation} \label{ho_density}
  C_H = C_H(\bm{A}^{(1)}) = \frac{1}{\lvert \mathcal{V} \rvert} \sum_{\bm{v} \in \mathcal{V}} \prod_{\ell=1}^L (\bm{A}^{(1)}_{v_{\ell}})^{\tau_{\ell}}(1 - \bm{A}_{v_{\ell}}^{(1)})^{1 - \tau_{\ell}}.
\end{equation}
Note that $v_{\ell}$ is a node pair, so we use the subscript $v_{\ell}$ to index the edge variable which connects said pair.
Define $x_k(\theta)$ and $y_k(\theta)$ by
\[
  \mathbb{E}_{\theta}\left( \bm{Y}_{ij}^{(k)} ~\vert~ \Anetk{1} \right) = \begin{cases}
    x_k(\theta), \quad &\bm{A}_{ij}^{(1)} = 0, \\
    y_k(\theta), \quad &\bm{A}_{ij}^{(1)} = 1.
\end{cases}
\]
Rearranging, we have
\[
  \bm{A}_{ij}^{(1)} = \frac{\mathbb{E}_{\theta}\left( \bm{Y}_{ij}^{(k)} \right) - x_k(\theta)}{y_k(\theta) - x_k(\theta)}.
\]
Thus \eqref{ho_density} can be rewritten as
\begin{equation} \label{ho_density_expectation}
  C_H =  \frac{1}{\lvert \mathcal{V} \rvert} \cdot \frac{1}{\{y_k(\theta) - x_k(\theta)\}^L} \cdot \sum_{\bm{v} \in \mathcal{V}} \prod_{\ell=1}^L \mathbb{E}_{\theta} \left\{ \phi^{(k)}_{\ell,\theta}(\bm{Y}^{(k)}_{v_{\ell}}) \right\},
\end{equation}
where
\[
  \phi^{(k)}_{\ell,\theta}(z) = \{z - x_k(\theta)\}^{\tau_{\ell}}\{y_k(\theta) - z\}^{1-\tau_{\ell}}.
\]
Note that \eqref{ho_density_expectation} provides $K$ different representations of $C_H$, each depending on adjusted edge expectations of snapshot $k$ for $k=1,\ldots,K$.

Define a family of unbiased estimators of $C_H$,
\begin{equation} \label{ho_density_estk}
  \widetilde{C}_H^{(k)} = \frac{1}{\lvert \mathcal{V} \rvert} \cdot \frac{1}{\{y_k(\theta) - x_k(\theta)\}^L} \cdot \sum_{\bm{v} \in \mathcal{V}} \prod_{\ell=1}^L \phi^{(k)}_{\ell,\theta}(\bm{Y}^{(k)}_{v_{\ell}})
\end{equation}
for $k=1,\ldots,K$.
Define
\[
  \widetilde{\bm{C}}_H = (\widetilde{C}_H^{(1)}, \ldots , \widetilde{C}_H^{(K)})^{\tp}.
\]
Our final estimator of $C_H$ will be $\bm{w}^{\tp} \widetilde{\bm{C}}_H$ for an adaptive weight vector $\bm{w} \in \real^K$ satisfying $\bm{w}^{\tp} \bm{1}_K = 1$.
Following \citebody{lavancier16general}, the efficient weights are given by
\begin{equation} \label{optwt}
  \bm{w}_*(\theta) \propto \operatorname{Cov}_{\theta}^{-1}(\widetilde{\bm{C}}_H)\bm{1}_K,
\end{equation}
normalized to sum to $1$.
Towards asymptotically valid and efficient inference, we generalize Propositions 2 and 3 in \citebody{chang22estimation} to our dynamic setting as Corollaries~\ref{cor:ho_order} and~\ref{cor:ho_S_approx}, respectively.
Their precise statements are given in 
\ifarxiv
Appendix~\ref{app:ho_theory},
\fi
and their proofs are analogous to those in \citebody{chang22estimation}.

Corollary~\ref{cor:ho_S_approx} relates the asymptotic behavior of each $\widetilde{C}_H^{(k)}$ to that of a new centered and linearized quantity
\begin{equation} \label{SH_linearized}
  S_H^{(k)} = \binom{n}{2}^{1/2} \cdot \frac{1}{\lvert \mathcal{V} \rvert} \cdot \frac{1}{\{y_k(\theta) - x_k(\theta)\}^L} \cdot \sum_{j=1}^L (-1)^{1-\tau_j} \sum_{\bm{v} \in \mathcal{V}} \left\{\bm{Y}^{(k)}_{v_j} - \mathbb{E}_\theta \bm{Y}^{(k)}_{v_j} \right\} \prod_{\ell \neq j} \mathbb{E}_{\theta} \left\{ \phi_{\ell,\theta}(\bm{Y}^{(k)}_{v_{\ell}}) \right\}.
\end{equation}
Defining
\[
  \bm{S}_H = (S_H^{(1)}, \ldots , S_H^{(K)})^{\tp},
\]
we can write
\[
  \bm{S}_H = \binom{n}{2}^{-1/2} \sum_{i < j} \bm{M}_{ij} \left( \vec{\bm{Y}}_{ij} - \mathbb{E}_{\theta} \vec{\bm{Y}}_{ij} \right)
\]
for a collection of diagonal, non-random $K \times K$ matrices $\bm{M}_{ij}$.
Moreover, $\bm{S}_H$ satisfies 
\begin{equation} \label{ho_S_approx}
  \binom{n}{2}^{1/2} (\widetilde{\bm{C}}_H - C_H\bm{1}_K) = \bm{S}_H + o_{\mathbb{P}}(1).
\end{equation}
Thus, the asymptotic behavior of $\widetilde{\bm{C}}_H$ can be understood through the linearized and centered $\bm{S}_H$, which has asymptotic covariance
\[
  \lim_{n \rightarrow \infty} \binom{n}{2}^{-1} \sum_{i < j} \bm{M}_{ij}\operatorname{Cov}(  \vec{\bm{Y}}_{ij}) \bm{M}_{ij}.
\]

To estimate the covariance matrix of $\bm{S}_H$, we develop a bootstrap sampler of edge sequences which exactly captures the conditional covariance structure: That is, $\vec{\bm{Y}}_{ij}^{\dagger}$ is sampled conditional on $\vec{\bm{Y}}_{ij} \in \{0,1\}^K$ such that
\begin{equation} \label{ho_bootcov}
  \mathbb{E}_{\theta}\{\operatorname{Cov}(\vec{\bm{Y}}_{ij}^{\dagger} \vert \vec{\bm{Y}}_{ij})\} = \operatorname{Cov}_{\theta}(\vec{\bm{Y}}_{ij}) \in \mathbb{R}^{K \times K}.
\end{equation}
The left-hand side $\operatorname{Cov}_{\theta}(\vec{\bm{Y}}_{ij})$ is one of two matrices (depending on the status of $\Anetk{1}_{ij}$), which can be calculated given the unknown parameters in $\theta$.
The following result shows that it is sufficient to define $\mathcal{B}(\vec{\bm{Y}}_{ij})$ through a mixture of two normal distributions, depending on the status of $\Ynetk{1}_{ij}$.

\begin{proposition} \label{prop:boot_sampler}
    Without loss of generality, define $\bm{\Sigma}^{(s)} = \operatorname{Cov}(\Yvec_{12} ~\vert~ A_{12}=s)$ for $s=0,1$, and
\begin{equation*}
  \bm{\Sigma}^{\dagger}_{0} = \frac{(1-\beta)\bm{\Sigma}^{(0)} - \alpha \bm{\Sigma}^{(1)}}{1 - \alpha - \beta}, \quad
  \bm{\Sigma}^{\dagger}_{1} = \frac{(1-\alpha)\bm{\Sigma}^{(1)} - \beta \bm{\Sigma}^{(0)}}{1 - \alpha - \beta}.
\end{equation*}
\[
  \Yvec^{\dagger}_{ij} ~\vert~ \Yvec_{ij} \sim \mathcal{N} \left( \bm{0}_K,\bm{\Sigma}^{\dagger}_{\Ynetk{1}_{ij}} \right),
\]
independently for $i < j$.
Then $\{ \vec{\bm{Y}}_{ij}^{\dagger} \}_{i < j}$ satisfies the covariance constraint \eqref{ho_bootcov}.
\end{proposition}

\begin{remark} \label{rem:mvn}
    Although the original edge sequence satisfies $\Yvec_{ij} \in \{0,1\}^K$, for convenience of implementation we draw bootstrap replicates which are multivariate Gaussian with mean zero.
    For the purposes of Propositions~\ref{thm:ho_bootdist} and \ref{prop:adaptive_distn}, this difference in support does not matter as long as \eqref{ho_bootcov} is satisfied.
    In principle, one could instead compute bootstrap distributions supported on $\{0,1\}^K$, similar to \citebody{chang22estimation}.
\end{remark}

Using the bootstrap sampled edge sequences, we can define, for each $k=1,\ldots,K$,
\begin{equation} \label{ho_Sboot}
  S_H^{\dagger,(k)} = \binom{n}{2}^{1/2} \cdot \frac{1}{\lvert \mathcal{V} \rvert} \cdot \frac{1}{\{y_k(\theta) - x_k(\theta)\}^L} \cdot \sum_{j=1}^L (-1)^{1-\tau_j} \sum_{\bm{v} \in \mathcal{V}} \bm{Y}^{\dagger,(k)}_{v_j} \prod_{\ell \neq j} \phi_{\ell,\theta}(\bm{Y}^{(k)}_{v_{\ell}}),
\end{equation}
using $\bm{S}_H^{\dagger}$ to denote the vector in $\real^K$.
We theoretically justify the asymptotic bootstrap distribution with the following result, based on a Cramer-Wold type linear combination.
Assumption~\ref{assump:c22} and the quantity $\aleph_{\mathcal{V}}$ enforce regularity conditions on the configuration of $H$, and are defined precisely in
\ifarxiv
Appendix~\ref{app:ho_theory}.
\fi
\begin{proposition} \label{thm:ho_bootdist}
  Suppose Assumption~\ref{assump:c22} holds, $\aleph_{\mathcal{V}}/\lvert \mathcal{V} \rvert = O_{\mathbb{P}}\left( n^{-2} \right)$,
  and $y_k(\theta) \neq x_k(\theta)$ for all $k \in \{1,\ldots,K\}$.
  Let $\bm{a} \in \real^K$.
  Then
  \[
    \sup_{z \in \mathbb{R}} \left\lvert \mathbb{P}\left(\binom{n}{2}^{1/2} \bm{a}^{\tp} (\widetilde{\bm{C}}_H - C_H\bm{1}_K) > z \right) - \mathbb{P}\left( \bm{a}^{\tp} \bm{S}_H^{\dagger} > z ~\vert~ \bm{Y} \right) \right\rvert \rightarrow 0
  \]
  as $n \rightarrow \infty$.
\end{proposition}

The bootstrap samples can also be used to recover asymptotically optimal adaptive weights, given by
\[
    \bm{w}^{\dagger}(\bm{Y};\theta) = \frac{\operatorname{Cov}^{-1}(\bm{S}^{\dagger}_H ~\vert~ \bm{Y}) \bm{1}_K}{\bm{1}_K^{\tp} \operatorname{Cov}^{-1}(\bm{S}^{\dagger}_H ~\vert~ \bm{Y}) \bm{1}_K},
  \]
and to estimate the variance of the resulting unbiased estimator
\begin{equation} \label{chhat}
  \widehat{C}_H = \bm{w}^{\dagger}(\bm{Y};\theta)^{\tp} \widetilde{\bm{C}}_H.
\end{equation}
To prove this, we require the following regularity assumption on the covariance matrices of $\widetilde{\bm{C}}_H$.
\begin{assumption} \label{assump:ho_cov}
  The normalized covariance matrices
  \begin{equation} \label{sigma_C}
    \Sigma_C(n) = \operatorname{Cov}_{\theta}\left\{ \binom{n}{2}^{1/2}(\widetilde{\bm{C}}_H - C_H \bm{1}_K)\right\}
  \end{equation}
  satisfy
  \[
    0 < c \leq \lambda_{\min}\left\{ \Sigma_C(n) \right\} \leq \lambda_{\max}\left\{ \Sigma_C(n) \right\} \leq C < \infty
  \]
  uniformly over $n$ for constants $c$ and $C$.
\end{assumption}

\begin{proposition} \label{prop:adaptive_distn}
  Under the conditions of Proposition~\ref{thm:ho_bootdist}, and Assumption~\ref{assump:ho_cov}, suppose that
  \[
    \binom{n}{2}^{1/2} \left\{ \bm{w}_*(\theta)^{\tp} \widetilde{\bm{C}}_H - C_H \right\}
  \]
  has a limiting distribution.
  Then
  \begin{equation*}
    \binom{n}{2}^{1/2} \left( \widehat{C}_H - C_H\right)
  \end{equation*}
  has the same limiting distribution, where $\widehat{C}_H$ is defined in \eqref{chhat}.
\end{proposition}

In the prior result, we compare an adaptive estimator to the estimator with optimal weights $\bm{w}_*(\theta)$. Define the optimal asymptotic variance
$$
  \sigma_H^2 = \lim_{n \rightarrow \infty} \bm{w}_*(\theta)^{\tp} \Sigma_C(n) \bm{w}_*(\theta) = \lim_{n \rightarrow \infty} \left\{ \bm{1}_K^{\tp}  \Sigma^{-1}_C(n) \bm{1}_K \right\}^{-1},
$$
assuming that this limit exists.
Under these conditions, 
$$
  \left\{ \bm{1}_K^{\tp}  \cov_{\theta}^{-1}\left( \bm{S}_H^{\dagger} ~\vert~ \bm{Y} \right) \bm{1}_K \right\}^{-1} \inprob \sigma_H^2.
$$
The inverse covariance matrix on the left-hand side can be estimated to arbitrary precision for any $n$ by increasing the number of bootstrap replicates.

\section{Dynamic comparison of subgraph densities} \label{sec:comparison}

Applied network analysis in the dynamic setting often poses questions about network comparison.
Using a variant on the model developed in Section~\ref{sec:model}, we develop methodology to perform network comparison through joint estimation and inference for edge densities, and more general higher-order subgraph densities of the first and last snapshot in an observed sequence.

As in Section 2, we specify edgewise error and evolution behavior for the observed sequence network snapshots $\{\Ynetk{k}\}_{k=1}^K$ through
\begin{align*}
  \mathbb{P}(\Ynetk{k}_{ij}=1 \vert \Anetk{k}_{ij}=0) &= \alpha, \quad
  &\mathbb{P}(\Ynetk{k}_{ij}=0 \vert \Anetk{k}_{ij}=1) &= \beta, \\
  \mathbb{P}(\Anetk{k+1}_{ij}=1 \vert \Anetk{k}_{ij}=0) &= \lambda,
  &\mathbb{P}(\Anetk{k+1}_{ij}=0 \vert \Anetk{k}_{ij}=1) &= \mu
\end{align*}
for $k=1,\ldots,K$, independently over $i < j$, and governed by parameters $\theta = (\alpha,\beta,\lambda,\mu)$.
However, we now treat both the first underlying snapshot $\Anetk{1}$, and the final snapshot $\Anetk{K}$ as high-dimensional unknown parameters.
Adding this flexibility in terms of the final snapshot requires a modified model.
In Section~\ref{sec:model}, the joint distribution of the edge sequence for each vertex pair followed a hidden Markov model with deterministic initial state and homogeneous transition kernel.
Now, edge sequences still follow hidden Markov models, but with inhomogeneous bridge distributions depending on the final state, but computable from the original homogeneous transition distributions.
For instance, for $k \in \{2,\ldots,K-1\}$ and $s \in \{0,1\}$, we can expand the transition probability from $0$ to $1$ in terms of the original homogeneous transitions as
\begin{align*}
  \prob(\Anetk{k+1}_{ij}=1 ~\vert~ \Anetk{k}_{ij}=0,\Anetk{K}_{ij}=s) 
  = &\frac{\prob(\Anetk{K}_{ij}=s ~\vert~ \Anetk{k+1}_{ij}=1)\prob(\Anetk{k+1}_{ij}=1 ~\vert~ \Anetk{k}_{ij}=0)}{\prob(\Anetk{K}_{ij}=s ~\vert~ \Anetk{k}_{ij}=0)}.
\end{align*}
Analogous derivations can be used to compute other transition probabilities.
This inhomogeneous transition probability now depends on $s \in \{0,1\}$, the final state of the edge at time $K$, as well as $k$, which determines the number of remaining time steps before it must reach its final state.

In Section~\ref{subsec:gmm_comparison}, we extend the GMM methodology to jointly estimate error, evolution and starting and ending edge density parameters
$$
    \delta_1 = \binom{n}{2}^{-1} \sum_{i < j} \Anetk{1}_{ij}, \quad \delta_K = \binom{n}{2}^{-1} \sum_{i < j} \Anetk{K}_{ij}
$$
from the observed sequence of noisy, evolving networks.
In Section~\ref{subsec:ho_comparison} we provide an outline of the analogous methodology and bootstrapping algorithm for higher-order subgraph densities.
Finally, in Section~\ref{subsec:comparison_problems} we discuss generalizations to other more complex comparison problems.
All detailed results for this section are given in 
\ifarxiv
Appendix~\ref{app:comparison_theory}.
\fi

\subsection{GMM estimation for edge density comparison} \label{subsec:gmm_comparison}

Recall the notation of Section~\ref{sec:estimation}, where
$
  m: \{0,1\}^K \rightarrow \real^p
$
denotes a mapping on edge sequences.
Now, in order to summarize the behavior of all edge sequences, we must track the underlying status of each edge jointly in snapshots $1$ and $K$.
Thus we will aim to estimate three edge densities
\begin{align*}
  \rho_1 = \binom{n}{2}^{-1} \sum_{i < j} \Anetk{1}_{ij} (1 - \Anetk{K}_{ij}), ~~&~~
  \rho_K = \binom{n}{2}^{-1} \sum_{i < j} (1 - \Anetk{1}_{ij}) \Anetk{K}_{ij}, \\
  \rho_{1K} = \binom{n}{2}^{-1} &\sum_{i < j} \Anetk{1}_{ij} \Anetk{K}_{ij},
\end{align*}
so that our snapshot density parameters of interest can be written as $\delta_1 = \rho_1 + \rho_{1K}$ and $\delta_K = \rho_K + \rho_{1K}$.
The edge-averaged expectation is given by
$$
  \bm{m}(\rho_1,\rho_K,\rho_{1K},\theta) = \rho_{1K} \bm{m}_{11}(\theta) + \rho_1 \bm{m}_{10}(\theta) + \rho_K \bm{m}_{01}(\theta) + (1 - \rho_1 - \rho_K - \rho_{1K})\bm{m}_{00}(\theta),
$$
where
$$
  \bm{m}_{st}(\theta) = \expect_{\theta} \left\{ m\left( \Yvec_{12} \right) ~\big\vert~ \Anetk{1}_{12}=s,~\Anetk{K}_{12}=t \right\} \in \realp.
$$

We specify two mappings for GMM estimation, one for initialization and another for the second-stage reweighted estimator.
We can show empirically that these mappings lead to local and global identification (similar to Assumptions \ref{assump:global_identification} and \ref{assump:local_identification_body}), as well as covariance non-degeneracy (similar to Assumption~\ref{assump:cov_spectrum}) when $K \geq 5$.

For initialization with $K > 5$, we use
\[
  m_{\mathcal{C}}^{(\mathrm{init})} = (D_1,D_K,T_{000},T_{001},T_{010},T_{100},T_{011},T_{101},T_{110},P_K)
\]
where $D_k$ and $T_{e_1e_2e_3}$ are defined as in Section~\ref{sec:estimation}, and
\[
  P_K(y_1,\ldots,y_K) = y_1y_K
\]
is added to identify the joint densities for the starting and ending snapshots. For $K=5$, we drop the mapping $T_{110}$ from $m_{\mathcal{C}}^{(\mathrm{init})}$ to ensure linear independence of the moments.

For the second-stage reweighted estimator, we keep additional local density information, and specify
\[
  m_{\mathcal{C}}^* = (D_1,\ldots,D_{K_{\mathrm{cent}}-1},D_{K_{\mathrm{cent}}+1},\ldots,D_K,T_{000},T_{001},T_{010},T_{100},T_{101},P_K),
\]
where $K_{\mathrm{cent}}$ is a central index $\lceil (K+1)/2 \rceil$.
In 
\ifarxiv
Appendix~\ref{app:comparison_theory}.
\fi
of the supplementary materials, we sketch the analogous theoretical results and proofs to find an asymptotic normal distribution for the resulting GMM estimators.

\subsection{Higher-order subgraph density comparison} \label{subsec:ho_comparison}

In this section, for a subgraph $H$ of arbitrary order, we consider subgraph density estimators for
$$
  C_H^{(1)} = \frac{1}{\lvert \mathcal{V} \rvert} \sum_{\bm{v} \in \mathcal{V}} \prod_{\ell=1}^L (\Anetk{1}_{v_{\ell}})^{\tau_{\ell}}(1 - \Anetk{1}_{v_{\ell}})^{1 - \tau_{\ell}}, \quad C_H^{(K)} = \frac{1}{\lvert \mathcal{V} \rvert} \sum_{\bm{v} \in \mathcal{V}} \prod_{\ell=1}^L (\Anetk{K}_{v_{\ell}})^{\tau_{\ell}}(1 - \Anetk{K}_{v_{\ell}})^{1 - \tau_{\ell}}
$$ 
based on the initial and final network snapshots $\Ynetk{1}$ and $\Ynetk{K}$.
In principle, the methods in Section~\ref{sec:ho} can be extended to produce reweighted estimators based on all observed network snapshots $\{\Ynetk{k}\}_{k=1}^K$, but we do not develop those details here. 
Instead we estimate $C_H^{(k)}$ for $k=1, K$ using the corresponding $\Ynetk{k}$ only.

As in Section~\ref{sec:ho}, subgraph density estimation requires plug-in of adjusted surrogate edges.
Define a mapping
$$
    \phi_{\theta}(x) = \frac{x - \alpha}{1 - \beta - \alpha},
$$
and note that
\begin{align}
  \expect_{\theta}\left\{ \phi_{\theta}(\Ynetk{1}_{12}) ~\vert~ \Anetk{1}_{12}=s, \Anetk{K}_{12}=t \right\} &= s, \nonumber \\
  \expect_{\theta}\left\{ \phi_{\theta}(\Ynetk{K}_{12}) ~\vert~ \Anetk{1}_{12}=s, \Anetk{K}_{12}=t \right\} &= t, \label{comparison_phis}
\end{align}
for $s,t \in \{0,1\}$, since $\Ynetk{1}_{12}$ and $\Ynetk{K}_{12}$ are independent when $\Anetk{1}$ and $\Anetk{K}$ are viewed as unknown parameters.
Next, for $\ell=1,\ldots,L$, define
\[
  \phi_{\ell,\theta}(x) = \{\phi_{\theta}(x)\}^{\tau_{\ell}} \{1 - \phi_{\theta}(x)\}^{1-\tau_{\ell}}.
\]
It follows that
\[
  \widetilde{C}_{H,\mathcal{C}}^{(1)} =  \frac{1}{\lvert \mathcal{V} \rvert} \sum_{\bm{v} \in \mathcal{V}} \prod_{\ell=1}^L \phi_{\ell,\theta}(\Ynetk{1}_{v_{\ell}}), \quad \widetilde{C}_{H,\mathcal{C}}^{(K)} =  \frac{1}{\lvert \mathcal{V} \rvert} \sum_{\bm{v} \in \mathcal{V}} \prod_{\ell=1}^L \phi_{\ell,\theta}(\Ynetk{K}_{v_{\ell}})
\]
are unbiased estimators of $C_H^{(1)}$ and $C_H^{(K)}$ under the dynamic comparison model.
Note that by construction, $\widetilde{C}_{H,\mathcal{C}}^{(1)}$ and $\widetilde{C}_{H,\mathcal{C}}^{(K)}$ are independent since we condition on $\Anetk{1}$ and $\Anetk{K}$. 

When $\theta$ is known, estimators of the asymptotic variances of $\widetilde{C}_{H,\mathcal{C}}^{(1)}$ and $\widetilde{C}_{H,\mathcal{C}}^{(K)}$ can be found by an application of the bootstrap methodology developed in \citebody{chang22estimation}, or the related bootstrap developed in Section~\ref{sec:ho} (specified to $K=1$). 

When $\theta$ is unknown, some non-negligible covariance is induced between the two plug-in estimators
\[
  \widehat{C}_{H,\mathcal{C}}^{(1)} =  \frac{1}{\lvert \mathcal{V} \rvert} \sum_{\bm{v} \in \mathcal{V}} \prod_{\ell=1}^L \phi_{\ell,\hat{\theta}}(\Ynetk{1}_{v_{\ell}}), \quad \widehat{C}_{H,\mathcal{C}}^{(K)} =  \frac{1}{\lvert \mathcal{V} \rvert} \sum_{\bm{v} \in \mathcal{V}} \prod_{\ell=1}^L \phi_{\ell,\hat{\theta}}(\Ynetk{K}_{v_{\ell}}). 
\]
We provide the details of this approach in 
\ifarxiv
Appendix~\ref{app:ho_comparison_theta_unknown}.
\fi

\subsection{Other comparison problems} \label{subsec:comparison_problems}

In principle, the methodology developed in this section can be extended to other subgraph density comparison problems.
Suppose we observe a sequence of network snapshots $\{\bm{Y}_k\}_{k=1}^K$ and wish to do (joint) inference on the subgraph densities at a collection of these snapshot times $\mathcal{I} \subset \{1,\ldots,K\}$.
To this end, assume the edgewise error and evolution model specified in Section~\ref{sec:model}, but fix the subgraph densities of interest by treating the underlying network snapshots at all time indices in $\mathcal{I}$ as high-dimensional nuisance parameters.

Put into this general framework, Sections~\ref{sec:estimation} and \ref{sec:ho} develop methodology for $\mathcal{I} = \{1\}$, and Sections~\ref{subsec:gmm_comparison} and \ref{subsec:ho_comparison}, develop methodology for $\mathcal{I} = \{1,K\}$.
Other choices of interest could be a regular grid of indices, or a sequence of consective snapshots at the beginning or end of the observation interval.

To apply the GMM estimation and inference described in Sections~\ref{sec:estimation} and~\ref{subsec:gmm_comparison}, we require sufficiently many snapshots {\em not} in $\mathcal{I}$ for parameter identification: a subsequence of at least 2 at the beginning or end of the sequence, or at least 3 strictly between two snapshot indices in $\mathcal{I}$.
The higher-order estimation and inference methods developed in Sections~\ref{sec:ho} and~\ref{subsec:ho_comparison} can be extended to accomodate arbitrary $\mathcal{I}$, with suitable adjustments to the bootstrap sampler.

\section{Simulations} \label{sec:simulation}

In this section, we perform simulation studies to verify the performance of our proposed GMM estimators, and the nominal confidence interval (CI) coverage of the asymptotic and bootstrap approximations. All CI's are constructed to have 90\% nominal coverage rate.

\subsection{Inference for edge density} \label{subsec:density_sim}

In this section we verify the properties of our GMM estimators of edge density under the noisy dynamic model, through $3$ simulation settings. 
In all settings, $\alpha=0.05$, $\beta=0.2$.
The initial snapshot $\Anetk{1}$ is assigned edges uniformly at random to have edge density $\delta_1 = 0.4$.
We compare our GMM estimator of $\delta_1$ to a naive estimator, the empirical edge density of $\Ynetk{1}$, as well as the estimator and asymptotic CI developed in \citebody{chang22estimation} (``C22'') which uses only the first $3$ observed snapshots.

In setting 1, we consider sequences of noisy dynamic networks with $K=11$, $\lambda=0.12$, $\mu=0.08$, and vary $n$.
Results are reported in Figure~\ref{fig:density_varyn}.
From the left panel of Figure~\ref{fig:density_varyn}, our GMM CI's achieve the nominal coverage rate, while the C22 CI's have decreasing coverage rate as $n$ increases.
As C22 assumes iid replicates, the addition of stochastic evolution biases the estimators of $\alpha$, $\beta$, and $\delta_1$, seen in the right panel of Figure~\ref{fig:density_varyn}.

\begin{figure}[ht]
\twoImages{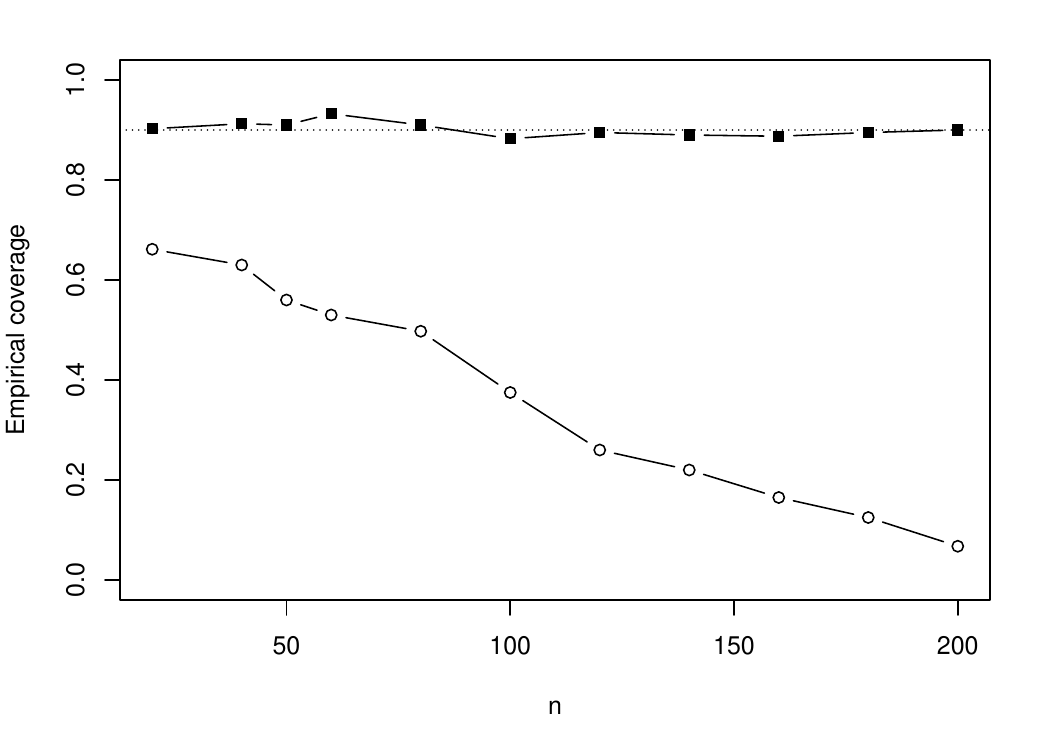}{0.5\textwidth}{}{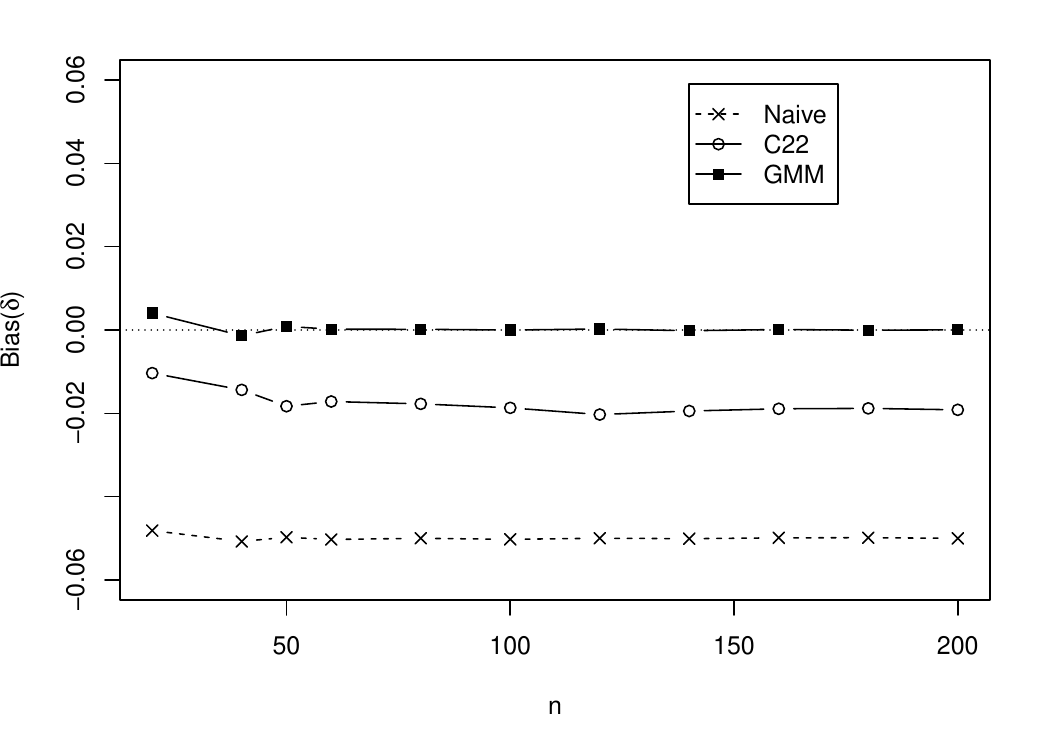}{0.5\textwidth}{}
\caption{Plots of 90\% CI coverage and bias for estimators of edge density $\delta_1$. \label{fig:density_varyn}}
\end{figure}

In setting 2, we consider sequences of noisy dynamic networks with $n=50$, $\lambda=0.12$, $\mu=0.08$, and vary $K$.
Results are reported in Figure~\ref{fig:density_varyK}.
From the left panel of Figure~\ref{fig:density_varyK}, our GMM CI's achieve the nominal coverage rate, while the C22 CI's do not.
In the right panel of Figure~\ref{fig:density_varyK}, our GMM estimator of $\delta_1$ achieves a smaller RMSE than the naive or C22 estimators. Moreover, our estimator has decreasing RMSE as we incorporate more network snapshots, until the edge sequences mix. 
In Figure~\ref{fig:density_pieK}, we show that this increased efficiency is due to the sharing of information through the reweighted snapshot estimators, by visualizing the optimal snapshot-wise weights (cf.~Remark~\ref{rem:thetaknown}).

\begin{figure}[ht]
\twoImages{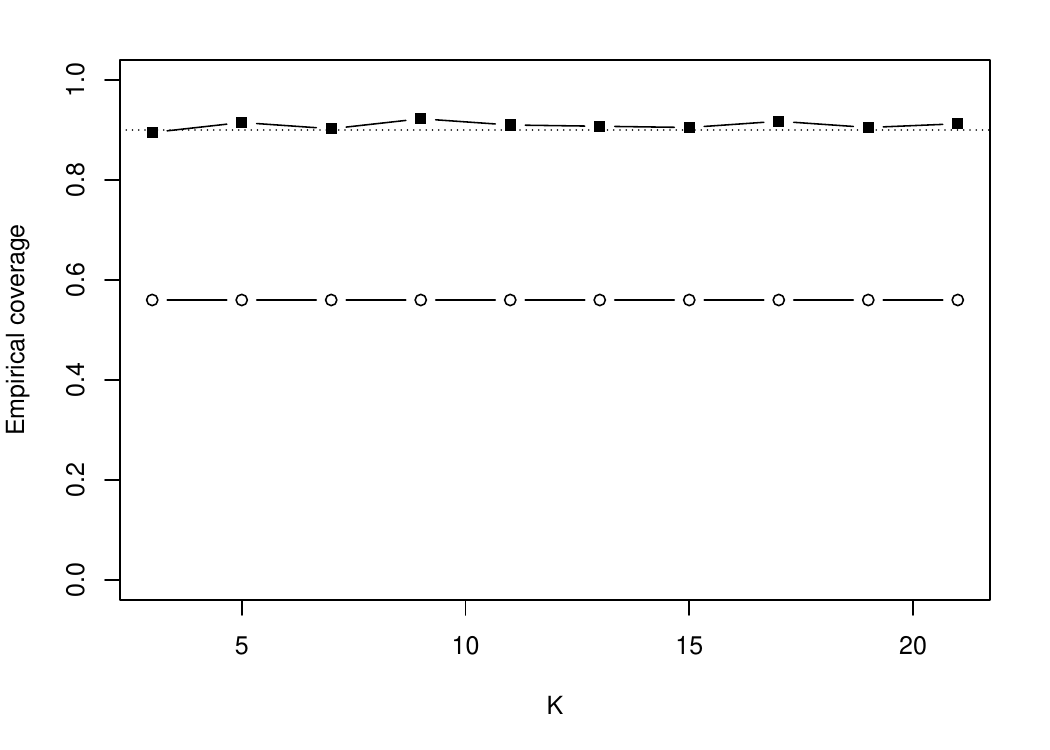}{0.5\textwidth}{}{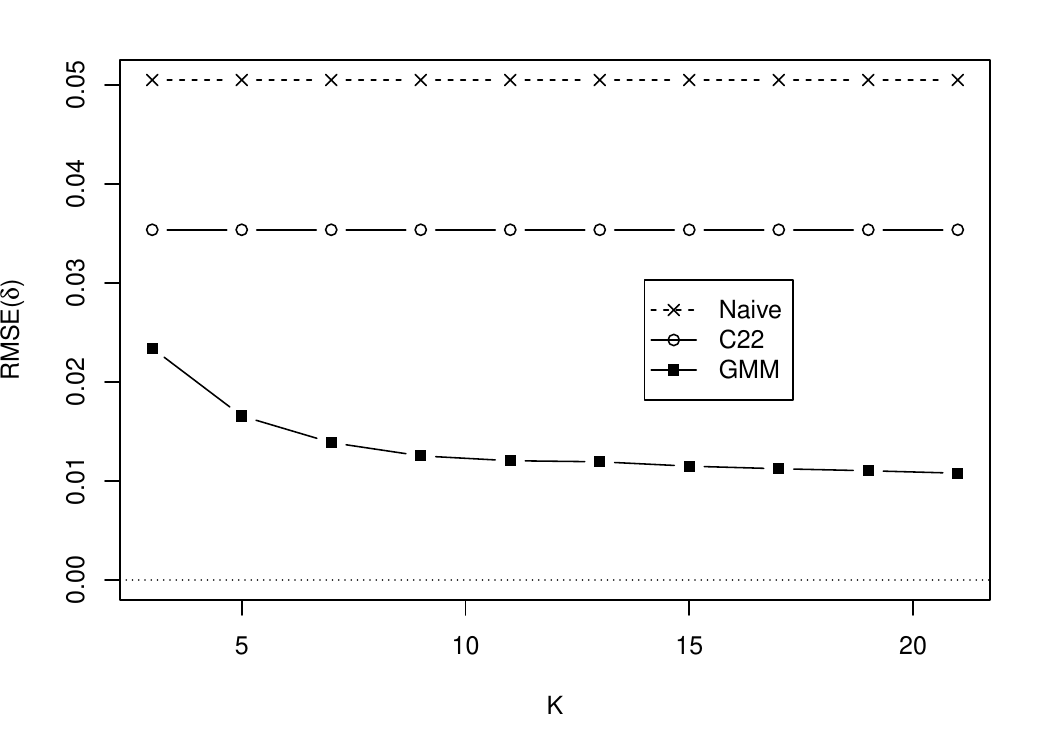}{0.5\textwidth}{}
\caption{Plots of 90\% CI coverage and RMSE for estimators of edge density $\delta_1$. \label{fig:density_varyK}}
\end{figure}

\begin{figure}[ht]
\begin{center}
  \includegraphics[width=0.9\textwidth]{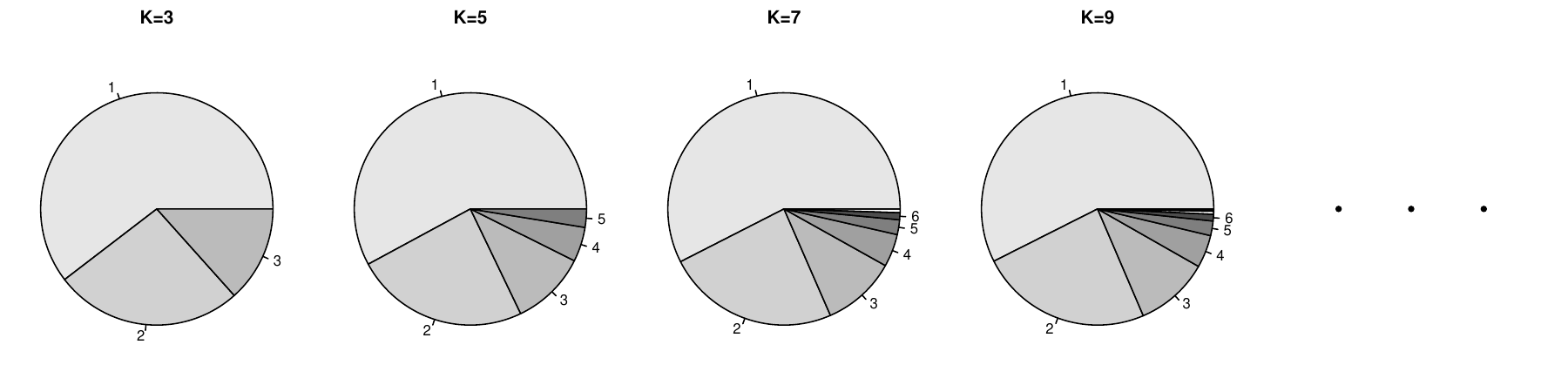}
\end{center}
\caption{Pie chart visualizing weighting of snapshot-wise edge density estimators, $\delta_1=0.4$, $\alpha=0.05$, $\beta=0.2$, $\lambda=0.12$, $\mu=0.08$, varying $K$. Weights for $k \geq 7$ are negligible. Weights for $K > 9$ are visually identical to the weights for $K=9$.\label{fig:density_pieK}}
\end{figure}

In setting 3, we consider sequences of noisy dynamic networks with $n=50$, $K=11$, $\lambda=0.6(1-\gamma)$, $\mu=0.4(1-\gamma)$, and vary $\gamma$.
As $\gamma$ approaches $1$, later underlying snapshots will more closely resemble $\Anetk{1}$.
Results are similar to those in Figure~\ref{fig:density_varyK}, and our estimator has decreasing RMSE as $\gamma$ approaches $1$.
The RMSE for the C22 estimator also decreases, as the smaller evolution rate leads to less bias in their estimates of $\alpha$ and $\beta$.
In Figure~\ref{fig:density_piegamma}, we again visualize the optimal snapshot-wise weights, and see that indeed more information can be used from later snapshots when $\gamma$ is close to $1$.

\begin{figure}[ht]
\twoImages{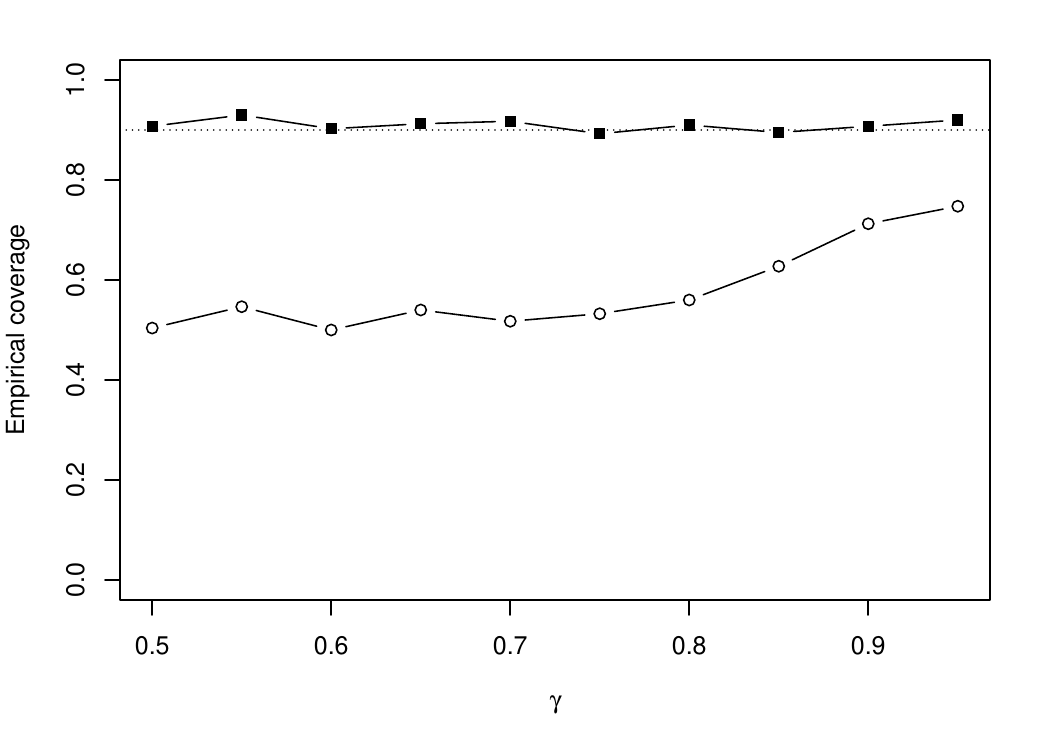}{0.5\textwidth}{}{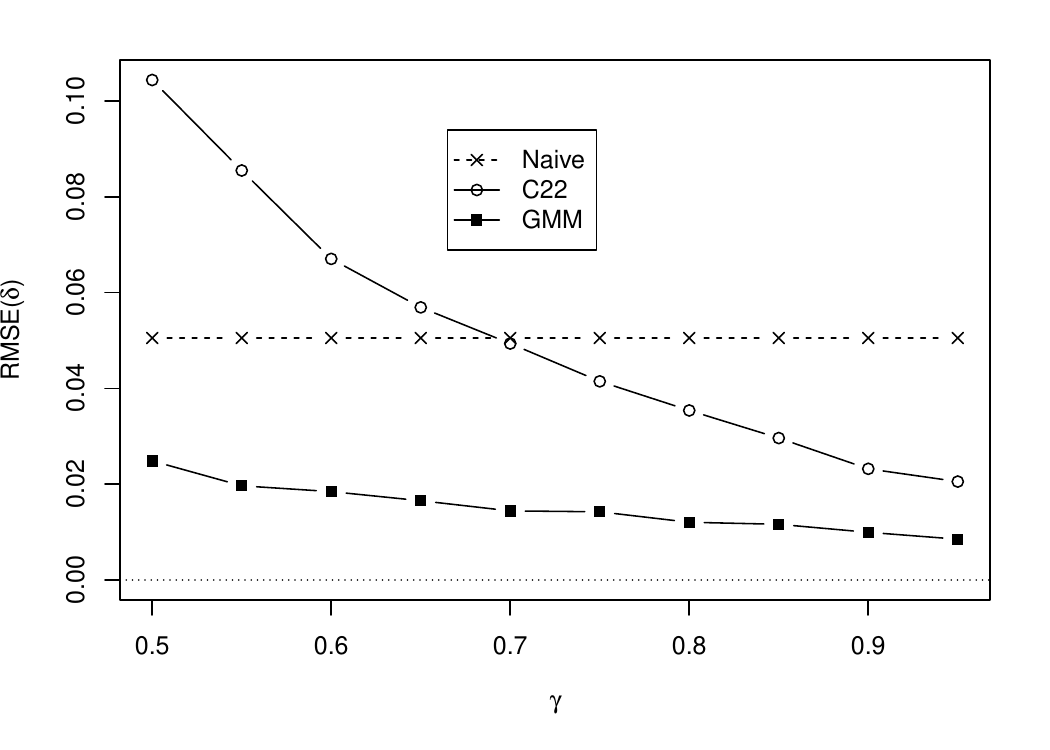}{0.5\textwidth}{}
\caption{Plots of 90\% CI coverage and RMSE for estimators of edge density $\delta_1$. \label{fig:density_varygamma}}
\end{figure}

Results are similar to those in Figure~\ref{fig:density_varyK}, and our estimator has decreasing RMSE as $\gamma$ approaches $1$.
The RMSE for the C22 estimator also decreases, as the smaller evolution rate leads to less bias in their estimates of $\alpha$ and $\beta$.
In Figure~\ref{fig:density_piegamma}, we again visualize the optimal snapshot-wise weights, and see that indeed more information can used from later snapshots when $\gamma$ is close to $1$.

\begin{figure}[ht]
\begin{center}
  \includegraphics[width=\textwidth]{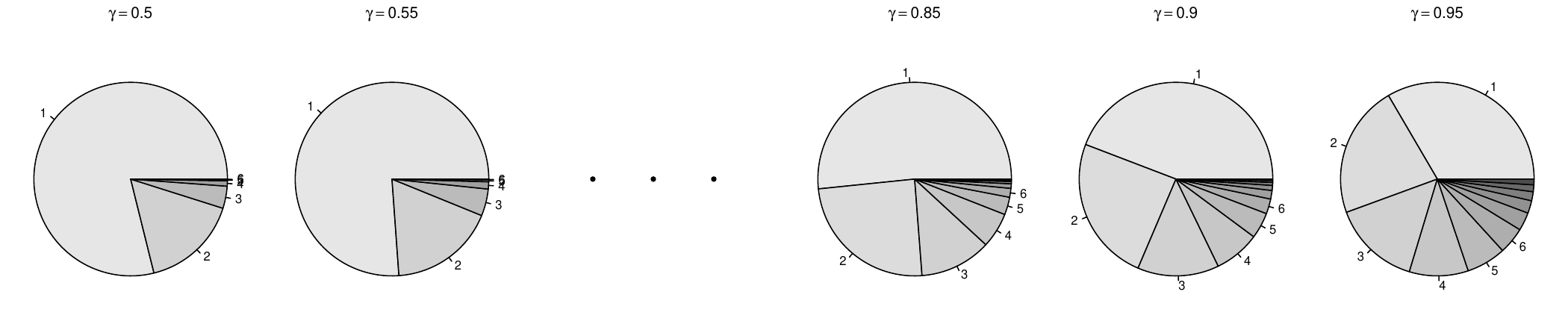}
\end{center}
\caption{Pie chart visualizing weighting of snapshot-wise edge density estimators, $K=11$, $\delta_1=0.4$, $\alpha=0.05$, $\beta=0.2$, $\lambda=0.6(1-\gamma)$, $\mu=0.4(1-\gamma)$, varying $\gamma$. \label{fig:density_piegamma}}
\end{figure}

\subsection{Inference for higher-order subgraph densities}

In this section we verify the inferential properties of our weighted estimators of higher-order subgraph densities under the noisy dynamic model, assuming the parameters in $\theta$ are known (as in Section~\ref{sec:ho}).
We compare our adaptively reweighted ``Boot-wt'' estimator with bootstrap CI based on $500$ samples to the ``C22'' estimator developed in \citebody{chang22estimation} with bootstrap CI based on $500$ samples, and the naive empirical subgraph density of $\Ynetk{1}$.
For comparison, we also implement the Boot-wt estimator which uses only the first two observed snapshots for estimation (``Boot-wt, $K=2$'').
We evaluate the performance for the triangle density of $\Anetk{1}$, denoted $C_{\triangle}$, which counts the rate at which all three potential edges are present among a set of three nodes.

In these simulations, we vary $n$, fix $K=5$, and consider four settings for the error and evolution parameters: ``Low error, low stability'' ($\alpha=0.05$, $\beta=0.1$, $\lambda=0.2$, $\mu=0.2$), ``High error, low stability'' ($\alpha=0.25$, $\beta=0.25$, $\lambda=0.2$, $\mu=0.2$), ``Low error, high stability'' ($\alpha=0.05$, $\beta=0.1$, $\lambda=0.05$, $\mu=0.05$), and ``Low error, low stability'' ($\alpha=0.05$, $\beta=0.1$, $\lambda=0.05$, $\mu=0.05$).
Note that throughout this section, inference is performed using the ground truth values of these error and evolution parameters.
The initial network $\Anetk{1}$ is constructed as an induced subgraph of a perturbed lattice graph generated 
on $400$ nodes, and has ground truth edge density $\delta_1 \approx 0.4$, and triangle density $C_{\triangle} \approx 0.083$. 

In Figure~\ref{fig:tri_rmse}, we plot the RMSE of three estimators of $C_{\triangle}$. In all cases we see the Boot-wt and C22 estimators have RMSE approaching zero as $n$ increases, while the naive estimator remains biased. In three settings, the Boot-wt estimator is a more efficient estimator than the C22 estimator, by incorporating information from later snapshots.
The exceptional case is the ``High error, low stability'' setting, where the later snapshots provide very little information about $\Anetk{1}$, making the adaptive weighting approach highly unstable.
In this case, the Boot-wt $K=2$ estimator is preferred, and its performance is also comparable to C22.
Note that the effect of instability is exacerbated for estimation of higher-order subgraph density (relative to edge density): roughly, information about edge density will be preserved in later snapshots if edges do not transition out of their initial state, while information about triangle 
will only be preserved in later snapshots if all three potential edges do not transition out of their initial states.

\begin{figure}[ht]
\fourImages{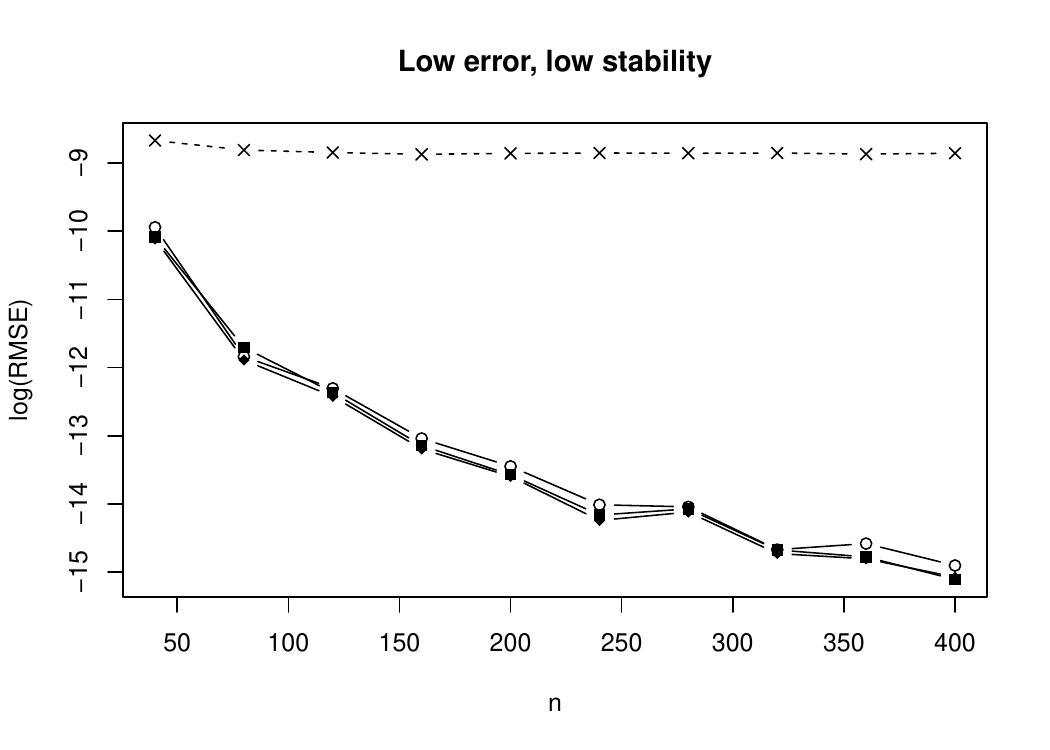}{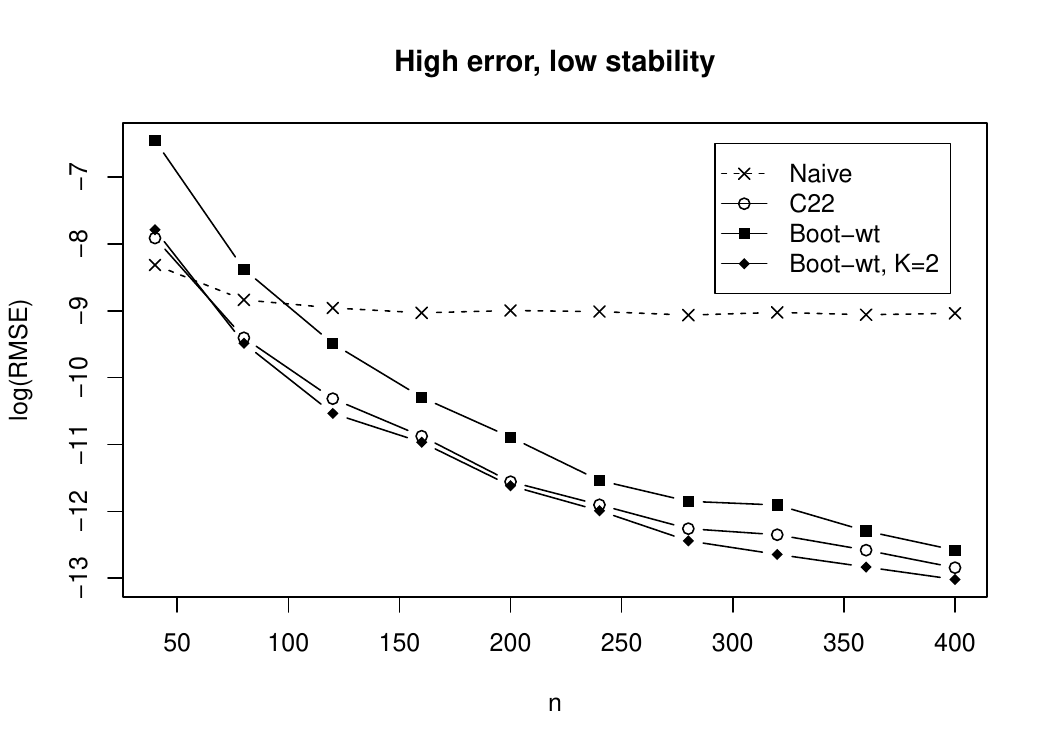}{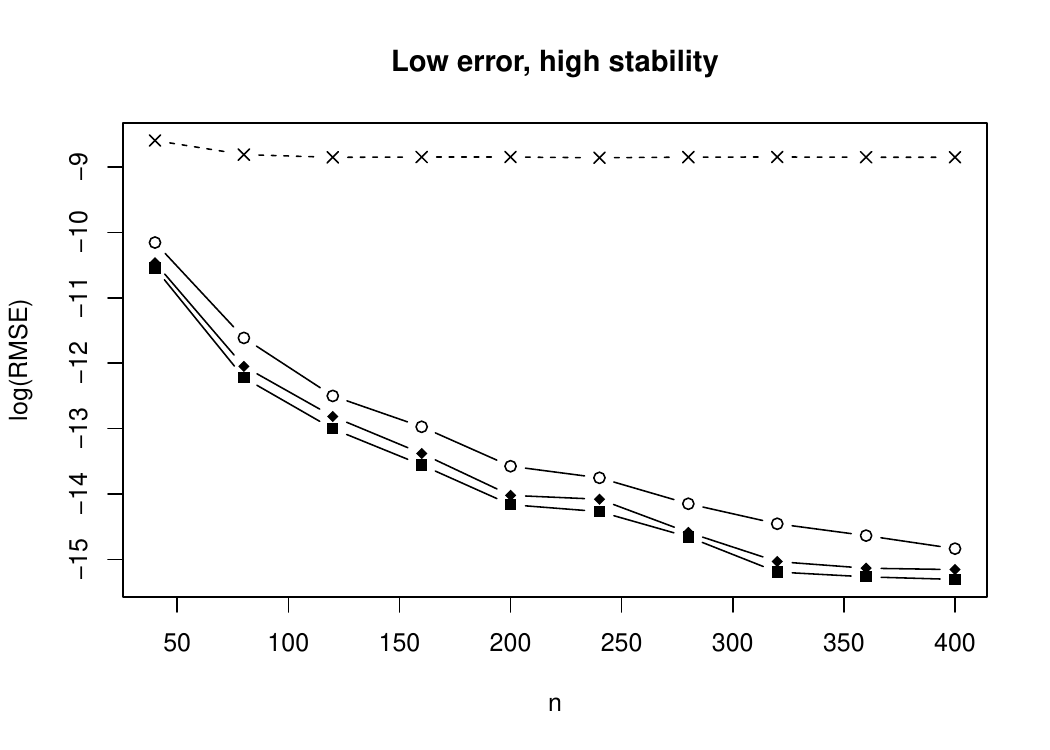}{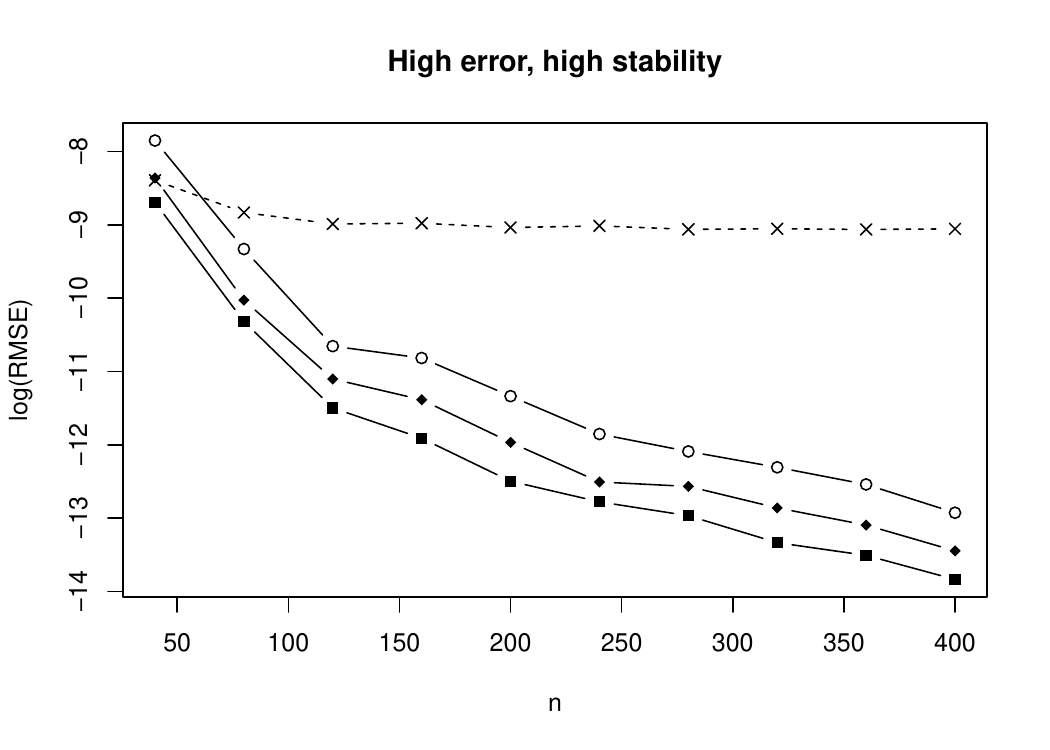}{0.55\textwidth}{0.55\textwidth}
\caption{Plots of RMSE for estimators of triangle density $C_{\triangle}$. \label{fig:tri_rmse}}
\end{figure}

In Figure~\ref{fig:tri_coverage}, we plot the 90\% CI coverage of two estimators of $C_{\triangle}$. Other than the ``High error, low stability'' setting, all methods give CI's with nominal coverage, even for very small values of $n$. The instability of adaptive weighting in the  ``High error, low stability'' setting causes issues with coverage of the Boot-wt estimator which uses all $5$ observed snapshots.
In practice, for inference on higher-order subgraph densities, we caution against incorporating later snapshots when they provide negligible information about $\Anetk{1}$.
For instance, in the ``High error, low stability'' setting, valid and efficient inference can be achieved using only the first 1-2 observed snapshots, while later snapshots make the adaptive weights and covariance estimation unstable.
Finally, note that this study assumes oracle knowledge of the error and evolution rate parameters in $\theta$.
Following Section~\ref{subsec:density_sim}, when $\theta$ is unknown, our GMM-based approach can provide an unbiased estimators of $\theta$ in cases where the C22 estimator will be biased.

\begin{figure}[ht]
\fourImages{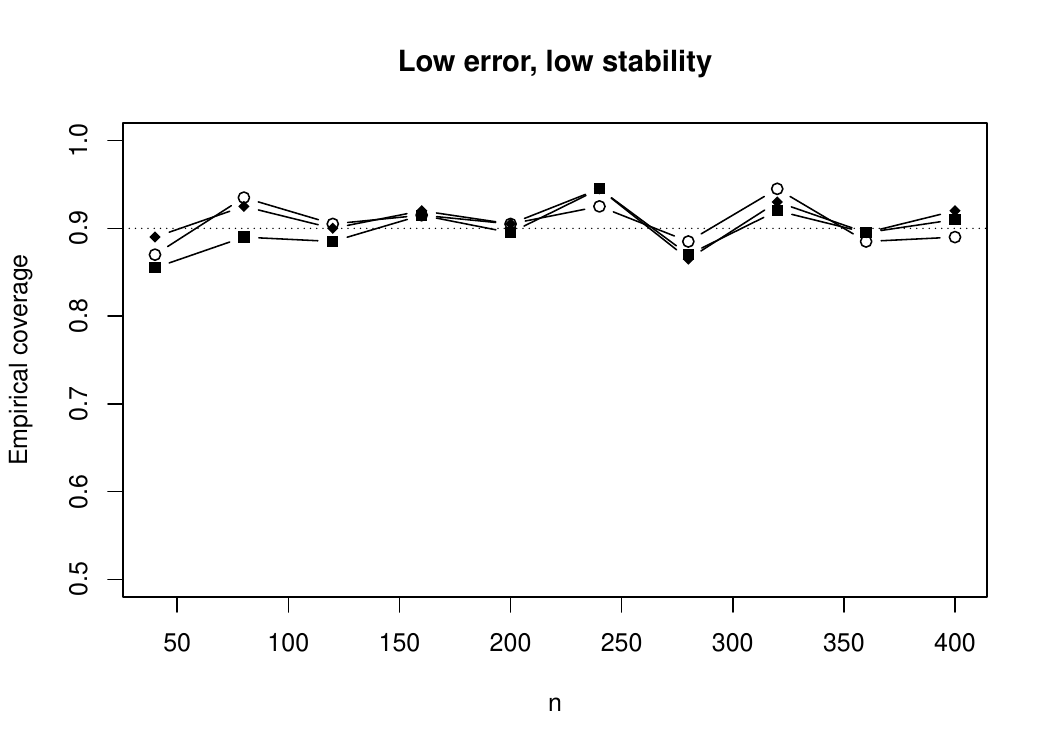}{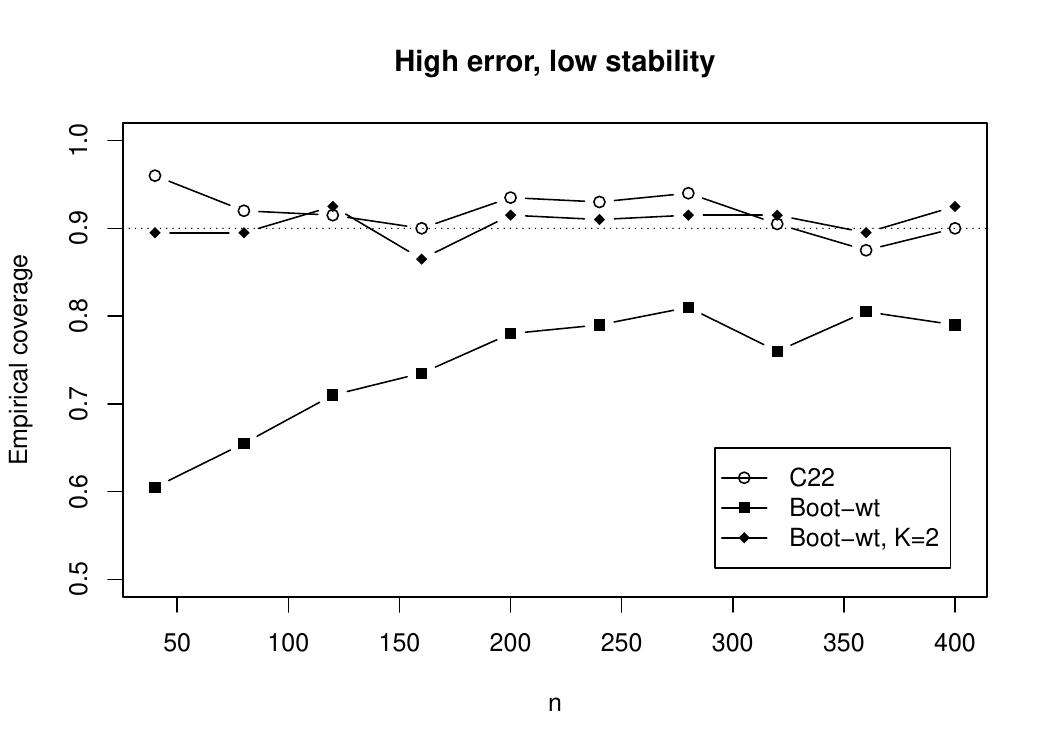}{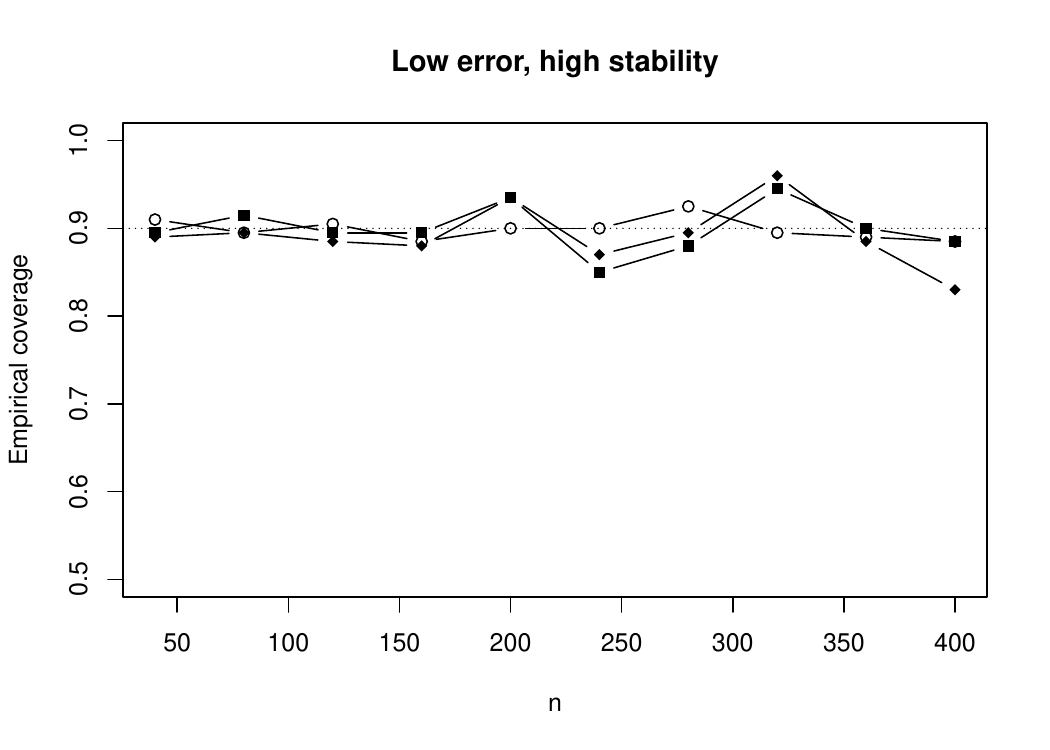}{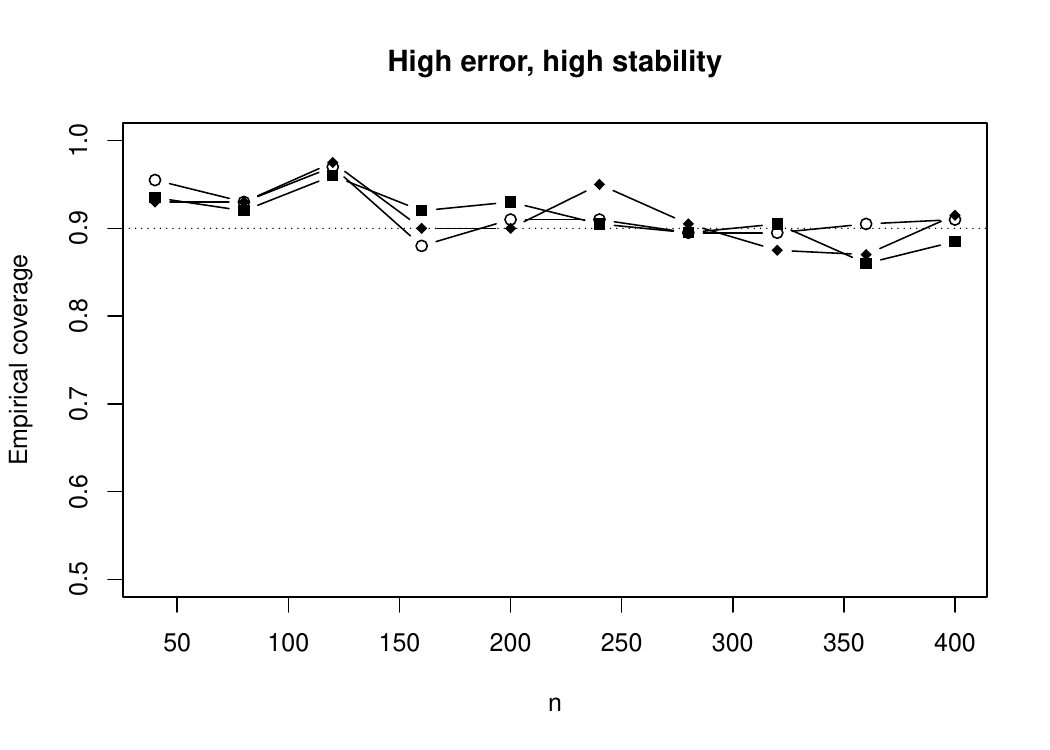}{0.55\textwidth}{0.55\textwidth}
\caption{Plots of 90\% CI coverage for estimators of triangle density $C_{\triangle}$. \label{fig:tri_coverage}}
\end{figure}

\subsection{Dynamic comparison of edge densities}

In this section we verify the properties of our GMM estimators of edge density under the dynamic comparison model, through two simulation settings.
In both settings, we consider sequences of noisy dynamic networks with $\alpha=0.05$, $\beta=0.2$, $\lambda=0.12$, and $\mu=0.08$.
The initial snapshot $\Anetk{1}$ is assigned edges uniformly at random to match the desired edge density $\delta_1 = 0.4$.
The final snapshot $\Anetk{K}$ is assigned edges uniformly at random with varying edge density $\delta_K$, where $\delta_K$ is chosen from $\{0.2,0.4,0.6,0.8\}$.
We compare our GMM estimator of $\delta_1$ to a naive estimator, the empirical edge density of $\Ynetk{1}$, as well as C22.
The C22 estimator is implemented by applying their methodology to the first $3$ and last $3$ snapshots respectively, and taking the difference of the two edge density estimators.
A CI for the C22 estimate of the difference in densities is found from the individual CI's, assuming the two estimators are independent.

In setting 1, we consider sequences of noisy dynamic networks with $K=11$, and vary $n$ and $\delta_K$.
Results are reported in Figure~\ref{fig:comparison_varyn}.
From the left panel of Figure~\ref{fig:comparison_varyn}, our GMM CI's achieve the nominal coverage rate, while the C22 CI's do not.
The coverage of the C22 CI's becomes worse as $n$ increases, or as the difference between $\delta_1$ and $\delta_K$ increases.
In the right panel of Figure~\ref{fig:comparison_varyn}, our GMM estimator of $\delta_K - \delta_1$ achieves a smaller RMSE than the naive or C22 estimators.
When $\delta_K = 0.4$, the effect of the bias on the naive estimators will cancel out.

\begin{figure}[ht]
\twoImages{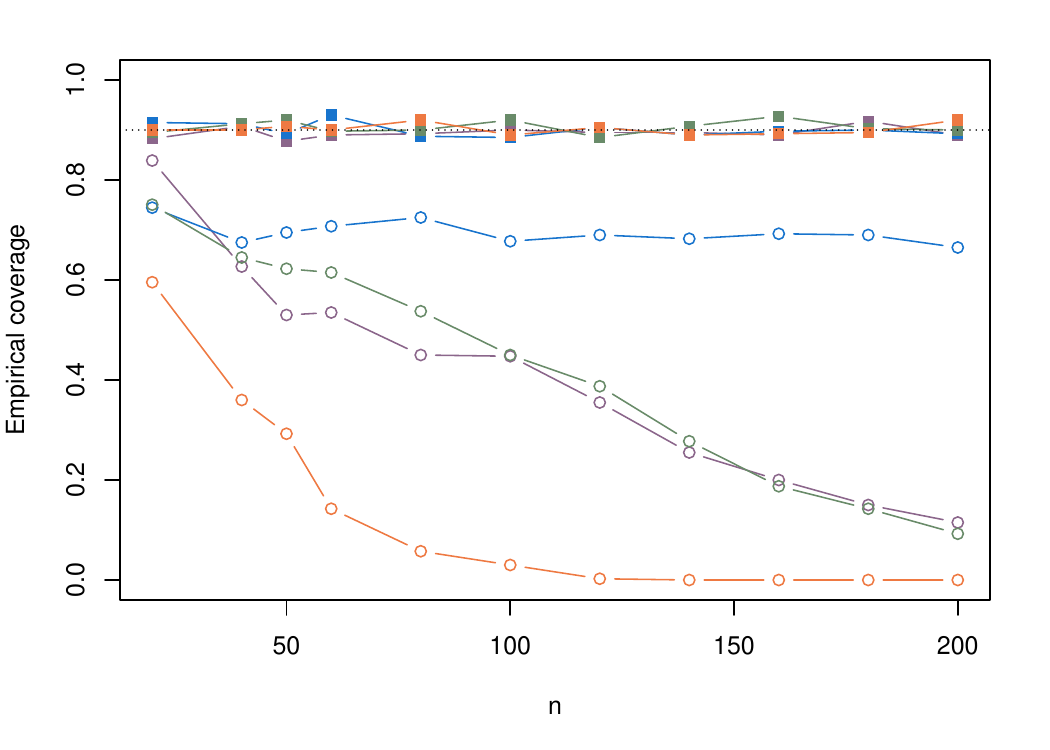}{0.55\textwidth}{}{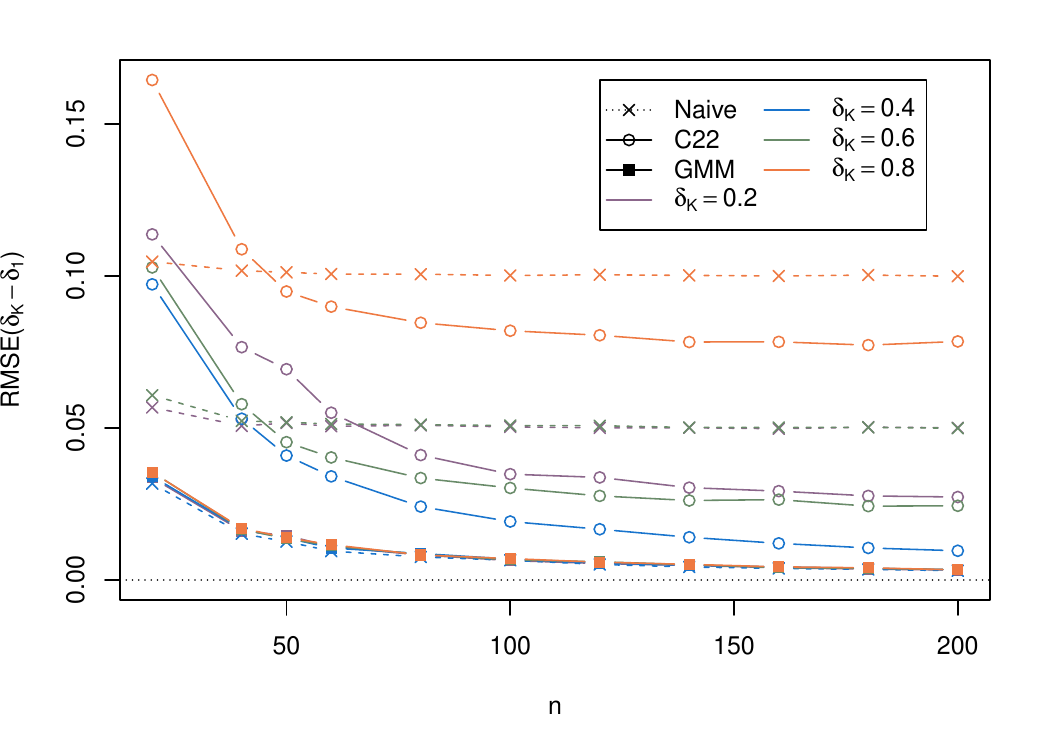}{0.55\textwidth}{}
\caption{Plots of 90\% CI coverage and RMSE for estimators of difference in edge density $\delta_K - \delta_1$. Points are colored according to the value of $\delta_K$. \label{fig:comparison_varyn}}
\end{figure}

In setting 2, we consider sequences of noisy dynamic networks with $n=50$, and vary $K$ and $\delta_K$.
Results are reported in Figure~\ref{fig:comparison_varyK}.
From the left panel of Figure~\ref{fig:comparison_varyK}, we again see that our GMM CI's achieve the nominal coverage rate, while the C22 CI's do not.
In the right panel of Figure~\ref{fig:comparison_varyK}, our GMM estimator of $\delta_K - \delta_1$ achieves a smaller RMSE than the naive or C22 estimators.
The improvement in RMSE for the GMM estimator is modest as $K$ increases.

\begin{figure}[ht]
\twoImages{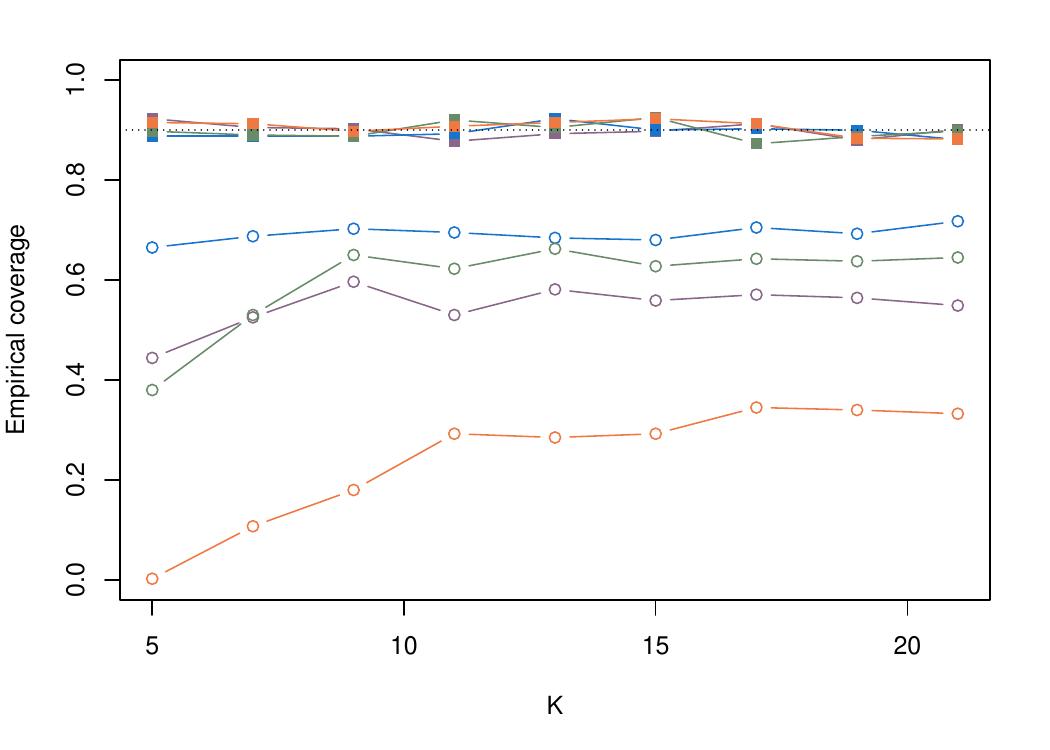}{0.55\textwidth}{}{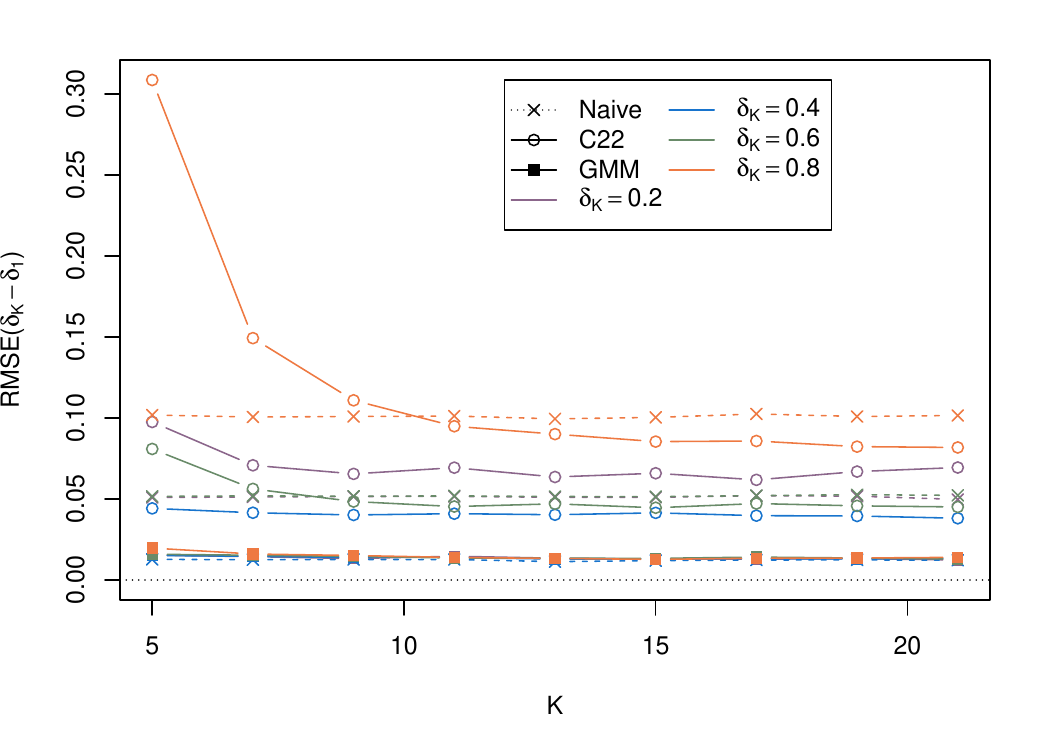}{0.55\textwidth}{}
\caption{Plots of 90\% CI coverage and RMSE for estimators of difference in edge density $\delta_K - \delta_1$. Points are colored according to the value of $\delta_K$. \label{fig:comparison_varyK}}
\end{figure}

\section{Analysis of Learning Brain data} \label{sec:real_data}

In this section, we analyze networks constructed from the Learning Brain dataset, available from OpenNeuro \citepbody{learningbrain}.
This dataset includes $K=5$ longitudinal observations of $23$ female subjects aged between 18 and 26, collected over 6 months while subjects received training in piano.
For each session, functional magnetic resonance imaging (fMRI) scans were taken in a resting state, as well as during a passive listening task.
The sessions were not equally spaced in time, with the schedule (0, 1, 6, 13 and 26 weeks) designed to mimic the non-linear nature of brain plasticity, as the rate of change is hypothesized to slow as training progresses \citepbody{olszewska2021musical}. Thus this design supports the use of time homogeneous evolution parameters, as in Section~\ref{sec:comparison}.
Previous published analysis of this data investigated blood oxygenation level-dependent (BOLD) signals, rather than network connectivity or correlations.
Analysis of the listening task scans did not find any clusters of brain voxels whose BOLD signals changed significantly over the course of training \citepbody{olszewska2025piano}.

For this analysis, BOLD signals were processed using a standard pipeline, using the Python library {\tt nilearn}, to remove confounding effects due to head movement and non-neural activity and to map voxel-wise signals onto a standard atlas with $n=200$ ROIs \citepbody{schaefer2018local}.
Binary adjacency matrices were constructed by thresholding the absolute cross-correlations between signals at a fixed level $0.4$.
We use the methods developed in Section~\ref{subsec:gmm_comparison} and 
\ifarxiv
Appendix~\ref{app:ho_comparison_theta_unknown}
\fi
to compare both edge density and summaries based on higher-order subgraph densities at the beginning and end of each subject's training.
We analyze each subject independently, and freely estimate individual evolution, error and subgraph density parameters.
In addition to the model parameters $(\delta_1,\delta_K,\alpha,\beta,\lambda,\mu)$, we also perform inference for the difference in edge density between underlying snapshots $1$ and $K$, denoted $\Delta \delta = \delta_K - \delta_1$, as well as changes in two higher-order summaries: change in clustering coefficient, given by
$$
    \Delta \mathrm{cc} = \frac{C_{\triangle}^{(K)}}{C_{\triangle}^{(K)} + C_{\vee}^{(K)}} - \frac{C_{\triangle}^{(1)}}{C_{\triangle}^{(1)} + C_{\vee}^{(1)}},
$$
which can be interpreted as the proportion of closed two-stars; and change in normalized triangle density, given by
$$
   \Delta \mathrm{nt} = \frac{C^{(K)}_{\triangle}}{\delta_K^3} - \frac{C^{(1)}_{\triangle}}{\delta_1^3}.
$$
Asymptotic standard errors for higher-order summaries are based on a bootstrap sampler with $400$ samples.

\begin{figure}[ht]
\centering
\includegraphics[width=\linewidth]{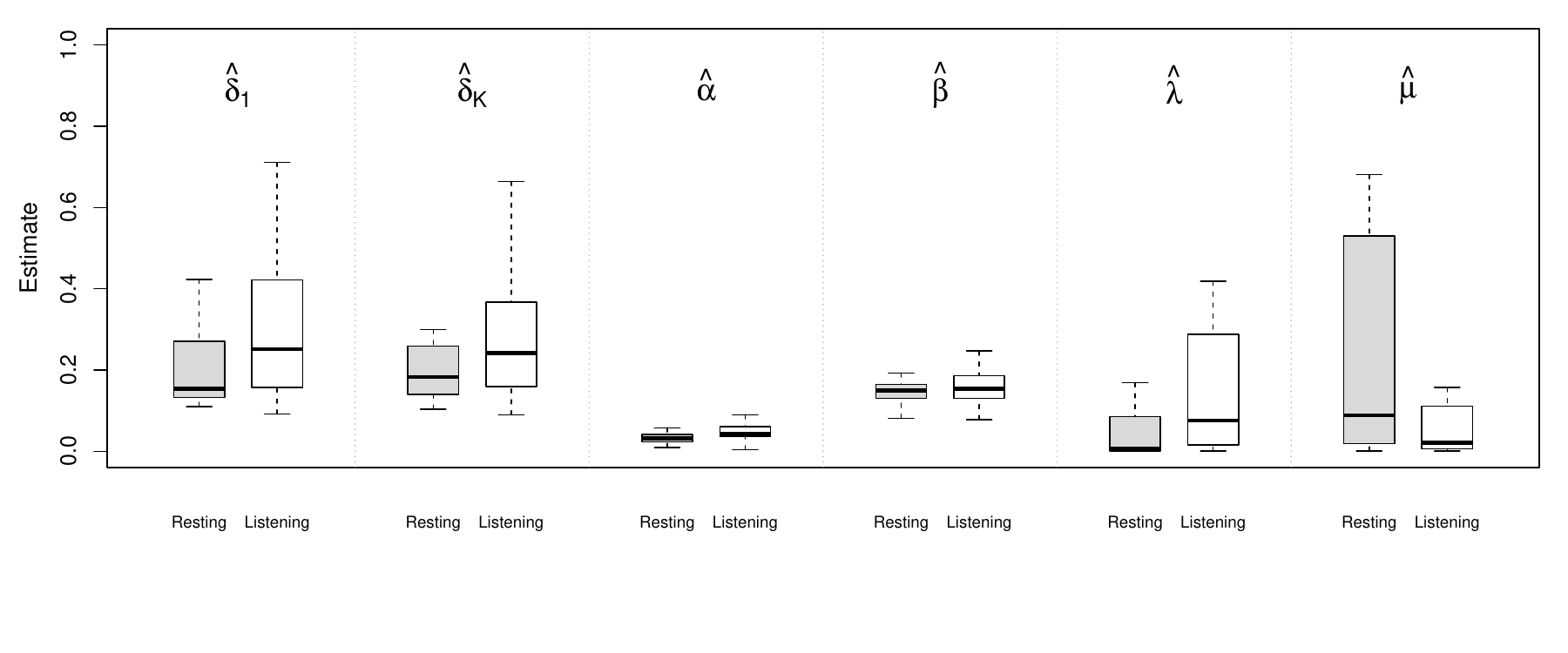}
\caption{Boxplots of parameter estimates for 23 subjects, resting and listening networks. \label{fig:neuro_box}}
\end{figure}

From the summaries in Figure~\ref{fig:neuro_box}, we see that these adjacency matrix sequences tend to show asymmetric type I and type II error rates (which will bias a naive edge density estimator), and non-zero edge transition rates (which bias the C22 method that assumes iid replicates).  Figure~\ref{fig:neuro_box}, along with full testing results in 
\ifarxiv
Appendix~\ref{app:realdata_2},
\fi
show that these subjects are highly heterogeneous.
This is to be expected, as the high-dimensional parameters $\Anetk{1}$ and $\Anetk{K}$ are allowed to vary freely for each subject.
Overall, we find that the listening networks tend to have higher density than the corresponding resting networks at both the beginning and end of the study, and find many significant changes in edge density in both sets of networks at the Bonferroni-corrected 5\% level (20/23 resting state, 19/23 listening task).
The box plots show that listening scan sequences are more often characterized by birth of new edges and low death rate, where the opposite is true for the resting scan sequences.
Some subjects show significant differences in clustering coefficient and/or triangle density, implying changes in higher level organization, with 17 total significant differences in clustering coefficient and triangle density found in the listening networks, compared to 11 in the resting state networks. Detailed testing results are provided in Table~\ref{tab:neuro_tests} in 
\ifarxiv
Appendix~\ref{app:realdata_2}.
\fi

\section{Conclusions and future work} \label{sec:conclusion}

In this work, we develop a model for dynamically observed noisy networks, and methodology to estimate and perform inference on edge densities and higher-order subgraph densities, either at a fixed snapshot or at multiple snapshots, relaxing the assumption of iid replication and allowing for dynamic comparison of subgraph densities.
We extend a bootstrapping approach of \citebody{chang22estimation} to tractably estimate the asymptotic variance and covariances of these estimators.

There are several potential extensions of the model presented in this paper, including heterogeneity of error and evolution parameters across both nodes and times, or dependent bit-flipping models as in \citebody{chang2026autoregressive}. 
Analysis in high-dimensional settings where the number of error and evolution parameters grows with $n$ may also require more intricate asymptotics for growing $K$, to understand the balance between model flexibility and sample size.
Variance estimation without bootstrapping is also possible in principle (cf. Remark~\ref{rem:ustats} in the \ifarxiv
appendix),
\fi
but is computationally challenging; the bootstrap developed in this paper is one form of stochastic approximation, but it may be valuable to consider whether other approximations or fast variance computations are possible.
Finally, we can consider the application of this noisy dynamic evolution model to perform inference for other statistics, for instance other global graph summaries such as branching factor \citebody{li2022estimation}, or the density of a {\em temporal motif} as in \citebody{zhu2022quantifying}. 

\noindent
\ifblind
\else
\ifarxiv
\else
{\bf Funding.} This work was supported by the National Sciences and Engineering Research Council of Canada (NSERC) grant RGPIN-2025-02892 to PWM, an NSERC Postdoctoral Fellowship to PWM, and NSERC grants RGPIN-2023-03566 and DGDND-2023-03566 to EDK.

{\bf Disclosure statement.} The authors report there are no competing interests to declare.

{\bf Declaration of generative AI use.} The authors report that generative AI (Microsoft Copilot) was used as coding assistance to create visualizations and mathematical symbols when typesetting this manuscript. 
\fi
\fi

\ifarxiv
{\bf Data.}
\fi
The data that support the findings in this paper are openly available in OpenNeuro at \url{http://doi.org/10.18112/openneuro.ds007022.v1.0.1}.

\bibliographystylebody{abbrvnat}
\bibliographybody{mybib0}

\pagebreak

\appendix

\section{Implementation of GMM estimation}

\subsection{Evolution and error model} \label{app:gmm_computation}

Implementation of GMM estimation requires the evaluation of joint expectations and covariances under the evolution and error model. We prove the following result, a special case of \citesupp{hsu12spectralB}, Lemma 1, which gives a simple formulation for the joint probabilities from finite state space hidden Markov models (HMMs) in terms of so-called {\em observation operators} \citepsupp{jaeger00observableB}.

\begin{lemma} \label{lem:gmm_computation}
  Suppose $\{Z_t\}_{t \geq 0}$ is a finite state, homogeneous HMM on state space $\{1,\ldots,S\}$. Denote the underlying chain by $\{X_t\}_{t \geq 0}$, also on state space $\{1,\ldots,S\}$. Denote the $s$-step transition matrix ($s \geq 0$) by
  $$
    P_s[i,j] = \prob(X_{t+s}=j ~\vert~ X_t=i) \in \real^{S \times S},
  $$
  and the observation distribution matrix by
  $$
    E[i,j] = \prob(Z_t=i ~\vert~ X_t=j) \in \real^{S \times S},
  $$
  for $t=1,2,\ldots$.
  For $s \geq 0$ and $i \in \{1,\ldots,S\}$, denote the observation operator matrix
  $$
    O^{(s)}_{i} = \begin{pmatrix}
      E[i,1] & & 0 \\
      & \ddots & \\
      0 & & E[i,S]
  \end{pmatrix} P_s \in \real^{S \times S}.
  $$
  Then the joint probabilities of $\{Z_t\}_{t \geq 1}$ conditional on the initial underlying state $X_0$ satisfy
  \begin{equation} \label{forward_prob}
    \prob\left( Z_{t_1}=i_1,\ldots,Z_{t_K}=i_K ~\vert~ X_0 \right) = \bm{1}_S^{\tp} O^{(t_K-t_{K-1})}_{i_1} \cdots O^{(t_2-t_1)}_{i_1} O^{(t_1)}_{i_1} \bm{e}_{X_0}
  \end{equation}
  for state sequence $\{i_1,\ldots,i_K\}$ and increasing seqence of integers $0 \leq t_1 < t_2 < \ldots < t_K$.
  Here, $e_x$ for $x \in \{1,\ldots,S\}$ denotes the $x$th standard basis vector in $\real^S$.
\end{lemma}

In our setting with binary edges, we will have $S=2$, and the expectations and covariances (second moments) of individual edge sequences can be written through indicator functions.
For example, without loss of generality using node pair $(1,2)$, and conditioning on $\Anetk{1}_{12}=0$,
\begin{equation} \label{dens_example}
  \expect_{\theta}( D_{5,12} ~\vert~ \Anetk{1}_{12}=0 ) = \prob_{\theta}(\Ynetk{5}_{12}=1 ~\vert~ \Anetk{1}_{12}=0 )
\end{equation}
\begin{equation} \label{triple_example}
  \expect_{\theta}( T_{011,12} ~\vert~ \Anetk{1}_{12}=0 ) = \sum_{k=3}^K \prob_{\theta}(\Ynetk{k}_{12}=0,\Ynetk{k-1}_{12}=1,\Ynetk{k-2}_{12}=1 ~\vert~ \Anetk{1}_{12}=0),
\end{equation}
and
\begin{equation} \label{cov_example}
   \expect_{\theta}( D_{5,12} T_{011,12} ~\vert~ \Anetk{1}_{12}=0 )
   = \sum_{k=3}^K \prob_{\theta}(\Ynetk{5}_{12}=1,\Ynetk{k}_{12}=0,\Ynetk{k-1}_{12}=1,\Ynetk{k-2}_{12}=1 ~\vert~ \Anetk{1}_{12}=0).
\end{equation}
Note that in \eqref{cov_example}, a probability on the right-hand side may not be of a sequence of four states, as it is written.
If $k \in \{6,7\}$, then the first event is contained in the third or fourth, leading to the probability of a sequence of three states.
If $k=5$, then the first and second events are mutually exclusive, and the probability is zero.

Thus, with some care, the probabilities on the right-hand sides of \eqref{dens_example}---\eqref{cov_example} can be evaluated using \eqref{forward_prob}, and the resulting expectations can be used to find the expectations and covariances needed for GMM estimation, see Section~\ref{sec:estimation}.
We evaluate derivatives of expectations (also needed for covariance estimation) numerically rather than analytically.

In order to efficiently calculate the $s$-step transition probabilities resulting from the one-step evolution model
\begin{equation*}
  \mathbb{P}(\Anetk{k+1}_{ij}=1 \vert \Anetk{k}_{ij}=0) = \lambda, \quad
  \mathbb{P}(\Anetk{k+1}_{ij}=0 \vert \Anetk{k}_{ij}=1) = \mu,
\end{equation*}
we use a reparameterization in terms of parameters $\tau$ and $\gamma$,
$$
  \tau = \frac{\lambda}{\lambda + \mu} \in (0,1), \quad \gamma = 1 - (\lambda + \mu) \in (-1,1),
$$
so that
$$
  P_s = P_s(\tau,\gamma) = \begin{pmatrix}
  1-\tau & \tau \\ 1-\tau & \tau
\end{pmatrix} + \gamma^s \begin{pmatrix}
  \tau & -\tau \\
  -(1-\tau) & 1-\tau
\end{pmatrix} \in (0,1)^{2 \times 2}
$$
is easy to calculate for all values of $s$.
The inverse parameterization is given by
$$
  \lambda = \tau (1 - \gamma), \quad \mu = (1-\tau)(1-\gamma).
$$

\subsection{Dynamic comparison model} \label{app:gmm_computation_compare}

In order to evaluate joint probabilities under the dynamic comparison model, we use a generalization of Lemma~\ref{lem:gmm_computation} to HMMs with inhomogeneous transition kernels

\begin{lemma} \label{lem:gmm_computation_comparison}
  Suppose $\{Z_t\}_{t \geq 0}$ is a finite state, inhomogeneous HMM on state space $\{1,\ldots,S\}$. Denote the underlying chain by $\{X_t\}_{t \geq 0}$, also on state space $\{1,\ldots,S\}$. Denote the
  transition matrix from step $t$ to step $t+s$ ($t \geq 0$, $s \geq 0$) by
  $$
    P_{t,s}[i,j] = \prob(X_{t+s}=j ~\vert~ X_t=i) \in \real^{S \times S},
  $$
  and the observation distribution matrix at step $t$ ($t \geq 0)$ by
  $$
    E_t[i,j] = \prob(Z_t=i ~\vert~ X_t=j) \in \real^{S \times S},
  $$
  For $t \geq 0$, $s \geq 0$ and $i \in \{1,\ldots,S\}$, denote the observation operator matrix
  $$
    O^{(t,s)}_{i} = \begin{pmatrix}
      E_{t+s}[i,1] & & 0 \\
      & \ddots & \\
      0 & & E_{t+s}[i,S]
  \end{pmatrix} P_{t,s} \in \real^{S \times S}.
  $$
  Then the joint probabilities of $\{Z_t\}_{t \geq 1}$ conditional on the initial underlying state $X_0$ satisfy
  \begin{equation} \label{forward_prob_compare}
    \prob\left( Z_{t_1}=i_1,\ldots,Z_{t_K}=i_K ~\vert~ X_0 \right) = \bm{1}_S^{\tp} O^{(t_{K-1},t_K-t_{K-1})}_{i_1} \cdots O^{(t_1,t_2-t_1)}_{i_1} O^{(0,t_1)}_{i_1} \bm{e}_{X_0}
  \end{equation}
  for state sequence $\{i_1,\ldots,i_K\}$ and increasing seqence of integers $0 \leq t_1 < t_2 < \ldots < t_K$.
\end{lemma}

In our setting with binary edges, we will have $S=2$, and time homogeneous observation distribution matrices $E_t = E$ for all $t \geq 0$.
The time inhomogeneous transition matrices for a given node pair $(i,j)$ will depend on the final state of edge $(i,j)$ in the underlying adjacency matrix $\bm{A}^{(K)}$.

\section{Theory for GMM estimation} \label{app:proofs_gmm}

\subsection{Inference with known \texorpdfstring{$\theta$}{theta}}
\label{subsec:theta_known}

Define
\begin{equation*}
x_k(\theta) = \expect_{\theta} \left\{ D_k(\Yvec_{12}) ~\big\vert~ \Anetk{1}_{12}=0 \right\}, \quad
  y_k(\theta) = \expect_{\theta} \left\{  D_k(\Yvec_{12}) ~\big\vert~ \Anetk{1}_{12}=1 \right\}.
\end{equation*}
Following the discussion in Section~\ref{sec:estimation}, a method of moments estimator for $\delta_1$ based on $\Ynetk{k}$ would solve
$$
  \binom{n}{2}^{-1} \sum_{i < j} D_k(\Yvec_{ij}) = \binom{n}{2}^{-1} \sum_{i < j} \Ynetk{k}_{ij}
  = \hat{\delta}_1^{(k)} y_k(\theta) + \left(1-\hat{\delta}_1^{(k)}\right)x_k(\theta).
$$
This estimator has a closed form
\begin{equation} \label{adjusted_density}
  \hat{\delta}_1^{(k)} = \binom{n}{2}^{-1} \sum_{i < j} \frac{\Ynetk{k}_{ij} - x_k(\theta)}{y_k(\theta) - x_k(\theta)}
\end{equation}
for $k=1,\ldots,K$.
If there are no errors and no evolution, we will have $x_k(\theta)=0$, $y_k(\theta)=1$, and \eqref{adjusted_density} will reduce to the usual edge density.

To produce a final estimator, we follow \citesupp{lavancier16generalB} to adaptively reweight these $K$ (correlated) unbiased estimators.
Denote the vector of unbiased estimators
\begin{equation}
  \hat{\bm{\delta}}_{1} = (\hat{\delta}_{1}^{(0)}, \ldots, \hat{\delta}_1^{(K)})
\end{equation}
An optimal estimator will minimize the MSE of the linear combination $\bm{w}^{\tp} \hat{\bm{\delta}}_{1}$
over coefficient vectors $\bm{w} \in \mathbb{R}^{K}$, subject to $\bm{w}^{\tp} \bm{1}_K = 1$.
\citesupp{lavancier16generalB} prove that the oracle linear combination is given by
\begin{equation*}
  \bm{w}_* = \bm{w}_*(\delta_1,\theta) = \frac{\left\{ \bm{\Sigma}^{(\delta)}(\delta_1,\theta) \right\}^{-1} \bm{1}_K}{\bm{1}_K^{\tp}\left\{ \bm{\Sigma}^{(\delta)}(\delta_1,\theta) \right\}^{-1} \bm{1}_K} \in \real^K,
\end{equation*}
where
$$
  \bm{\Sigma}^{(\delta)}(\delta_1,\theta) = \delta_1 \bm{\Sigma}_1^{(\delta)}(\theta) + (1 - \delta_1) \bm{\Sigma}_0^{(\delta)}(\theta),
$$
and
$$
  \bm{\Sigma}_s^{(\delta)}(\theta) = \cov_{\theta} \left\{ \left( \frac{\Ynetk{1}_{12} - x_1(\theta)}{y_1(\theta) - x_1(\theta)}, \ldots, \frac{\Ynetk{K}_{12} - x_K(\theta)}{y_K(\theta) - x_K(\theta)} \right) ~\bigg\vert~ \Anetk{1}_{12}=s \right\} \in \real^{K \times K}.
$$
The oracle linear combined estimator remains unbiased, with variance
\begin{equation*}
  \operatorname{Var}(\bm{w}_*^{\tp}  \hat{\bm{\delta}}_{1,K})
  = \bm{w}_*^{\tp} \left\{ \bm{\Sigma}^{(\delta)}(\delta_1,\theta) \right\}^{-1} \bm{w}_*
  = \frac{1}{\bm{1}_K^{\tp} \left\{ \bm{\Sigma}^{(\delta)}(\delta_1,\theta) \right\}^{-1} \bm{1}_K}
\end{equation*}

To find an estimator which is asymptotically efficient, we recommend to choose the weight vector iteratively, and evaluate
$$
  \hat{\delta}_{1,h} =  \bm{w}_*\left( \hat{\delta}_{1,h-1},\theta \right)^{\tp} \hat{\bm{\delta}}_{1}
$$
for $h=1,2,\ldots$ to convergence, initializing with $\delta_{1,0} = \hat{\delta}_1^{(1)}$, the estimator using only the first snapshot.
Asymptotic efficiency of the adaptively weighted estimator can be proven for $\hat{\delta}_{1,1}$, after just one iteration.
In the following result, we allow $\delta_1 = \delta_1(n)$ to depend on $n$, subject to the following mild condition.
\begin{assumption} \label{assump:delta_seq}
  $\delta_1(n) \in (0,1)$ uniformly over $n$, and the sequence converges, that is $\delta_1(n) \rightarrow \delta_1(\infty) \in (0,1)$.
\end{assumption}
Throughout, we will explicitly note when the limiting value $\delta_1(\infty)$ is referred to.
When $\delta_1$ is written, it should be interpreted as $\delta_1(n)$ for the appropriate choice of $n$.
\begin{proposition} \label{prop:delta_weighted}
  Suppose Assumption~\ref{assump:delta_seq} holds, and $(\delta_1,\alpha,\beta,\lambda,\mu) \in \Psi$.
  Then
  $$
      \binom{n}{2}^{1/2} \left( \bm{w}_*(\hat{\delta}_1^{(1)},\theta)^{\tp} \hat{\bm{\delta}}_{1} - \delta_1 \right)
      \indist \mathcal{N}\left( 0, \frac{1}{\bm{1}_K^{\tp} \left\{ \bm{\Sigma}^{(\delta)}(\delta_1(\infty),\theta) \right\}^{-1} \bm{1}_K} \right).
  $$
\end{proposition}

Other than some degenerate cases, typically $\bm{w}_*$ will have non-negative entries which sum to $1$ and decay as $k$ increases, since the edges become progressively less correlated with the initial snapshot.

\subsubsection{Proofs for inference with known \texorpdfstring{$\theta$}{theta}}

We begin with a supporting lemma, a central limit theorem for the empirical moments, averaged over the node pairs.

\begin{lemma} \label{lem:mstar_clt}
  Define $m^*$ as in Section~\ref{subsec:theta_unknown}, and suppose Assumption~\ref{assump:delta_seq} holds.
  Then
  \begin{align} \label{mstar_clt}
    &\binom{n}{2}^{-1/2}\left[ \sum_{i < j} m^*(\Yvec_{ij}) - \binom{n}{2} \left\{ \delta_1 \bm{m}^*_{\theta,1} + (1-\delta_1)\bm{m}^*_{\theta,0} \right\} \right] \nonumber \\
    \indist &\mathcal{N}\left( \bm{0}_{K+4}, \delta_1(\infty) \bm{\Sigma}^*_{\theta,1} + (1 - \delta_1(\infty)) \bm{\Sigma}^*_{\theta,0} \right)
  \end{align}
  as $n \rightarrow \infty$, where
  $$
    \bm{m}^*_{\theta,s} = \expect_{\theta} \left\{ m^*\left( \Yvec_{12}\right) ~\big\vert~ \Anetk{1}_{12}=s \right\} \in \real^{K+4},
  $$
  and
  $$
    \bm{\Sigma}^*_{\theta,s} = \cov_{\theta} \left\{ m^*\left( \Yvec_{12} \right) ~\big\vert~ \Anetk{1}_{12}=s \right\} \in \real^{(K+4) \times (K+4)}.
  $$
\end{lemma}

\begin{proof}
  For $s=0,1$, define
  $$
    N_s = N_s(n) = \binom{n}{2} \delta_1^s(1-\delta_1)^{1-s},
  $$
  and
  $$
    \mathcal{E}_s = \left\{ (i,j) : \Anetk{1}_{ij}=s \right\}.
  $$
  By assumption, $\lvert \mathcal{E}_s \rvert > 0$, and $\lvert \mathcal{E}_s \rvert \geq \binom{n}{2} \min\{ \delta_1(n), 1 - \delta_1(n) \} \rightarrow \infty$ as $n \rightarrow \infty$, since
  $$
    \min\{ \delta_1(n), 1 - \delta_1(n) \} \rightarrow \min\{ \delta_1(\infty), 1 - \delta_1(\infty) \} > 0.
  $$

  Then, for $s=0,1$, and conditional on $\Anetk{1}$,
  $$
    \{ m^*(\Yvec_{ij}) \}_{(i,j) \in \mathcal{E}_s}
  $$
  is a set of iid random vectors with mean $\bm{m}^*_{\theta,s}$ and covariance $\bm{\Sigma}^*_{\theta,s}$. By multivariate CLT,
  \begin{equation*} 
    N_s^{-1/2} \left\{ \sum_{(i,j) \in \mathcal{E}_s} m^*(\Yvec_{ij}) - N_s \bm{m}^*_{\theta,s} \right\}
    \indist \mathcal{N}\left( \bm{0}_{K+4}, \bm{\Sigma}^*_{\theta,s} \right).
  \end{equation*}
  Expanding, we get
  $$
    \left\{ \binom{n}{2} \delta_1^s(1-\delta_1)^{1-s} \right\}^{-1/2} \left\{ \sum_{(i,j) \in \mathcal{E}_s} m^*(\Yvec_{ij}) - \binom{n}{2} \delta_1^s(1-\delta_1)^{1-s} \bm{m}^*_{\theta,s} \right\}
    \indist \mathcal{N}\left( \bm{0}_{K+4}, \bm{\Sigma}^*_{\theta,s} \right).
  $$
  Multiplying by
  $$
    \left\{ \delta_1^s(1-\delta_1)^{1-s} \right\} \rightarrow \left\{ (\delta_1(\infty))^s(1-\delta_1(\infty))^{1-s} \right\}^{1/2}
  $$
  gives
  $$
    \hspace{-0.5cm}
    \binom{n}{2}^{-1/2} \left\{ \sum_{(i,j) \in \mathcal{E}_s} m^*(\Yvec_{ij}) - \binom{n}{2} \delta_1^s(1-\delta_1)^{1-s} \bm{m}^*_{\theta,s} \right\}
    \indist \mathcal{N}\left( \bm{0}_{K+4}, (\delta_1(\infty))^s(1-\delta_1(\infty))^{1-s} \bm{\Sigma}^*_{\theta,s} \right).
  $$
  by Slutsky's Theorem.
  Finally, adding the two independent convergent sequences gives the desired result \eqref{mstar_clt}, again by Slutsky's Theorem.
\end{proof}

\begin{remark} \label{rem:delta_asymptotic}
  For a given $K$, applying the lemma with $\tilde{K} = K+3$, the asymptotic distribution of the first $K$ coordinates (the local densities $(D_1,\ldots,D_K)$ used in Section~\ref{subsec:theta_known}) is multivariate normal.
  \begin{equation*}
    \binom{n}{2}^{1/2} \begin{pmatrix}
      \binom{n}{2}^{-1} \sum_{i < j} \Ynetk{1}_{ij} - \delta_1 y_1(\theta) - (1-\delta_1) x_1(\theta) \\
      \vdots \\
      \binom{n}{2}^{-1} \sum_{i < j} \Ynetk{K}_{ij} - \delta_1 y_K(\theta) - (1-\delta_1) x_K(\theta)
  \end{pmatrix} \indist \mathcal{N}\left( \bm{0}_{K}, \Sigma^{(D)}(\delta_1(\infty),\theta) \right).
  \end{equation*}
  where $\Sigma^{(D)}(\delta_1(\infty),\theta)$ is the top-left $K \times K$ block of $\delta_1(\infty) \bm{\Sigma}^*_{\theta,1} + (1 - \delta_1(\infty)) \bm{\Sigma}^*_{\theta,0}$.
\end{remark}

We can also estabish conditions on the error parameters such that $\Sigma^{(D)}(\delta_1(\infty),\theta)$ is positive definite.

\begin{lemma} \label{lem:density_pd}
  Suppose $\delta_1(\infty), \alpha, \beta \in (0,1)$. Then $\Sigma^{(D)}(\delta_1(\infty),\theta)$ is positive definite.
\end{lemma}

\begin{proof}
  Fix $s \in \{0,1\}$ and node pair $(1,2)$ without loss of generality.
  Fix an arbitrary vector $\bm{v} \in \real^K$ and denote
  $$
    \vec{\bm{Y}}_{12} = (\Ynetk{1}_{12},\ldots,\Ynetk{K}_{12})^{\tp}, \quad   \vec{\bm{A}}_{12} = (\Anetk{1}_{12},\ldots,\Anetk{K}_{12})^{\tp}.
  $$
  Note that we can write each entry as
  $$
    \Ynetk{k}_{12} = \Anetk{k}_{12} + E_{k}(\Anetk{k}_{12}),
  $$
  where
  $$
    E_{k}(\Anetk{k}_{12}) \sim \begin{cases} \operatorname{Bernoulli}(\alpha) \quad &\Anetk{k}_{12}=0, \\
    -\operatorname{Bernoulli}(\beta) \quad &\Anetk{k}_{12}=1. \end{cases}
  $$
  and the $E_k$ variables are mutually independent over $k$.
  \begin{align*}
    \var(\bm{v}^{\tp} \vec{\bm{Y}}_{12} ~\vert~ \Anetk{1}_{12}=s) &\geq \expect\{ \var(\bm{v}^{\tp} \vec{\bm{Y}}_{12} ~\vert~ \vec{\bm{A}}_{12},\Anetk{1}_{12}=s) \} \\
    &= \expect\{ \var(\bm{v}^{\tp} \{ \Anetk{k}_{12} + E_{k}(\Anetk{k}_{12})\} ~\vert~ \vec{\bm{A}}_{12},\Anetk{1}_{12}=s) \} \\
    &= \expect\{ \var(\bm{v}^{\tp} E_{k}(\Anetk{k}_{12}) ~\vert~ \vec{\bm{A}}_{12},\Anetk{1}_{12}=s) \} \\
    &\geq \min\{ \alpha(1-\alpha), \beta(1-\beta) \} \bm{v}^{\tp}\bm{v}.
  \end{align*}
  This variance is zero if and only if $\bm{v} = \bm{0}_K$, thus for each $s$ the top-left $K \times K$ block of $\bm{\Sigma}^*_{\theta,s}$ is positive definite.
  The result is complete since the convex combination of positive definite matrices is positive definite.
\end{proof}

In the following lemma, we establish the joint asymptotic normality of the local density estimators used in Section~\ref{subsec:theta_known}.

\begin{lemma} \label{lem:delta_asymptotic}
  Under Assumption~\ref{assump:delta_seq},
  $$
    \alpha + \beta \neq 1, \quad \lambda + \mu \neq 1,
  $$
  we have
  \begin{equation*}
    \binom{n}{2}^{1/2} \left( \hat{\bm{\delta}}_{1,K} - \delta_1 \bm{1}_K \right) \indist \mathcal{N}\left( \bm{0}_{K}, \Sigma^{(\delta)}(\delta_1(\infty),\theta) \right).
  \end{equation*}
  Moreover,
  $$
    \Sigma^{(\delta)}(\delta_1(\infty),\theta) = S_K(\theta)^{-1} \Sigma^{(D)}(\delta_1(\infty),\theta) S_K(\theta)^{-1},
  $$
  where $S_K(\theta)$ is the diagonal matrix
  $$
    \begin{pmatrix}
      y_1(\theta) - x_1(\theta) & & 0 \\
      & \ddots & \\
      0 & & y_K(\theta) - x_K(\theta)
    \end{pmatrix}.
  $$
\end{lemma}

\begin{proof}
  Starting from Remark~\ref{rem:delta_asymptotic}, we apply the linear function
  $$
    g(z_1,\ldots,z_K) = \begin{pmatrix}
      \frac{z_1 - x_1(\theta)}{y_1(\theta) - x_1(\theta)} \\
      \vdots \\
      \frac{z_K - x_K(\theta)}{y_K(\theta) - x_K(\theta)}
  \end{pmatrix}.
  $$
  Some algebra shows that under the reparameterized error and evolution models, we can write
  \begin{align*}
    x_k(\theta) &= \alpha(1-\tau) + (1-\beta)\tau - (1 - \beta - \alpha)\tau \gamma^{k-1}, \\
    y_k(\theta) &= \alpha(1-\tau) + (1-\beta)\tau + (1 - \beta - \alpha)(1 -\tau) \gamma^{k-1}
  \end{align*}
  for $k=1,\ldots,K$.
  $$
    y_k(\theta) - x_k(\theta) = (1 - \beta - \alpha)\gamma^{k-1} =  (1 - \beta - \alpha)(1 - \lambda - \mu)^{k-1}.
  $$
  Thus   $\alpha + \beta \neq 1$ and $\lambda + \mu \neq 1$ is a sufficient condition for $y_k(\theta) - x_k(\theta) \neq 0$ for all $k=1,\ldots,K$, so this linear function is well-defined.
  The result follows by definition of $\hat{\bm{\delta}}_{1,K}$, and since
  $$
    S_K(\theta)^{-1} = \nabla_{\bm{z}} g(z_1,\ldots,z_K).
  $$
\end{proof}

We are now ready to prove the main result, Proposition~\ref{prop:delta_weighted}.

\begin{proof}
  By \citesupp{lavancier16generalB}, Proposition 3.3, it is sufficient to show that
  \begin{equation} \label{lavancier_sufficient}
    \bm{\Sigma}^{(\delta)}(\hat{\delta}_1^{(1)},\theta) \left\{ \bm{\Sigma}^{(\delta)}(\delta_1(\infty),\theta) \right\}^{-1} \inprob I_K.
  \end{equation}
  Expand these definitions as
  $$
    \left\{ \hat{\delta}_1^{(1)} \bm{\Sigma}^{(\delta)}_{\theta,1} + (1 - \hat{\delta}_1^{(1)}) \bm{\Sigma}^{(\delta)}_{\theta,0} \right\} \left\{ \delta_1(\infty) \bm{\Sigma}^{(\delta)}_{\theta,1} + (1 - \delta_1(\infty)) \bm{\Sigma}^{(\delta)}_{\theta,0} \right\}^{-1}
  $$
  By Slutsky's Theorem and consistency of $\hat{\delta}_1^{(1)}$ it follows that
  $$
    \hat{\delta}_1^{(1)} \bm{\Sigma}^{(\delta)}_{\theta,1} + (1 - \hat{\delta}_1^{(1)}) \bm{\Sigma}^{(\delta)}_{\theta,0}
    \inprob \delta_1(\infty) \bm{\Sigma}^{(\delta)}_{\theta,1} + (1 - \delta_1(\infty)) \bm{\Sigma}^{(\delta)}_{\theta,0}.
  $$
  Then \eqref{lavancier_sufficient} follows since
  $$
    \delta_1(\infty) \bm{\Sigma}^{(\delta)}_{\theta,1} + (1 - \delta_1(\infty)) \bm{\Sigma}^{(\delta)}_{\theta,0}
  $$
  is positive definite, so that right multiplication by its inverse is a well-defined continuous mapping.
\end{proof}

\subsection{Inference with unknown \texorpdfstring{$\theta$}{theta}}

In this section denote the true parameter by
$$
  \psi_0 = (\delta_1,\theta) = (\delta_1,\alpha,\beta,\lambda,\mu) \in \bar{\Psi},
$$
where recall we have modified the parameter space to be compact,
$$
  \bar{\Psi} = \bar{\Psi}(\xi) = \{ \psi \in [\xi,1-\xi]^5 : \psi_2 + \psi_3 \leq 1-\xi,~\psi_4+\psi_5 \leq 1-\xi \}
$$
for a small constant $\xi > 0$.
Also recall that we implicitly allow $\delta_1 = \delta_1(n)$ to depend on $n$, and write the limiting parameter as
$$
  \psi_0(\infty) = (\delta_1(\infty),\alpha,\beta,\lambda,\mu) \in \bar{\Psi}.
$$
Throughout, we will make Assumption~\ref{assump:delta_seq_compact} from Section~\ref{subsec:theta_unknown} on the sequence of true parameters.

We will analyze the GMM estimator using a generic mapping $m$, later these general results can be applied to either the full set of moments $m^*$ or the initializing set $m^{(\mathrm{init})}$.
Define the GMM objective
$$
  f(\psi) = \{\hat{\bm{m}} - \bm{m}(\psi)\}^{\tp} \bm{W} \{\hat{\bm{m}} - \bm{m}(\psi)\}
$$
where $\hat{\bm{m}}$ are the empirical edge-averaged moments, and $\bm{m}(\psi) = \bm{m}(\delta_1,\theta)$ are the population counterparts for a generic mapping $m$ from $\{0,1\}^K \rightarrow \real^p$.
Note that $f$ implicitly depends on the data through $\hat{\bm{m}}$ as well as $n$ through $\delta_1(n)$.

We make an assumption analogous to Assumption~\ref{assump:global_identification} for this generic mapping.
\begin{assumption} \label{assump:global_identification_generic}
The mapping $\bm{m}(\psi)$ is an injective function on $\Psi$.
\end{assumption}

To make the result applicable to adaptive weighting, we also allow a stochastic sequence of weighting matrices $\bm{W}_n$ which depend on $n$, subject to the following assumption.

\begin{assumption} \label{assump:w_conditions}
  The sequence of weighting matrices $\{\bm{W}\}_{n=1}^{\infty}$ satisfies
  \begin{align}
    &\bm{W}_n \inprob \bm{W}_{\infty} \in \real^{p \times p} \label{w_condition1} \\
    &0 < c' \leq \lambda_{\mathrm{min}}(\bm{W}_n) \leq \lambda_{\mathrm{max}}(\bm{W}_n) \leq C' < \infty \label{w_condition2}
  \end{align}
  uniformly over $n$, where $c'$ and $C'$ are constants which are free of $n$, but may depend on $\xi$.
\end{assumption}

Note that throughout the section, we omit the subscript $n$ on $\bm{W}$, instead adding the subscript $\infty$ to the probability limit of the weighting matrices.

\begin{remark} \label{rem:m_polynomial}
  By Lemma~\ref{lem:gmm_computation}, all the entries of $\bm{m}(\psi)$ are polynomials of the entries of $\psi$.
  Viewing one entry as the argument, this is a polynomial with bounded coefficients (in the unit interval) over a bounded domain (the unit interval).
  Thus $\bm{m}$ is smooth ($\infty$-times continuously differentiable), and all the entries of all of its derivatives are uniformly bounded (over $\psi$) by a constant which may depend on $K$.
\end{remark}

Denote the derivative matrix
$$
  \bm{Dm}(\psi) = \pd{\bm{m}(\psi)}{\psi} \in \real^{p \times 5},
$$
and second derivative tensor $\mathcal{D}^2\bm{m}(\psi) \in \real^{5 \times 5 \times p}$ with slices
$$
  \pd{^2 \bm{m}_{\ell}(\psi)}{\psi \partial \psi^{\tp}} \in \real^{5 \times 5}, \quad \ell=1,\ldots,p.
$$

$$
  \pd{f(\psi)}{\psi} = -2 \bm{Dm}(\psi)^{\tp} \bm{W} \left\{ \hat{\bm{m}} - \bm{m}(\psi) \right\} \in \real^5
$$
$$
  \pd{^2f(\psi)}{\psi \partial \psi^{\tp}} = 2 \bm{Dm}(\psi)^{\tp} \bm{W} \bm{Dm}(\psi) - 2 \mathcal{D}^2\bm{m}(\psi) \bar{\times}_3 \left\{ \bm{W} (\hat{\bm{m}} - \bm{m}(\psi)) \right\} \in \real^{5 \times 5}.
$$

We define the estimator $\hat{\psi}$ as the minimizer of $f(\psi)$.
By differentiability of $f$, this minimizer will solve the system of non-linear equations
\begin{equation} \label{first_order_condition}
 \bm{0}_5 = \bm{Dm}(\hat{\psi})^{\tp} \bm{W} \left\{ \hat{\bm{m}} - \bm{m}(\hat{\psi}) \right\}.
\end{equation}

We first establish consistency of $\hat{\psi}$ for $\psi_0(\infty)$.

\begin{proposition} \label{prop:theta_consistent}
  Under Assumptions~\ref{assump:delta_seq_compact}, \ref{assump:global_identification_generic}, and \ref{assump:w_conditions},
  $$
    \hat{\psi} \inprob \psi_0(\infty)
  $$
\end{proposition}

\begin{proof}
  This result follows by an application of \citesupp{newey94largeB}, Theorem 2.1.
  Define the population analog of the GMM objective function,
  $$
    f_0(\psi) = (\bm{m}(\psi_0(\infty)) - \bm{m}(\psi))^{\tp} \bm{W}_{\infty} (\bm{m}(\psi_0(\infty)) - \bm{m}(\psi)).
  $$
  Note that $f_0(\psi) \geq 0$ for any $\psi \in \bar{\Psi}$. By Assumption~\ref{assump:global_identification}, and since $\bm{W}_{\infty}$ is positive-definite, $f_0(\psi)$ is uniquely minimized at $\psi = \psi_0(\infty)$.
  We also have that $\Psi$ is compact by construction, and $f_0(\psi)$ is continuous by Remark~\ref{rem:m_polynomial}.

  All that remains is to show that $f(\psi)$ converges to $f_0(\psi)$ uniformly over $\Psi$.
  \begin{align*}
    \lvert f(\psi) - f_0(\psi) \rvert &= \left\lvert \hat{\bm{m}}^{\tp} \bm{W} \hat{\bm{m}} - \{\bm{m}(\psi_0(\infty))\}^{\tp}\bm{W}_{\infty} \bm{m}(\psi_0(\infty)) - 2\bm{m}(\psi)^{\tp} \bm{W} \hat{\bm{m}} + 2\bm{m}(\psi)^{\tp} \bm{W}_{\infty} \bm{m}(\psi_0^{\infty}) \right\rvert \\
    &\leq \lvert \hat{\bm{m}}^{\tp} ( \bm{W} - \bm{W}_{\infty} ) \hat{\bm{m}} \rvert + \lvert \hat{\bm{m}}^{\tp} \bm{W}_{\infty} \hat{\bm{m}} - \{\bm{m}(\psi_0(\infty))\}^{\tp}\bm{W}_{\infty} \bm{m}(\psi_0(\infty)) \rvert + \cdots \\
    &\cdots + \lvert 2 \bm{m}(\psi)^{\tp} (\bm{W} - \bm{W}_{\infty}) \hat{\bm{m}} \rvert + \lvert 2 \bm{m}(\psi)^{\tp} \bm{W}_{\infty} \{ \hat{\bm{m}} - \bm{m}(\psi_0(\infty)) \} \rvert.
  \end{align*}
  Note that
  $$
    \sup_{\psi \in \bar{\Psi}} \lVert \bm{m}(\psi) \rVert_2 \leq C < \infty
  $$
  by Remark~\ref{rem:m_polynomial}, where $C$ is a constant that is free of $\psi$.
  Thus, by Cauchy-Schwarz inequality,
  \begin{align*}
    \sup_{\psi \in \bar{\Psi}} \lvert f(\psi) - f_0(\psi) \rvert &\leq \lvert \hat{\bm{m}}^{\tp} ( \bm{W} - \bm{W}_{\infty} ) \hat{\bm{m}} \rvert + \lvert \hat{\bm{m}}^{\tp} \bm{W}_{\infty} \hat{\bm{m}} - \{\bm{m}(\psi_0(\infty))\}^{\tp}\bm{W}_{\infty} \bm{m}(\psi_0(\infty)) \rvert + \cdots \\
    &\cdots + 2C \lVert (\bm{W} - \bm{W}_{\infty}) \hat{\bm{m}} \rVert_2 + 2C \lVert \bm{W}_{\infty} \{ \hat{\bm{m}} - \bm{m}(\psi_0(\infty)) \} \rVert_2.
  \end{align*}
  The first and third terms converge in probability to zero by \eqref{w_condition1}, and since $\lVert \hat{\bm{m}} \rVert_2$ is (a.s.) bounded.
  The second and fourth terms converge to zero in probability by Lemma~\ref{lem:mstar_clt}, and \eqref{w_condition2}.
  Thus, we verify all the conditions of \citesupp{newey94largeB}, Theorem 2.1, and conclude $\hat{\psi} \inprob \psi_0(\infty)$.
\end{proof}

The following central limit theorem will be used to justify the asymptotic normality of $\hat{\psi}$.

\begin{lemma} \label{lem:gmm_clt}
  Under Assumptions~\ref{assump:delta_seq_compact} and \ref{assump:w_conditions},
  $$
    \binom{n}{2}^{1/2} \bm{Dm}(\psi_0)^{\tp} \bm{W} \left\{ \hat{\bm{m}} - \bm{m}(\psi_0) \right\}
    \indist \mathcal{N}\left( \bm{0}_5~,~ \widetilde{\bm{\Sigma}}^{(\psi)}(\psi_0(\infty)) \right),
  $$
  where $\psi_0(\infty) = (\delta_1(\infty),\theta)$
  $$
    \widetilde{\bm{\Sigma}}^{(\psi)}(\psi_0(\infty))= \bm{Dm}(\psi_0(\infty))^{\tp} \bm{W}_{\infty} \left\{ \delta_1(\infty) \bm{\Sigma}_1(\theta) + (1 - \delta_1(\infty)) \bm{\Sigma}_0(\theta) \right\} \bm{W}_{\infty} \bm{Dm}(\psi_0(\infty))
  $$
  with $\psi_0(\infty) = (\delta_1(\infty),\theta)$, and $\bm{\Sigma}_s(\theta)$ defined as in \eqref{sigma_s} for $s=0,1$.
\end{lemma}

\begin{proof}
  Note that
  $$
    \binom{n}{2}^{1/2} \bm{Dm}(\psi_0)^{\tp} \bm{W} \left\{ \hat{\bm{m}} - \bm{m}(\psi_0) \right\} = \binom{n}{2}^{-1/2} \sum_{i < j} \bm{Dm}(\psi_0)^{\tp} \bm{W} \left\{ m(\Yvec_{ij}) - \bm{m}(\psi_0) \right\}
  $$
  Then, beginning from Lemma~\ref{lem:mstar_clt}, apply the linear transformation
  $$
    \bm{Dm}(\psi_0(\infty))^{\tp} \bm{W}_{\infty} \in \real^{5 \times p},
  $$
  which implies
  $$
    \binom{n}{2}^{-1/2} \sum_{i < j} \bm{Dm}(\psi_0(\infty))^{\tp} \bm{W}_{\infty} \left\{ m(\Yvec_{ij}) - \bm{m}(\psi_0) \right\}
  \indist \mathcal{N}\left( \bm{0}_5~,~ \widetilde{\bm{\Sigma}}^{(\psi)}(\psi_0(\infty)) \right).
  $$
  The result follows by Slutsky's Theorem, since
  \begin{equation*}
    \binom{n}{2}^{-1/2} \sum_{i < j} \left\{ \bm{W} \bm{Dm}(\psi_0) - \bm{W}_{\infty} \bm{Dm}(\psi_0(\infty)) \right\}^{\tp} \left\{ m(\Yvec_{ij}) - \bm{m}(\psi_0) \right\} = o_p(1).
  \end{equation*}
  This follows since
  $$
    \bm{W} \bm{Dm}(\psi_0) - \bm{W}_{\infty} \bm{Dm}(\psi_0(\infty)) = o_p(1),
  $$
  by \eqref{w_condition1} and continuity of $\bm{Dm}$; and
  $$
    \binom{n}{2}^{-1/2} \sum_{i < j} \left\{ m(\Yvec_{ij}) - \bm{m}(\psi_0) \right\} = O_p(1)
  $$
  by Lemma~\ref{lem:mstar_clt}.
\end{proof}

The main result of this section is a central limit theorem for a GMM estimator $\hat{\psi}$. For this result, we require a generic version of the local identification Assumption~\ref{assump:local_identification_body} on the derivative matrix of the edge-averaged moments.
\begin{assumption} \label{assump:local_identification}
  The derivative matrix $\bm{Dm}(\psi)$ has full column rank for any $\psi \in \bar{\Psi}$.
\end{assumption}

\begin{proposition} \label{prop:theta_limiting}
  Under Assumptions~\ref{assump:delta_seq_compact}, \ref{assump:global_identification_generic}, \ref{assump:w_conditions}, and \ref{assump:local_identification},
  $$
    \binom{n}{2}^{1/2} (\hat{\psi} - \psi) \indist \mathcal{N}\left(0 , \bm{\Sigma}^{(\psi)}(\psi_0(\infty))\right),
  $$
  where
  \begin{align*}
    &\bm{\Sigma}^{(\psi)}(\psi_0(\infty)) \\ = &\left\{ \bm{Dm}(\psi_0(\infty))^{\tp} \bm{W}_{\infty} \bm{Dm}(\psi_0(\infty))\right\}^{-1} \widetilde{\bm{\Sigma}}^{(\psi)}(\psi_0(\infty)) \left\{ \bm{Dm}(\psi_0(\infty))^{\tp} \bm{W}_{\infty} \bm{Dm}(\psi_0(\infty))\right\}^{-1}.
  \end{align*}
\end{proposition}

\begin{proof}
  Begin with an application of the mean value theorem to the coordinate functions of \eqref{first_order_condition} between $\hat{\psi}$ and $\psi_0$, giving
  \begin{equation} \label{gmm_expansion}
      \bm{0}_5 = \bm{Dm}(\psi_0)^{\tp} \bm{W} \left\{ \hat{\bm{m}} - \bm{m}(\psi_0) \right\} + D_2(\tilde{\psi}) (\hat{\psi} - \psi_0)
  \end{equation}
  where
  $$
    D_2(\psi) = \mathcal{D}^2\bm{m}(\psi) \bar{\times}_3 \left\{ \bm{W}_{\infty} (\hat{\bm{m}} - \bm{m}(\psi)) \right\} - \bm{Dm}(\psi)^{\tp} \bm{W}_{\infty} \bm{Dm}(\psi) \in \real^{5 \times 5}.
  $$
  and $\tilde{\psi}$ is on the line segment between $\psi_0$ and $\hat{\psi}$.

  We first show that $D_2(\tilde{\psi})$ converges in probability to limiting matrix which is free of $\tilde{\psi}$.
  \begin{scriptsize}
  \begin{align*}
    & \lVert D_2(\tilde{\psi}) - D_2(\psi_0(\infty)) \rVert_2 = \underbrace{\lVert \left\{ \mathcal{D}^2\bm{m}(\tilde{\psi}) - \mathcal{D}^2\bm{m}(\psi) \right\} \bar{\times}_3 \bm{W} \hat{\bm{m}} \rVert_2}_{(I)} + \cdots \\
    &\hspace{-0.95cm} \cdots + \underbrace{\lVert \mathcal{D}^2\bm{m}(\tilde{\psi}) \bar{\times}_3 \bm{W} \bm{m}(\tilde{\psi}) - \bm{Dm}(\tilde{\psi})^{\tp} \bm{W} \bm{Dm}(\tilde{\psi})
    - \mathcal{D}^2\bm{m}(\psi_0(\infty)) \bar{\times}_3 \bm{W} \bm{m}(\psi_0(\infty)) + \bm{Dm}(\psi_0(\infty))^{\tp} \bm{W} \bm{Dm}(\psi_0(\infty))\rVert_2}_{(II)}.
  \end{align*}
  \end{scriptsize}
  The second term is a difference of a continuous function of $\tilde{\psi}$ and $\psi_0(\infty)$, thus $(II) \inprob 0$ by Proposition~\ref{prop:theta_consistent} and \eqref{w_condition2}. 
  For the first term, by triangle inequality
  $$
    (I) \leq \sum_{\ell=1}^p \lvert (\bm{W}\hat{\bm{m}})_{\ell} \rvert \left\lVert \pd{^2 m_{\ell}}{\psi \partial \psi^{\tp}}(\tilde{\psi}) - \pd{^2 m_{\ell}}{\psi \partial \psi^{\tp}}(\psi_0(\infty)) \right\rVert_2.
  $$
  Then $(I) \inprob 0$ since the second derivatives are continuous, and the entries of $\bm{W}\hat{\bm{m}}$ are (a.s.) bounded: $\lVert \bm{W} \rVert_2$ is uniformly bounded by \eqref{w_condition2}, and $\hat{\bm{m}} \in [0,1]^p$, as they are empirical proportions.
  Thus,
  \begin{equation} \label{D2_convergence}
    \lVert D_2(\tilde{\psi}) - D_2(\psi_0(\infty)) \rVert_2 \inprob 0.
  \end{equation}

  Next, note that
  \begin{equation} \label{D2_randompart}
    \lVert \mathcal{D}^2\bm{m}(\psi_0(\infty)) \bar{\times}_3 \left\{ \bm{W} (\hat{\bm{m}} - \bm{m}(\psi_0(\infty)) \right\} \rVert_2 \inprob 0
  \end{equation}
  again by Proposition~\ref{prop:theta_consistent}, and since the entries of $\mathcal{D}^2\bm{m}$ are bounded by Remark~\ref{rem:m_polynomial}.
  Combining \eqref{D2_convergence}, \eqref{D2_randompart}, and finally \eqref{w_condition1}, we get
  $$
    D_2(\tilde{\psi}) = -\bm{Dm}(\psi_0(\infty))^{\tp} \bm{W}_{\infty} \bm{Dm}(\psi_0(\infty)) + o_p(1).
  $$

  Thus, \eqref{gmm_expansion} becomes
  $$
    \left\{ -\bm{Dm}(\psi_0(\infty))^{\tp} \bm{W}_{\infty} \bm{Dm}(\psi_0(\infty)) + o_p(1) \right\} (\hat{\psi} - \psi_0) = - \bm{Dm}(\psi_0)^{\tp} \bm{W} \left\{ \hat{\bm{m}} - \bm{m}(\psi_0) \right\}.
  $$
  By Assumption~\ref{assump:local_identification}, the leading term is invertible.
  Rearranging, and applying Lemma~\ref{lem:gmm_clt}, we get the desired limiting distribution for
  $$
    \binom{n}{2}^{1/2}(\hat{\psi} - \psi_0).
  $$
\end{proof}

\begin{remark} \label{rem:covariance_estimation}
  Proposition~\ref{prop:theta_limiting} justifies asymptotic normality for $\hat{\psi}$, however its asymptotic covariance matrix depends on the unknown parameters.
  It is natural to estimate this covariance parametrically by plugging in the estimate of $\psi$.
  Normalizing by the estimated covariance will still lead to consistent inference for $\psi$.
\end{remark}

\begin{remark} \label{rem:covariance_optimality}
  Proposition~\ref{prop:theta_limiting} can be used to find the weighting matrix which produces an asymptotically efficient estimator of $\hat{\psi}$.
  In particular, the minimial covariance matrix in terms of the partial order induced by is achieved for a sequence of weighting matrices which satisfy Assumption~\ref{assump:w_conditions}, and converge to
  $$
    \bm{W}_{\infty} = \left\{ \delta_1(\infty) \bm{\Sigma}_1(\theta) + (1 - \delta_1(\infty)) \bm{\Sigma}_0(\theta) \right\}^{-1}.
  $$
  The resulting asymptotic covariance matrix of $\hat{\psi}$ is
  $$
    \bm{Dm}(\psi_0)^{\tp} \left\{ \delta_1(\infty) \bm{\Sigma}_1(\theta) + (1 - \delta_1(\infty)) \bm{\Sigma}_0(\theta) \right\}^{-1} \bm{Dm}(\psi_0),
  $$
  where $\psi_0 = (\delta_0,\theta)$.
\end{remark}

We now complete this section by using the above results to prove Proposition~\ref{prop:gmm_adaptive_clt}.

\begin{proof}
  Recall that Assumption~\ref{assump:delta_seq_compact} is assumed throughout.

  Our goal will be to apply Proposition~\ref{prop:gmm_adaptive_clt} with the mapping $m^*$ and adaptive weighting matrices $\widehat{\bm{W}}$.
  The result will follow if we can verify the assumptions of Proposition~\ref{prop:gmm_adaptive_clt}.

  We first apply Proposition~\ref{prop:theta_consistent} with the mapping $m^{(\mathrm{init})}$ and fixed weighting matrices $\bm{W}_n = \bm{I}_8$.
  Assumption~\ref{assump:global_identification_generic} follows from Assumption~\ref{assump:global_identification}.
  Assumption~\ref{assump:w_conditions} holds vacuously.
  Thus, by Proposition~\ref{prop:theta_consistent},
  $$
    (\hat{\delta}_1^{(\mathrm{init})},\hat{\theta}^{(\mathrm{init})}) \inprob (\delta_1(\infty),\theta).
  $$

  By Remark~\ref{rem:m_polynomial}, $\widehat{\bm{W}}$ is a continuous function  of $(\hat{\delta}_1^{(\mathrm{init})},\hat{\theta}^{(\mathrm{init})})$, so
  \begin{equation} \label{w_limit}
    \widehat{\bm{W}} \inprob \left\{ \delta_1(\infty) \bm{\Sigma}_1(\theta) + (1 - \delta_1(\infty)) \bm{\Sigma}_0(\theta) \right\}^{-1}.
  \end{equation}
  Moreover, \eqref{w_condition2} follows from Assumption~\ref{assump:cov_spectrum}.
  Thus, the sequence of adaptive weighting matrices $\widehat{\bm{W}}$ satisfy Assumption~\ref{assump:w_conditions}, with optimal probability limit given by \eqref{w_limit} (see Remark~\ref{rem:covariance_optimality}).

  We complete the verification of the assumptions of Proposition~\ref{prop:gmm_adaptive_clt}.   Assumption~\ref{assump:global_identification_generic} follows from Assumption~\ref{assump:global_identification},
  and Assumption~\ref{assump:local_identification} follows from Assumption~\ref{assump:local_identification_body}.
  This completes the verification of the assumptions of Proposition~\ref{prop:gmm_adaptive_clt}.
\end{proof}

\section{Theory for higher-order subgraph densities} \label{app:ho_theory}

\subsection{Inference with known \texorpdfstring{$\theta$}{theta}}

To formally state our results about higher-order subgraph densities, we introduce notation originally defined by \citesupp{chang22estimationB}. Recall that we are working with a subgraph $H$ on $L$ edges, based on the configuration of the vertex pairs in each $\bm{v} = (v_1,\ldots,v_L) \in \mathcal{V}$.

For some $s \in \{1,\ldots,L-1\}$ and $1 \leq \ell_1 < \cdots < \ell_s \leq L$, define
$$
    \mathcal{G}_{\ell_1,\ldots,\ell_s}(\bm{v}) = \{(u_1,\ldots,u_s) : (v_1,\ldots,v_{\ell_1-1},u_1,v_{\ell_1+1},\ldots,v_{\ell_s-1},u_s,v_{\ell_s+1},\ldots,v_L) \in \mathcal{V} \},
$$
and 
$$
    \aleph_{\mathcal{V}}(s) = \max_{\bm{v} \in \mathcal{V}} \max_{\ell_1 < \cdots < \ell_s} \lvert \mathcal{G}_{\ell_1,\ldots,\ell_s}(\bm{v}) \rvert.
$$
Roughly, for $1 \leq s \leq L-1$, $\aleph_{\mathcal{V}}(L-s)$ enumerates the (maximal) number of subgraphs that can be constructed, after fixing exactly $s$ of the vertex pairs. 
The following is a restatement of \citesupp{chang22estimationB}, Proposition 2 in our dynamic setting. The proof is identical to theirs.

\begin{corollary} \label{cor:ho_order}
  Suppose
  $y_k(\theta) \neq x_k(\theta)$ for all $k=1,\ldots,K$.
  Then
  \[
    \lvert \widetilde{C}_H^{(k)} - C_H \rvert = O_{\mathbb{P}}\left( \sqrt{\frac{\aleph_{\mathcal{V}}(L-1)}{\lvert \mathcal{V} \rvert}} \right)
  \]
\end{corollary}

We also require some regularity conditions on the index sets described by $\mathcal{V}$.

\begin{assumption}[\citesupp{chang22estimationB}, Assumption 2] \label{assump:c22}
    ~~
    \begin{enumerate}
        \item $\aleph_{\mathcal{V}}(s) / \aleph_{\mathcal{V}}(L-1) \rightarrow 0$ for any $1 \leq s \leq L-2$.
        \item $\max_{\bm{v} \in \mathcal{V}} \max_{\ell_1 < \cdots < \ell_{L-1}} \lvert \mathcal{G}_{\ell_1,\ldots,\ell_{L-1}}(\bm{v}) \rvert \asymp \min_{\bm{v} \in \mathcal{V}} \min_{\ell_1 < \cdots < \ell_{L-1}} \lvert \mathcal{G}_{\ell_1,\ldots,\ell_{L-1}}(\bm{v}) \rvert$.
    \end{enumerate}
\end{assumption}

Assumption~\ref{assump:c22} says that (i) the count with more than one edge fixed is asymptotically dominated by the count with exactly one edge fixed, and (ii) up to constant factors, this maximum does not depend on which vertex pairs are fixed.
To replace $\widetilde{C}_H^{(k)}$ by the linearized $S_H^{(k)}$ after centering and scaling, we also require that $n^2 \aleph_{\mathcal{V}}(L-1) \lesssim \lvert \mathcal{V} \rvert$, which roughly says that we can enumerate the index sets in $\mathcal{V}$ by fixing any one vertex pair, then the cardinality of the set of index sets with that fixed vertex pair is $\aleph_{\mathcal{V}}(L-1)$, up to constant factors.
The following is a restatement of \citesupp{chang22estimationB}, Proposition 3 in our dynamic setting. The proof is identical to theirs.

\begin{corollary} \label{cor:ho_S_approx}
  Suppose Assumption~\ref{assump:c22} holds,
  \[
      \aleph_{\mathcal{V}}(L-1) / \lvert \mathcal{V} \rvert \asymp n^{-2} 
  \]
  and $y_k(\theta) \neq x_k(\theta)$.
  Then we have that
  \[
    \binom{n}{2}^{1/2} (\widetilde{C}_H^{(k)} - C_H) = S_H^{(k)} + o_{\mathbb{P}}(1).
  \]
\end{corollary}

We next prove Proposition~\ref{prop:boot_sampler}, which justifies the choice of bootstrap distribution for $\vec{\bm{Y}}^{\dagger}_{ij}$.

\begin{proof}
Let $\Delta(2^K)$ denote the probability simplex on $\{0,1\}^K$.
Without loss of generality, let $\pi^{(0)} \in \Delta(2^K)$ denote the distribution of $\vec{\bm{Y}}_{12}$ for $\bm{A}_{12}^{(1)}=0$ and $\pi^{(1)} \in \Delta(2^K)$ denote the distribution of $\vec{\bm{Y}}_{12}$ for $\bm{A}_{12}^{(1)}=1$.
Note that $\pi^{(0)}$ and $\pi^{(1)}$ are fully parameterized by the error and evolution parameters $\theta = (\alpha,\beta,\lambda,\mu)$.
Thus $\bm{\Sigma}^{(s)} = \operatorname{Cov}_{\pi^{(s)}}(\Yvec)$
for $s=0,1$, which matches the definition from Section~\ref{sec:ho}.

For each pair of snapshot indices $k_1 \leq k_2$, the conditional covariance constraint \eqref{ho_bootcov} can be rewritten as
\begin{align*}
  &\sum_{\bm{y} \in \{0,1\}^K} \operatorname{Cov}(\mathcal{B}(\bm{y})^{(k_2)},\mathcal{B}(\bm{y})^{(k_2)})\{ (1-\bm{A}_{12}^{(1)})\pi^{(0)}_{\bm{y}} + \bm{A}_{12}^{(1)}\pi^{(1)}_{\bm{y}} \} \\
  = &(1-\bm{A}_{12}^{(1)})\left[ \bm{\Sigma}^{(0)} \right]_{k_1k_2} + \bm{A}_{12}^{(1)}\left[ \bm{\Sigma}^{(1)} \right]_{k_1k_2}.
\end{align*}
Expanding as a linear function of $\bm{A}^{(1)}$ leads to two sets of moment constraints on the bootstrapped edge sequence $\mathcal{B}(\bm{y})$:
\begin{align} \label{ho_bootcov_constraints}
  \sum_{\bm{y} \in \{0,1\}^K} \operatorname{Cov}(\mathcal{B}(\bm{y})^{(k_1)},\mathcal{B}(\bm{y})^{(k_2)})\pi^{(0)}_{\bm{y}} &= \left[ \bm{\Sigma}^{(0)} \right]_{k_1k_2} \nonumber \\
  \sum_{\bm{y} \in \{0,1\}^K} \operatorname{Cov}(\mathcal{B}(\bm{y})^{(k_1)},\mathcal{B}(\bm{y})^{(k_2)})\pi^{(1)}_{\bm{y}} &= \left[ \bm{\Sigma}^{(1)} \right]_{k_1k_2}.
\end{align}

To simplify this calculation, we reduce the number of conditional distributions, by assuming that the distribution of $\mathcal{B}(\bm{y})$ depends only on the first observed edge $\bm{y}_1$.
In particular, we will assume
\[
  \mathcal{B}(\bm{y}) \sim \mathcal{N}(\bm{0}_K,\bm{\Sigma}^{\dagger}_{\bm{y}_1}).
\]
Then \eqref{ho_bootcov_constraints} reduces to the system of equations
\begin{align*} 
  \sum_{\bm{y} \in \{0,1\}^K} \left[ \bm{\Sigma}^{\dagger}_{\bm{y}_1} \right]_{k_1k_2} \pi^{(0)}_{\bm{y}} &= \left[ \bm{\Sigma}^{(0)} \right]_{k_1k_2} \nonumber \\
  \sum_{\bm{y} \in \{0,1\}^K} \left[ \bm{\Sigma}^{\dagger}_{\bm{y}_1} \right]_{k_1k_2} \pi^{(1)}_{\bm{y}} &= \left[ \bm{\Sigma}^{(1)} \right]_{k_1k_2},
\end{align*}
which further reduces to
\begin{align*} 
  \left[ \bm{\Sigma}^{\dagger}_{0} \right]_{k_1k_2} \prob(\bm{Y}_{12}^{(1)}=0 ~\vert~ \bm{A}_{12}^{(1)}=0) +  \left[ \bm{\Sigma}^{\dagger}_{1} \right]_{k_1k_2} \prob(\bm{Y}_{12}^{(1)}=1 ~\vert~ \bm{A}_{12}^{(1)}=0) &= \left[ \bm{\Sigma}^{(0)} \right]_{k_1k_2} \nonumber \\
  \left[ \bm{\Sigma}^{\dagger}_{0} \right]_{k_1k_2} \prob(\bm{Y}_{12}^{(1)}=0 ~\vert~ \bm{A}_{12}^{(1)}=1) +  \left[ \bm{\Sigma}^{\dagger}_{1} \right]_{k_1k_2} \prob(\bm{Y}_{12}^{(1)}=1 ~\vert~ \bm{A}_{12}^{(1)}=1) &= \left[ \bm{\Sigma}^{(1)} \right]_{k_1k_2}, \nonumber
\end{align*}
a $2 \times 2$ system of linear equations. Note that the coefficient matrix of this system,
\[
  \begin{pmatrix} 1-\alpha & \alpha \\ \beta & 1 - \beta \end{pmatrix}
\]
is invertible under the restriction $\alpha + \beta < 1$. This implies the closed form solutions
\begin{equation*}
  \bm{\Sigma}^{\dagger}_{0} = \frac{(1-\beta)\bm{\Sigma}^{(0)} - \alpha \bm{\Sigma}^{(1)}}{1 - \alpha - \beta}, \quad
  \bm{\Sigma}^{\dagger}_{1} = \frac{(1-\alpha)\bm{\Sigma}^{(1)} - \beta \bm{\Sigma}^{(0)}}{1 - \alpha - \beta}.
\end{equation*}
Thus, $\Yvec^{\dagger}_{ij} \sim \mathcal{N} \left( \bm{0}_K,\bm{\Sigma}^{\dagger}_{\Ynetk{1}_{ij}} \right)$ will satisfy \eqref{ho_bootcov}, as desired.
\end{proof}

We proceed to prove Theorem~\ref{thm:ho_bootdist}.

\begin{proof}
  If $\bm{a} = \bm{0}_K$, the result holds vacuously. Otherwise, from \eqref{ho_S_approx}, we have that
  \begin{align*}
    &\binom{n}{2}^{1/2} \bm{a}^{\tp}(\widetilde{\bm{C}}_H - C_H\bm{1}_K) \\
    =  &\binom{n}{2}^{1/2} \sum_{k=1}^K \frac{\bm{a}_k}{\lvert \mathcal{V} \rvert \{y_k(\theta) - x_k(\theta)\}^L} \cdot \sum_{j=1}^L (-1)^{1-\tau_j} \sum_{\bm{v} \in \mathcal{V}} \left\{\bm{Y}^{(k)}_{v_j} - \mathbb{E}_\theta \bm{Y}^{(k)}_{v_j} \right\} \prod_{\ell \neq j} \mathbb{E}_{\theta} \left\{ \phi_{\ell,\theta}(\bm{Y}^{(k)}_{v_{\ell}}) \right\} + o_{\mathbb{P}}(1).
  \end{align*}

  Define
  $$
    \eta = \expect\left\{ \left( \sum_{k=1}^K \frac{\bm{a}_k}{\lvert \mathcal{V} \rvert \{y_k(\theta) - x_k(\theta)\}^L} \cdot \sum_{j=1}^L (-1)^{1-\tau_j} \sum_{\bm{v} \in \mathcal{V}} \left\{\bm{Y}^{(k)}_{v_j} - \mathbb{E}_\theta \bm{Y}^{(k)}_{v_j} \right\} \prod_{\ell \neq j} \mathbb{E}_{\theta} \left\{ \phi_{\ell,\theta}(\bm{Y}^{(k)}_{v_{\ell}}) \right\} \right)^2 \right\}.
  $$
  Then by an application of Berry-Esseen Theorem, similar to \citesupp{chang22estimationB}, we have
  \begin{equation} \label{BE_S}
    \sup_{z \in \real} \left\lvert \prob\left\{ \binom{n}{2}^{1/2} \bm{a}^{\tp}(\widetilde{\bm{C}}_H - C_H\bm{1}_K) \leq z \right\} - \Phi\left\{ \binom{n}{2}^{-1/2} \eta^{-1/2} z \right\} \right\rvert = o(1).
  \end{equation}

  Expanding the argument in $\eta$, we will get expectations of product terms
  $$
    \sum_{k_1,k_2=1}^K \sum_{j_1,j_2=1}^L \sum_{\bm{v},\tilde{\bm{v}} \in \mathcal{V}} (\cdots)_{k_1,j_1,\bm{v}} (\cdots)_{k_2,j_2,\tilde{\bm{v}}}.
  $$
  For fixed $k_1,k_2,j_1,j_2$, consider the expectations of the internal terms which depend on $\bm{v}$ and $\tilde{\bm{v}}$,
  \begin{equation} \label{c22_s7}
     \sum_{\bm{v},\tilde{\bm{v}} \in \mathcal{V}} \expect\left( \left\{\bm{Y}^{(k_1)}_{v_{j_1}} - \mathbb{E}_\theta \bm{Y}^{(k_1)}_{v_{j_1}} \right\} \left\{\bm{Y}^{(k_2)}_{\tilde{v}_{j_2}} - \mathbb{E}_\theta \bm{Y}^{(k_2)}_{\tilde{v}_{j_2}} \right\} \prod_{\ell \neq j_1} \mathbb{E}_{\theta} \left\{ \phi_{\ell,\theta}(\bm{Y}^{(k_1)}_{v_{\ell}}) \right\} \prod_{\ell' \neq j_2} \mathbb{E}_{\theta} \left\{ \phi_{\ell',\theta}(\bm{Y}^{(k_2)}_{\tilde{v}_{\ell'}}) \right\} \right).
  \end{equation}
  By independence of edges, we have that this term will only be non-zero if $v_{j_1} = \tilde{v}_{j_2}$.
  Otherwise we have
  $$
    \expect\left( \left\{\bm{Y}^{(k_1)}_{v_{j_1}} - \mathbb{E}_\theta \bm{Y}^{(k_1)}_{v_{j_1}} \right\} \left\{\bm{Y}^{(k_2)}_{\tilde{v}_{j_2}} - \mathbb{E}_\theta \bm{Y}^{(k_2)}_{\tilde{v}_{j_2}} \right\} \right) = \cov \left( \bm{Y}^{(k_1)}_{v_{j_1}}, \bm{Y}^{(k_2)}_{v_{j_1}} \right).
  $$
  Define
  \begin{equation} \label{Vj_set}
    \mathcal{V}_{j_1,j_2}(\bm{v}) = \{ \tilde{\bm{v}} \in \mathcal{V} : \tilde{v}_{j_2} = v_{j_1}\}.
  \end{equation}
  Then \eqref{c22_s7} becomes
  \begin{equation} \label{gamma_defn}
    \Gamma_{k_1,k_2,j_1,j_2} := \sum_{\bm{v} \in \mathcal{V}} \cov \left( \bm{Y}^{(k_1)}_{v_{j_1}}, \bm{Y}^{(k_2)}_{v_{j_1}} \right) \prod_{\ell \neq j_1} \mathbb{E}_{\theta} \left\{ \phi_{\ell,\theta}(\bm{Y}^{(k_1)}_{v_{\ell}}) \right\} \sum_{\tilde{\bm{v}} \in \mathcal{V}_{j_1,j_2}(\bm{v})} \prod_{\ell' \neq j_2} \mathbb{E}_{\theta} \left\{ \phi_{\ell',\theta}(\bm{Y}^{(k_2)}_{\tilde{v}_{\ell'}}) \right\}.
  \end{equation}
  for $k_1,k_2 \in \{1,\ldots,K\}$ and $j_1,j_2 \in \{1,\ldots,L\}$, and we can write
  $$
    \eta = \sum_{k_1,k_2=1}^K \sum_{j_1,j_2=1}^L \frac{\bm{a}_{k_1}\bm{a}_{k_2}}{\{y_{k_1}(\theta) - x_{k_1}(\theta)\}^L\{y_{k_2}(\theta) - x_{k_2}(\theta)\}^L} (-1)^{2-\tau_{j_1}-\tau_{j_2}} \lvert \mathcal{V} \rvert^{-2} \Gamma_{k_1,k_2,j_1,j_2}.
  $$

  From \eqref{ho_Sboot}, we can apply the Berry-Esseen Theorem conditional on $\bm{Y} = (\Ynetk{1}, \ldots, \Ynetk{K})$, and perform similar manipulations on the variance term to conclude
  \begin{equation} \label{BE_bootstrap}
    \sup_{z \in \real} \left\lvert \prob\left\{ \bm{a}^{\tp}\bm{S}_H^{\dagger} \leq z ~\vert~ \bm{Y} \right\} - \Phi\left\{ \binom{n}{2}^{-1/2} (\eta^{\dagger})^{-1/2} z \right\} \right\rvert = o(1),
  \end{equation}
  where
  $$
    \eta^{\dagger} = \sum_{k_1,k_2=1}^K \sum_{j_1,j_2=1}^L \frac{\bm{a}_{k_1}\bm{a}_{k_2}}{\{y_{k_1}(\theta) - x_{k_1}(\theta)\}^L\{y_{k_2}(\theta) - x_{k_2}(\theta)\}^L} (-1)^{2-\tau_{j_1}-\tau_{j_2}} \lvert \mathcal{V} \rvert^{-2} \Gamma^{\dagger}_{k_1,k_2,j_1,j_2},
  $$
  and
  \begin{equation} \label{gammadagger_defn}
    \Gamma^{\dagger}_{k_1,k_2,j_1,j_2} = \sum_{\bm{v} \in \mathcal{V}} \cov \left( \bm{Y}^{\dagger,(k_1)}_{v_{j_1}}, \bm{Y}^{\dagger,(k_2)}_{v_{j_1}} ~\vert~ \vec{\bm{Y}}_{v_{j_1}} \right) \prod_{\ell \neq j_1}  \phi_{\ell,\theta}(\bm{Y}^{(k_1)}_{v_{\ell}}) \sum_{\tilde{\bm{v}} \in \mathcal{V}_{j_1,j_2}(\bm{v})} \prod_{\ell' \neq j_2} \phi_{\ell',\theta}(\bm{Y}^{(k_2)}_{\tilde{v}_{\ell'}}).
  \end{equation}
  Note that we use that the edges of the bootstrapped network sequence are independent (across node pairs) conditional on $\bm{Y}$.

  By \eqref{BE_S} and \eqref{BE_bootstrap}, it suffices to show that
  $$
    \lvert \eta^{\dagger} - \eta \vert = o_{\prob}\left\{ \binom{n}{2}^{-1} \right\} = o_{\prob}(n^{-2})
  $$
  which follows if
  \begin{equation} \label{gamma_kj}
      \lvert \mathcal{V} \rvert^{-2} \lvert \Gamma^{\dagger}_{k_1,k_2,j_1,j_2} - \Gamma_{k_1,k_2,j_1,j_2} \vert = o_{\prob}(n^{-2})
  \end{equation}
  for all $k_1,k_2,j_1,j_2$.

  To collapse some notation, for $\ell=1,\ldots,L$, define
  $$
    \tilde{\phi}_{\ell,\theta}(\vec{\bm{Y}}_{v_{\ell}}) = \begin{cases}
        \cov \left( \bm{Y}^{\dagger,(k_1)}_{v_{j_1}}, \bm{Y}^{\dagger,(k_2)}_{v_{j_1}} ~\vert~ \vec{\bm{Y}}_{v_{j_1}} \right), \quad &\ell = j_1 \\
        \phi_{\ell,\theta}(\bm{Y}_{v_{\ell}}^{(k_1)}), \quad &\ell \neq j_1.
  \end{cases}
  $$
  Note that
  $$
    \expect \left\{ \cov \left( \bm{Y}^{\dagger,(k_1)}_{v_{j_1}}, \bm{Y}^{\dagger,(k_2)}_{v_{j_1}} ~\vert~ \vec{\bm{Y}}_{v_{j_1}} \right) \right\} = \cov \left( \bm{Y}^{(k_1)}_{v_{j_1}}, \bm{Y}^{(k_2)}_{v_{j_1}} \right)
  $$
  by construction.
  Then we can decompose
  \begin{align*}
   &\Gamma^{\dagger}_{k_1,k_2,j_1,j_2} - \Gamma_{k_1,k_2,j_1,j_2} \\
   = &\sum_{\bm{v} \in \mathcal{V}} \prod_{\ell = 1}^L   \tilde{\phi}_{\ell,\theta}(\vec{\bm{Y}}_{v_{\ell}}) \sum_{\tilde{\bm{v}} \in \mathcal{V}_{j_1,j_2}(\bm{v})} \prod_{\ell' \neq j_2} \phi_{\ell',\theta}(\bm{Y}^{(k_2)}_{\tilde{v}_{\ell'}}) \\
   &\quad - \sum_{\bm{v} \in \mathcal{V}} \prod_{\ell=1}^L \mathbb{E}_{\theta} \left\{ \tilde{\phi}_{\ell,\theta}(\vec{\bm{Y}}_{v_{\ell}}) \right\} \sum_{\tilde{\bm{v}} \in \mathcal{V}_{j_1,j_2}(\bm{v})} \prod_{\ell' \neq j_2} \mathbb{E}_{\theta} \left\{ \phi_{\ell',\theta}(\bm{Y}^{(k_2)}_{\tilde{v}_{\ell'}}) \right\} \\
   = &\sum_{\bm{v} \in \mathcal{V}} \left[ \prod_{\ell = 1}^L   \tilde{\phi}_{\ell,\theta}(\vec{\bm{Y}}_{v_{\ell}}) - \prod_{\ell=1}^L \mathbb{E}_{\theta} \left\{ \tilde{\phi}_{\ell,\theta}(\vec{\bm{Y}}_{v_{\ell}}) \right\}\right] \sum_{\tilde{\bm{v}} \in \mathcal{V}_{j_1,j_2}(\bm{v})} \prod_{\ell' \neq j_2} \phi_{\ell',\theta}(\bm{Y}^{(k_2)}_{\tilde{v}_{\ell'}}) \\
   &\quad + \sum_{\bm{v} \in \mathcal{V}} \prod_{\ell = 1}^L   \mathbb{E}_{\theta} \left\{ \tilde{\phi}_{\ell,\theta}(\vec{\bm{Y}}_{v_{\ell}}) \right\} \sum_{\tilde{\bm{v}} \in \mathcal{V}_{j_1,j_2}(\bm{v})} \left[ \prod_{\ell' \neq j_2} \phi_{\ell',\theta}(\bm{Y}^{(k_2)}_{\tilde{v}_{\ell'}}) - \prod_{\ell' \neq j_2} \mathbb{E}_{\theta} \left\{ \phi_{\ell',\theta}(\bm{Y}^{(k_2)}_{\tilde{v}_{\ell'}}) \right\} \right] \\
   &\quad + \sum_{\bm{v} \in \mathcal{V}} \left[ \prod_{\ell = 1}^L   \tilde{\phi}_{\ell,\theta}(\vec{\bm{Y}}_{v_{\ell}}) - \prod_{\ell=1}^L \mathbb{E}_{\theta} \left\{ \tilde{\phi}_{\ell,\theta}(\vec{\bm{Y}}_{v_{\ell}}) \right\}\right] \cdot \sum_{\tilde{\bm{v}} \in \mathcal{V}_{j_1,j_2}(\bm{v})} \left[ \prod_{\ell' \neq j_2} \phi_{\ell',\theta}(\bm{Y}^{(k_2)}_{\tilde{v}_{\ell'}}) - \prod_{\ell' \neq j_2} \mathbb{E}_{\theta} \left\{ \phi_{\ell',\theta}(\bm{Y}^{(k_2)}_{\tilde{v}_{\ell'}}) \right\} \right].
  \end{align*}
  This three term decomposition is identical to \citesupp{chang22estimationB}, Proof of Theorem 4.
  Their proof can be used to conclude \eqref{gamma_kj}, noting that $\tilde{\phi}_{\ell,\theta}(\vec{\bm{Y}}_{v_{\ell}})$ is a bounded random variable for all $\ell=1,\ldots,L$.
\end{proof}

Towards the proof of Proposition~\ref{prop:adaptive_distn}., we state and prove the following supporting lemmas.

\begin{lemma} \label{lem:C_to_S_weights}
  Suppose Assumptions~\ref{assump:c22} and \ref{assump:ho_cov} both hold.
  Let
  \begin{align*}
    \bm{w}_*(\theta) &= \frac{\operatorname{Cov}_{\theta}^{-1}(\widetilde{\bm{C}}_H) \bm{1}_K}{\bm{1}_K^{\tp} \operatorname{Cov}_{\theta}^{-1}(\widetilde{\bm{C}}_H) \bm{1}_K} \\
    \bm{w}_{**}(\theta) &= \frac{\operatorname{Cov}_{\theta}^{-1}(\bm{S}_H) \bm{1}_K}{\bm{1}_K^{\tp} \operatorname{Cov}_{\theta}^{-1}(\bm{S}_H) \bm{1}_K}.
  \end{align*}
  Then
  \[
    \lVert \bm{w}_*(\theta) - \bm{w}_{**}(\theta) \rVert_2 \rightarrow 0
  \]
  as $n \rightarrow \infty$.
\end{lemma}

In order to prove Lemma~\ref{lem:C_to_S_weights}, we state and prove another technical lemma.

\begin{lemma} \label{lem:cov_to_weights}
  Suppose $\{A_n\}_{n=1}^{\infty}$ and $\{B_n\}_{n=1}^{\infty}$ are two sequences of $q \times q$ symmetric matrices such that
  $$
    0 < c \leq \lambda_{\min}(A_n) \leq \lambda_{\max}(A_n) \leq C < \infty
  $$
  uniformly over $n$, and $\lVert A_n - B_n \rVert_2 = o(1)$. Then
  $$
    \left\lVert \frac{A_n^{-1} \bm{1}_q}{\bm{1}_q^{\tp} A_n^{-1} \bm{1}_q} - \frac{B_n^{-1} \bm{1}_q}{\bm{1}_q^{\tp} B_n^{-1} \bm{1}_q} \right\rVert_2 = o(1).
  $$
\end{lemma}

\begin{proof}
  Note that
  $$
    \bm{1}_q^{\tp} A_n^{-1} \bm{1}_q \geq q \lambda_{\min}(A_n^{-1}) = \frac{q}{\lambda_{\max}(A_n)} \geq \frac{q}{C},
  $$
  so
  $$
    \frac{1}{\bm{1}_q^{\tp} A_n^{-1} \bm{1}_q} \leq \frac{C}{q}.
  $$
  Fix $n$ sufficiently large such that $\lVert A_n - B_n \rVert_2 \leq c/2$ and $\lvert \bm{1}_q^{\tp} (A_n - B_n) \bm{1}_q \rvert \leq q/(2C)$.
  \begin{align*}
    &\left\lVert \frac{A_n^{-1} \bm{1}_q}{\bm{1}_q^{\tp} A_n^{-1} \bm{1}_q} - \frac{B_n^{-1} \bm{1}_q}{\bm{1}_q^{\tp} B_n^{-1} \bm{1}_q} \right\rVert_2 \\
    \leq &\frac{\sqrt{q}}{\bm{1}_q^{\tp} A_n^{-1} \bm{1}_q} \lVert A_n^{-1} - B_n^{-1} \rVert_2 + \sqrt{q} \lVert B_n^{-1} \rVert_2 \left\lvert \frac{1}{\bm{1}_q^{\tp} A_n^{-1} \bm{1}_q} - \frac{1}{\bm{1}_q^{\tp} B_n^{-1} \bm{1}_q} \right\rvert \\
    \leq &\frac{C}{\sqrt{q}} \lVert A_n^{-1} - B_n^{-1} \rVert_2 + \left( \frac{\sqrt{q}}{c} + \sqrt{q} \lVert A_n^{-1} - B_n^{-1} \rVert_2 \right) \left\lvert \frac{1}{\bm{1}_q^{\tp} A_n^{-1} \bm{1}_q} - \frac{1}{\bm{1}_q^{\tp} B_n^{-1} \bm{1}_q} \right\rvert \\
  \end{align*}
  The result will follow if
  \begin{align*}
    \lVert A_n^{-1} - B_n^{-1} \rVert_2 &= o(1), \\
    \left\lvert \frac{1}{\bm{1}_q^{\tp} A_n^{-1} \bm{1}_q} - \frac{1}{\bm{1}_q^{\tp} B_n^{-1} \bm{1}_q} \right\rvert &= o(1).
  \end{align*}

  Since $\lVert A_n - B_n \rVert_2 = o(1)$, suppose $n$ is sufficiently large such that $\lVert A_n - B_n \rVert_2 < c/2$. Then
  \begin{align*}
    \lVert A_n^{-1} - B_n^{-1} \rVert_2 &\leq \frac{\lVert A_n - B_n \rVert_2}{\lambda_{\min}(A_n) \left\{ \lambda_{\min}(A_n) - \lVert A_n - B_n \rVert_2 \right\} } \\
    &\leq \frac{\lVert A_n - B_n \rVert_2}{c \left\{ c - \lVert A_n - B_n \rVert_2 \right\} } = o(1).
  \end{align*}

  Then
  $$
    \lvert \bm{1}_q^{\tp} (A_n^{-1} - B_n^{-1}) \bm{1}_q \rvert < q \lVert A_n^{-1} - B_n^{-1} \rVert_2 = o(1),
  $$
  so also suppose $n$ is sufficiently large such that
  $$
    \lvert \bm{1}_q^{\tp} (A_n^{-1} - B_n^{-1}) \bm{1}_q \rvert \leq q/(2C).
  $$
  \begin{align*}
    \left\lvert \frac{1}{\bm{1}_q^{\tp} A_n^{-1} \bm{1}_q} - \frac{1}{\bm{1}_q^{\tp} B_n^{-1} \bm{1}_q} \right\rvert &\leq \frac{\lvert \bm{1}_q^{\tp} (A_n^{-1} - B_n^{-1})\bm{1}_q \rvert}{\bm{1}_q^{\tp} A_n^{-1} \bm{1}_q \left\{ \bm{1}_q^{\tp} A_n^{-1} \bm{1}_q - \lvert \bm{1}_q^{\tp} (A_n^{-1} - B_n^{-1}) \bm{1}_q \rvert \right\} } \\
    &\leq \frac{\lvert \bm{1}_q^{\tp} (A_n^{-1} - B_n^{-1})\bm{1}_q \rvert}{(q/C) \left\{ (q/C)- \lvert \bm{1}_q^{\tp} (A_n^{-1} - B_n^{-1}) \bm{1}_q \rvert \right\} } = o(1).\\
  \end{align*}
\end{proof}

\begin{remark} \label{rem:cov_to_weights_prob}
  Lemma~\ref{lem:cov_to_weights} will continue to hold (with convergence in probability) if $B_n$ is a stochastic sequence of matrices, and $\lVert A_n - B_n \rVert = o_{\prob}(1)$.
\end{remark}

We now proceed to prove the original supporting Lemma~\ref{lem:C_to_S_weights}.

\begin{proof}
  First note that
  $$
    \bm{w}_*(\theta) = \frac{\operatorname{Cov}_{\theta}^{-1}(\widetilde{\bm{C}}_H) \bm{1}_K}{\bm{1}_K^{\tp} \operatorname{Cov}_{\theta}^{-1}(\widetilde{\bm{C}}_H) \bm{1}_K} =\frac{\Sigma_C^{-1}(n) \bm{1}_K}{\bm{1}_K^{\tp} \Sigma_C^{-1}(n) \bm{1}_K}
  $$
  where $\Sigma_C(n)$ is defined in \eqref{sigma_C}.

  By Lemma~\ref{lem:cov_to_weights} and Assumption~\ref{assump:ho_cov}, it suffices to show that
  $$
    \lVert \Sigma_C(n) - \operatorname{Cov}_{\theta}(\bm{S}_H) \rVert_2 = o(1).
  $$
  It is sufficient to show that the difference in $(k_1,k_2)$ entries of these matrices converge for each pair $k_1, k_2 \in \{1,\ldots,K\}$.

  For any $k$, we decompose
  $$
    \binom{n}{2}^{1/2}(\tilde{C}_H^{(k)} - C_H) = S_H^{(k)} + R_H^{(k)}
  $$
  where
  $$
    R_H^{(k)} = \binom{n}{2}^{1/2} \frac{1}{\{ y_k(\theta) - x_k(\theta) \}^L}\sum_{\substack{g_1 + \cdots + g_L \geq 2 \\ g_1,\ldots,g_L \in \{0,1\}}} \frac{1}{\lvert \mathcal{V} \rvert} \sum_{\bm{v} \in \mathcal{V}} \prod_{\ell = 1}^L \left\{ \phi_{\ell,\theta}(\bm{Y}^{(k)}_{\bm{v}_{\ell}}) - \expect \phi_{\ell,\theta}(\bm{Y}^{(k)}_{\bm{v}_{\ell}}) \right\}^{g_{\ell}} \left\{ \expect \phi_{\ell,\theta}(\bm{Y}^{(k)}_{\bm{v}_{\ell}}) \right\}^{1 - g_{\ell}},
  $$
  and both $S_H^{(k)}$ and $R_H^{(k)}$ have mean zero.
  The proof of Proposition 2 in \citesupp{chang22estimationB} can be easily adapted to show the following inequalities of the moments of $S_H^{(k)}$ and $R_H^{(k)}$:
  \begin{equation} \label{var_S}
    \var\left( S_H^{(k)} \right) \lesssim \binom{n}{2} \frac{\aleph_{\mathcal{V}}(L-1)}{\lvert \mathcal{V} \rvert} = O(1),
  \end{equation}
  and
  \begin{equation} \label{var_R}
    \var\left( R_H^{(k)} \right) \lesssim \binom{n}{2} \frac{\aleph_{\mathcal{V}}(L-2)}{\lvert \mathcal{V} \rvert} = \binom{n}{2} \frac{\aleph_{\mathcal{V}}(L-2)}{\aleph_{\mathcal{V}}(L-1)} \frac{\aleph_{\mathcal{V}}(L-1)}{\lvert \mathcal{V} \rvert} = o(1),
  \end{equation}
  both by assumption.

  Then
  \begin{align*}
    &\cov\left\{ \binom{n}{2}^{1/2}(\tilde{C}_H^{(k_1)} - C_H), \binom{n}{2}^{1/2}(\tilde{C}_H^{(k_2)} - C_H)\right\}  \\
    = &\cov\left( S_H^{(k_1)} + R_H^{(k_1)}, S_H^{(k_2)} + R_H^{(k_2)} \right),
  \end{align*}
  so that
  \begin{align*}
    &\lvert \cov\left\{ \binom{n}{2}^{1/2}(\tilde{C}_H^{(k_1)} - C_H), \binom{n}{2}^{1/2}(\tilde{C}_H^{(k_2)} - C_H)\right\} - \cov\left( S_H^{(k_1)}, S_H^{(k_2)} \right) \rvert \\
    \leq &\sqrt{\var(S_H^{(k_1)})\var(R_H^{(k_2)})} + \sqrt{\var(R_H^{(k_1)})\var(S_H^{(k_2)})} + \sqrt{\var(R_H^{(k_1)})\var(R_H^{(k_2)})} = o(1).
  \end{align*}
  by Cauchy-Schwarz inequality. This completes the proof.
\end{proof}

A final supporting lemma is used to justify the use of the bootstrapped samples to estimate a weight vector, and combine the subgraph density estimators from each network snapshot.
\begin{lemma} \label{prop:boot_weights}
  Under the conditions of Theorem~\ref{thm:ho_bootdist}, and Assumption~\ref{assump:ho_cov}, let
  \[
    \bm{w}^{\dagger}(\bm{Y};\theta) = \frac{\operatorname{Cov}^{-1}(\bm{S}^{\dagger}_H ~\vert~ \bm{Y}) \bm{1}_K}{\bm{1}_K^{\tp} \operatorname{Cov}^{-1}(\bm{S}^{\dagger}_H ~\vert~ \bm{Y}) \bm{1}_K}.
  \]
  Then
  \begin{equation*}
    \lVert \operatorname{Cov}(\bm{S}^{\dagger}_H ~\vert~ \bm{Y}) - \Sigma_C(n) \rVert_2 = o_{\prob}(1), \quad
    \lVert \bm{w}^{\dagger}(\bm{Y};\theta) - \bm{w}_*(\theta) \rVert_2 = o_{\prob}(1)
  \end{equation*}
  as $n \rightarrow \infty$, where $\bm{w}_*(\theta)$ is defined in \eqref{optwt}.
\end{lemma}

\begin{proof}
  By Lemma~\ref{lem:C_to_S_weights}, Lemma~\ref{lem:cov_to_weights}, and Remark~\ref{rem:cov_to_weights_prob}, it suffices to show that
  \begin{equation} \label{ho_bootcov_convergence}
    \lVert \cov(\bm{S}_H^{\dagger} ~\vert~ \bm{Y}) - \cov(\bm{S}_H) \rVert_2 = o_{\prob}(1).
  \end{equation}
  Similar to the proof of Lemma~\ref{lem:C_to_S_weights}, we show this by proving that the $(k_1,k_2)$ entries converge in probability for fixed $k_1, k_2 \in \{1,\ldots,K\}$.

  Recall the definition,
  $$
    S_H^{(k)} = \binom{n}{2}^{1/2} \frac{1}{\{ y_k(\theta) - x_k(\theta) \}^L} \frac{1}{\lvert \mathcal{V} \rvert} \sum_{j=1}^L (-1)^{1-\tau_j} \sum_{\bm{v} \in \mathcal{V}} (\bm{Y}^{(k)}_{\bm{v}_j} - \expect \bm{Y}^{(k)}_{\bm{v}_j} ) \prod_{\ell \neq j} \expect \phi_{\ell,\theta}(\bm{Y}^{(k)}_{\bm{v}_{\ell}}).
  $$
  To compute $\cov(S_H^{(k_1)},S_H^{(k_2)})$, consider the pairwise covariances of the terms with indices $(k_1,j_1,\bm{v}), (k_2,j_2,\tilde{\bm{v}})$:
  $$
    \cov\left\{ (\bm{Y}^{(k_1)}_{\bm{v}_{j_1}} - \expect \bm{Y}^{(k_1)}_{\bm{v}_{j_1}} ) \prod_{\ell \neq j_1} \expect \phi_{\ell,\theta}(\bm{Y}^{(k_1)}_{\bm{v}_{\ell}}), (\bm{Y}^{(k_2)}_{\tilde{\bm{v}}_{j_2}} - \expect \bm{Y}^{(k_2)}_{\tilde{\bm{v}}_{j_2}} ) \prod_{\ell \neq j_2} \expect \phi_{\ell,\theta}(\bm{Y}^{(k_2)}_{\tilde{\bm{v}}_{\ell}}) \right\}
  $$
  Note that this covariance is zero unless $\bm{v}_{j_1} = \tilde{\bm{v}}_{j_2}$, that is for terms with $\tilde{\bm{v}} \in \mathcal{V}_{j_1,j_2}(\bm{v})$, as defined in \eqref{Vj_set}.

  Expanding over the summations, we can simplify
  $$
    \cov(S_H^{(k_1)},S_H^{(k_2)}) = \binom{n}{2} \frac{1}{\{ y_{k_1}(\theta) - x_{k_1}(\theta) \}^L\{ y_{k_2}(\theta) - x_{k_2}(\theta) \}^L} \sum_{j_1,j_2=1}^L (-1)^{2 - \tau_{j_1} - \tau_{j_2}}\lvert \mathcal{V} \rvert^{-2} \Gamma_{k_1,k_2,j_1,j_2},
  $$
  where $\Gamma_{k_1,k_2,j_1,j_2}$ is defined in \eqref{gamma_defn}.

  Similar manipulations can be performed to show that
  $$
    \cov(S_H^{\dagger,(k_1)},S_H^{\dagger,(k_2)} ~\vert~ \bm{Y}) = \binom{n}{2} \frac{1}{\{ y_{k_1}(\theta) - x_{k_1}(\theta) \}^L\{ y_{k_2}(\theta) - x_{k_2}(\theta) \}^L} \sum_{j_1,j_2=1}^L (-1)^{2 - \tau_{j_1} - \tau_{j_2}}\lvert \mathcal{V} \rvert^{-2} \Gamma^{\dagger}_{k_1,k_2,j_1,j_2},
  $$
  where $\Gamma^{\dagger}_{k_1,k_2,j_1,j_2}$ is defined in \eqref{gammadagger_defn}.

  Thus we have
  $$
  \left\lvert \cov(S_H^{\dagger,(k_1)},S_H^{\dagger,(k_2)} ~\vert~ \bm{Y}) - \cov(S_H^{(k_1)},S_H^{(k_2)}) \right\rvert \lesssim n^2 \sum_{j_1,j_2=1}^L \lvert \mathcal{V} \rvert^{-2} \lvert \Gamma^{\dagger}_{k_1,k_2,j_1,j_2} - \Gamma_{k_1,k_2,j_1,j_2} \rvert.
  $$
  By \eqref{gamma_kj} (from the Proof of Theorem~\ref{thm:ho_bootdist}), we have
  $$
    \left\lvert \cov(S_H^{\dagger,(k_1)},S_H^{\dagger,(k_2)} ~\vert~ \bm{Y}) - \cov(S_H^{(k_1)},S_H^{(k_2)}) \right\rvert = o_{\prob}(1)
  $$
  Which suffices to prove \eqref{ho_bootcov_convergence}.
\end{proof}

With these lemmas in hand, we prove Proposition~\ref{prop:adaptive_distn}.

\begin{proof}
  \begin{align*}
    \binom{n}{2}^{1/2} \left( \widehat{C}_H - C_H\right) &= \binom{n}{2}^{1/2} \left\{ \bm{w}^{\dagger}(\bm{Y};\theta)^{\tp} \widetilde{\bm{C}}_H - C_H \right\} \\
    &= \binom{n}{2}^{1/2} \left\{ \bm{w}^{\dagger}(\bm{Y};\theta)^{\tp}\widetilde{\bm{C}}_H  - \bm{w}_*(\theta)^{\tp} \widetilde{\bm{C}}_H + \bm{w}_*(\theta)^{\tp} \widetilde{\bm{C}}_H - C_H \right\}.
  \end{align*}
  Thus, by Slutsky's Theorem, it suffices to show
  \begin{equation} \label{adaptive_convergence}
    \binom{n}{2}^{1/2} \left\{ \bm{w}^{\dagger}(\bm{Y};\theta)^{\tp}\widetilde{\bm{C}}_H  - \bm{w}_*(\theta)^{\tp} \widetilde{\bm{C}}_H \right\} = o_{\prob}(1).
  \end{equation}

  By construction,
  $$
    1 = \bm{w}^{\dagger}(\bm{Y};\theta)^{\tp}\bm{1}_K = \bm{w}_*(\theta)^{\tp}\bm{1}_K,
  $$
  which implies that
  \begin{equation*} 
    0 = \bm{w}_*(\theta)^{\tp}(C_H\bm{1}_K) - \bm{w}^{\dagger}(\bm{Y};\theta)^{\tp}(C_H \bm{1}_K).
  \end{equation*}
  Thus,
  \begin{align*}
    &\binom{n}{2}^{1/2} \left\{ \bm{w}^{\dagger}(\bm{Y};\theta)^{\tp}\widetilde{\bm{C}}_H  - \bm{w}_*(\theta)^{\tp} \widetilde{\bm{C}}_H \right\} \\
    = &\binom{n}{2}^{1/2} \left\{ \bm{w}^{\dagger}(\bm{Y};\theta)^{\tp}\widetilde{\bm{C}}_H  - \bm{w}^{\dagger}(\bm{Y};\theta)^{\tp}(C_H \bm{1}_K) + \bm{w}_*(\theta)^{\tp}(C_H\bm{1}_K) - \bm{w}_*(\theta)^{\tp} \widetilde{\bm{C}}_H \right\} \\
    = &\bm{w}^{\dagger}(\bm{Y};\theta)^{\tp} \left\{ \binom{n}{2}^{1/2}(\widetilde{\bm{C}}_H - C_H \bm{1}_K) \right\} - \bm{w}_*(\theta)^{\tp}\left\{ \binom{n}{2}^{1/2}(\widetilde{\bm{C}}_H - C_H \bm{1}_K) \right\} \\
    = &\left\{  \bm{w}^{\dagger}(\bm{Y};\theta) - \bm{w}_*(\theta)  \right\}^{\tp} \left\{ \binom{n}{2}^{1/2}(\widetilde{\bm{C}}_H - C_H \bm{1}_K) \right\}.
  \end{align*}
  By Corollary~\ref{cor:ho_order},
  $$
    \binom{n}{2}^{1/2}(\widetilde{\bm{C}}_H - C_H \bm{1}_K) = O_{\prob}(1),
  $$
  and by Proposition~\ref{prop:boot_weights},
  $$
    \lVert \bm{w}^{\dagger}(\bm{Y};\theta) - \bm{w}_*(\theta) \rVert_2 = o_{\prob}(1),
  $$
  which together are sufficient to show \eqref{adaptive_convergence}.
\end{proof}

\subsection{Inference with unknown \texorpdfstring{$\theta$}{theta}} \label{app:ho_theta_unknown}

We state and prove the following lemma, similar to Proposition 4 in \citesupp{chang22estimationB}.

When $\theta = (\alpha,\beta,\lambda,\mu)$ is unknown, define the plug-in analog of $\tilde{C}_H^{(k)}$ for $k=1,\ldots,K$ by
$$
  \widehat{C}_H^{(k)} = \frac{1}{\lvert \mathcal{V} \rvert} \cdot \frac{1}{\{y_k(\hat{\theta}) - x_k(\hat{\theta})\}^L} \cdot \sum_{\bm{v} \in \mathcal{V}} \prod_{\ell=1}^L \phi_{\ell,\hat{\theta}}(\bm{Y}^{(k)}_{v_{\ell}}).
$$
where $\hat{\theta}$ is the GMM estimator developed in Section~\ref{sec:estimation}.

\begin{lemma} \label{lem:ho_order_theta_unknown}
  Suppose the Assumptions of Proposition~\ref{prop:gmm_adaptive_clt} hold, Assumption~\ref{assump:c22} holds,
  \[
      \aleph_{\mathcal{V}}(L-1) / \lvert \mathcal{V} \rvert \asymp n^{-2}
  \]
  and $y_k(\theta) \neq x_k(\theta)$.
  Then we have that
  \begin{equation} \label{ho_theta_approx}
    \binom{n}{2}^{1/2} (\widehat{C}_H^{(k)} - C_H) = S_H^{(k)} + \left( \Delta_H^{(k)} \right)^{\tp} \left\{ \binom{n}{2}^{1/2}(\hat{\theta} - \theta) \right\} + o_{\mathbb{P}}(1).
  \end{equation}
  where $S_H^{(k)}$ is defined in \eqref{SH_linearized} and $\Delta_H^{(k)} \in \real^4$ has a closed form defined in the proof.
\end{lemma}

\begin{proof}
  By Proposition~\ref{prop:gmm_adaptive_clt}, we have
  $$
    \lVert \hat{\theta} - \theta \rVert_2 = O_{\prob}\left\{ \binom{n}{2}^{-1/2} \right\}.
  $$
  Denote its asymptotic covariance matrix by $\bm{\Sigma}$, which is a $4 \times 4$ block of $\bm{\Sigma}^{(\mathrm{GMM})}(\delta_1(\infty),\theta)$.

  Define
  \begin{align*}
    T_H^{(k)} &= \{y_k(\theta) - x_k(\theta)\}^L \cdot C_H, \\
    \widetilde{T}_H^{(k)} &= \{y_k(\theta) - x_k(\theta)\}^L \widetilde{C}_H, \\
    \widehat{T}_H^{(k)} &= \{y_k(\hat{\theta}) - x_k(\hat{\theta})\}^L \widehat{C}_H.
  \end{align*}

  By calculations similar to \citesupp{chang22estimationB}, display (S.9),
  \begin{align} \label{chang_s9}
    &\binom{n}{2}^{1/2} (\widehat{C}_H^{(k)} - C_H) \nonumber \\
    = &\binom{n}{2}^{1/2} \frac{\widehat{T}_H^{(k)} - T_H^{(k)}}{\{ y_k(\theta) - x_k(\theta)\}^L} + \frac{L \cdot C_H \binom{n}{2}^{1/2}\{ x_k(\hat{\theta}) - x_k(\theta) \}}{y_k(\theta) - x_k(\theta)} + \frac{L \cdot C_H \binom{n}{2}^{1/2}\{ y_k(\theta) - y_k(\hat{\theta}) \} }{y_k(\theta) - x_k(\theta)} + o_{\prob}(1) \nonumber \\
    = &S_H^{(k)} + \binom{n}{2}^{1/2} \frac{\widehat{T}_H^{(k)} - \widetilde{T}_H^{(k)}}{\{ y_k(\theta) - x_k(\theta)\}^L} + \frac{L \cdot C_H \{ \nabla x_k(\theta) - \nabla y_k(\theta) \} }{y_k(\theta) - x_k(\theta)} \left\{ \binom{n}{2}^{1/2}(\hat{\theta} - \theta) \right\} + o_{\prob}(1),
  \end{align}
  where we use the fact that both $x_k$ and $y_k$ are bounded polynomials in $\theta$, and $\lVert \hat{\theta} - \theta \rVert_2 = O_{\prob}\left\{ \binom{n}{2}^{-1/2} \right\}$, thus
  \begin{align*}
    x_k(\hat{\theta}) - x_k(\theta) &= \{\nabla x_k(\theta)\}^{\tp} (\hat{\theta} - \theta) + O_{\prob}\left\{ \binom{n}{2}^{-1} \right\}, \\
    y_k(\hat{\theta}) - y_k(\theta) &= \{\nabla y_k(\theta)\}^{\tp} (\hat{\theta} - \theta) + O_{\prob}\left\{ \binom{n}{2}^{-1} \right\}.
  \end{align*}

  It remains to analyze the difference $\widehat{T}_H^{(k)} - \widetilde{T}_H^{(k)}$.

  Again following the same calculations as \citesupp{chang22estimationB}, p. S15-S16, we find
  $$
    \widehat{T}_H^{(k)} - \widetilde{T}_H^{(k)} = \sum_{j=1}^L \{ x_k(\theta) - x_k(\hat{\theta}) \}^{\tau_j} \{ y_k(\hat{\theta}) - y_k(\theta) \}^{1-\tau_j} \lvert \mathcal{V} \rvert^{-1} \sum_{\bm{v} \in \mathcal{V}} \prod_{\ell \neq j} \expect_{\theta}\left\{ \phi_{\ell,\theta}(\bm{Y}^{(k)}_{v_{\ell}}) \right\} + O_{\prob}\left\{ \binom{n}{2}^{-1} \right\}.
  $$

  Also recall that
  $$
    \frac{\expect_{\theta}\left\{ \phi_{\ell,\theta}(\bm{Y}^{(k)}_{v_{\ell}}) \right\}}{y_k(\theta) - x_k(\theta)} = \left( A^{(1)}_{v_{\ell}} \right)^{\tau_{\ell}}\left( 1 - A^{(1)}_{v_{\ell}}\right)^{1-\tau_{\ell}}.
  $$
  For each $j=1,\ldots,L$, define
  $$
    U_H^{(-j)} = \lvert \mathcal{V} \rvert^{-1} \sum_{\bm{v} \in \mathcal{V}} \prod_{\ell \neq j} \left( A^{(1)}_{v_{\ell}} \right)^{\tau_{\ell}}\left( 1 - A^{(1)}_{v_{\ell}}\right)^{1-\tau_{\ell}}.
  $$
  This quantity $U_H^{(-j)}$ is a {\em leave-one-out} subgraph density, which enumerates all the configurations for subgraph $H$ but ignores the status of the $j$th edge.

  Combining these last four displays, we get
  $$
    \frac{\widehat{T}_H^{(k)} - \widetilde{T}_H^{(k)}}{\{y_k(\theta) - x_k(\theta)\}^L} = \left\{ \frac{\sum_{j : \tau_j=0} U_H^{(-j)} \nabla y_k(\theta) - \sum_{j : \tau_j=1} U_H^{(-j)} \nabla x_k(\theta)}{y_k(\theta) - x_k(\theta)}\right\}^{\tp} (\hat{\theta} - \theta) + O_{\prob}\left\{ \binom{n}{2}^{-1} \right\}.
  $$
  Plugging this approximation into \eqref{chang_s9} completes the proof of \eqref{ho_theta_approx}, with
  \begin{align*}
    \Delta_H^{(k)} &= \frac{\sum_{j : \tau_j=0} U_H^{(-j)} \nabla y_k(\theta) - \sum_{j : \tau_j=1} U_H^{(-j)} \nabla x_k(\theta) + L \cdot C_H \{ \nabla x_k(\theta) - \nabla y_k(\theta) \}}{y_k(\theta) - x_k(\theta)} \\
    &= \frac{\sum_{j : \tau_j=0} (U_H^{(-j)} - C_H) \nabla y_k(\theta) - \sum_{j : \tau_j=1} (U_H^{(-j)} - C_H) \nabla x_k(\theta) }{y_k(\theta) - x_k(\theta)}. 
  \end{align*}
\end{proof}

\begin{remark} \label{rem:loo_subgraphs}
    If $H$ is a triangle, then for $j=1$, the leave-one-out subgraph density considers all ordered triples of distinct nodes $(i_1,i_2,i_3)$ and counts whether the $i_2 \leftrightarrow i_3$ and $i_3 \leftrightarrow i_1$ edges are present. By rotational symmetry this subgraph density is the same for $j=2$ and $j=3$ as well. We refer to this as the {\em two-chain} density.

    If $H$ is a two-star, then for $j=1$ the leave-one-out subgraph density is the same two-chain density desribed above for triangles.
    However, if $j=2$, the leave-one-out subgraph density considers all ordered triples of distinct nodes $(i_1,i_2,i_3)$ and counts whether the $i_1 \leftrightarrow i_2$ edge is absent and the $i_3 \leftrightarrow i_1$ edge is present. By rotational symmetry, this leave-one-out subgraph density is the same for $j=3$. 
    We refer to this as the {\em half-chain} density.

    Including triangle density and two-star density, denote these four subgraph densities on ordered node triples by
    $$
        \subtri{C}, \quad \subtstar{C}, \quad \subtchain{C}, \quad \subhchain{C}
    $$
    respectively.     

    Then, for inference on triangle and two-star density, we can derive simplified expressions for $\Delta_H^{(k)}$ as follows.

    \begin{align*}
        \subtri{\Delta}^{(k)} &= \frac{-3 \subtstar{C} \nabla x_k(\theta)}{y_k(\theta) - x_k(\theta)}, \\
        \subtstar{\Delta}^{(k)} &= \frac{\subtri{C} \nabla y_k(\theta) - 2(\subhchain{C} - \subtstar{C}) \nabla x_k(\theta)}{y_k(\theta) - x_k(\theta)},
    \end{align*}
    where we use the relationship $\subtri{C} + \subtstar{C} = \subtchain{C}$.
\end{remark}

Combining the general result of Lemma~\ref{lem:ho_order_theta_unknown} over $k=1,\ldots,K$, we have
$$
  \binom{n}{2}^{1/2} \left( \widehat{\bm{C}}_H - C_H \bm{1}_K \right) = \bm{S}_H + \bm{\Delta}_H \left\{ \binom{n}{2}^{1/2}(\hat{\theta} - \theta) \right\} + o_{\prob}(1),
$$
where
$$
  \bm{\Delta}_H^{\tp} = \begin{pmatrix}
    \Delta_H^{(1)} & \cdots & \Delta_H^{(K)}
\end{pmatrix} \in \real^{4 \times K}.
$$

Based on the approximation in \eqref{ho_theta_approx}, an analog of \citepsupp{chang22estimationB}, Theorem 5 will also continue to hold in this setting, which we can use to develop an estimator for the asymptotic covariance matrix for $\widehat{\bm{C}}_H$.
In particular, if $\bm{\Delta}_H$ converges, we have
$$
  \acov\left( \widehat{\bm{C}}_H \right) = \acov(\bm{S}_H) + \bm{\Delta}_H \acov(\hat{\theta}) \bm{\Delta}_H^{\tp} + \acov(\bm{S}_H,\hat{\theta}) \bm{\Delta}_H^{\tp} + \bm{\Delta}_H\acov(\bm{S}_H,\hat{\theta})^{\tp}
$$

The first term $\acov(\bm{S}_H)$ can be consistently estimated using the bootstrapping approach described in Section~\ref{sec:ho} and by plugging in the GMM estimator $\hat{\theta}$.
For the second term, a consistent estimator for $\acov(\hat{\theta})$ is given in Section~\ref{subsec:theta_unknown}; $\bm{\Delta}_H$ can be estimated by plugging in the GMM estimator $\hat{\theta}$, $\widehat{C}_H^{(1)}$, and
\begin{equation*}
  \widehat{U}_H^{(-j)} = \frac{1}{\lvert \mathcal{V} \rvert \{ y_1(\hat{\theta}) - x_1(\hat{\theta})
  \}^{L-1}} \sum_{\bm{v} \in \mathcal{V}} \prod_{\ell \neq j} \phi_{\ell,\hat{\theta}}(\Ynetk{1}_{v_{\ell}})
\end{equation*}
for $j=1,\ldots,L$ based on the first noisy snapshot.
All that remains is to find a consistent estimator of the asymptotic cross-covariance matrix $\acov(\bm{S}_H,\hat{\theta})$.

Following the proof of Proposition~\ref{prop:theta_limiting}, we can write
$$
  \binom{n}{2}^{1/2} (\hat{\theta} - \theta) = \tilde{\bm{M}}(\psi) \left\{ \binom{n}{2}^{1/2} (\hat{\bm{m}} - \bm{m}(\psi))\right\} + o_{\prob}(1)
$$
for a matrix $\tilde{\bm{M}}(\psi)$ made up of the second through fifth rows of
$$
  \bm{M}(\psi) = (\bm{Dm}(\psi)^{\tp} \bm{W} \bm{Dm}(\psi))^{-1} \bm{Dm}(\psi)^{\tp} \bm{W}.
$$
In practice, to evaluate $\tilde{\bm{M}}(\psi)$ we substitute the optimal asymptotic weight matrix for $\bm{W}$:
 \begin{equation*} 
    \bm{W} = \left\{ \delta_1 \bm{\Sigma}_1(\theta) + (1 - \delta_1) \bm{\Sigma}_0(\theta) \right\}^{-1}.
  \end{equation*}
This matrix is a continuous function of the unknown edge density, error and evolution parameters, and recalling that $\bm{m}$ collects the edge averaged moments used for estimation.

Thus, it suffices to find the coordinate-wise covariances between $S^{(k)}_H$ and $\hat{\bm{m}}_h$ for fixed $k=1,\ldots,K$ and $h=1,\ldots,p$.
Note that both of these quantities decompose over the edges.
Define
\begin{equation*}
  \sigma_{s,kh}(\theta) = \cov_{\theta}\left( \Ynet_{ij}^{(k)}, m_h(\Yvec_{ij}) ~\vert~ \Anetk{1}_{ij}=s \right)
\end{equation*}
for $s \in \{0,1\}$, so that
$$
  \cov_{\theta}(\Ynetk{k}_{ij},m_h(\Yvec_{ij})) = \sigma_{0,kh}(\theta) + \Anetk{1}_{ij}\left\{ \sigma_{1,kh}(\theta)  - \sigma_{0,kh}(\theta) \right\}.
$$
Note that these covariances are continuous (polynomial) functions of the error and evolution parameters, and can be evaluated using observation operators. Then
\begin{align*}
  &\cov_{\theta}\left( S^{(k)}_H,\left[ \binom{n}{2}^{1/2} (\hat{\bm{m}} - \bm{m}(\psi))\right]_h \right) \\
  = &\frac{1}{2 \lvert \mathcal{V} \rvert \{y_k(\theta) - x_k(\theta)\}^L} \sum_{j=1}^L (-1)^{1-\tau_j} \sum_{\bm{v} \in \mathcal{V}} \sum_{s \neq t} \cov_{\theta}(\Ynetk{k}_{v_j},m_h(\Yvec_{st})) \prod_{\ell \neq j} \expect_{\theta}\left\{ \phi_{\ell,\theta}(\bm{Y}^{(k)}_{v_{\ell}}) \right\} \\
  = &\frac{\sigma_{0,kh}(\theta)}{2 \lvert \mathcal{V} \rvert \{y_k(\theta) - x_k(\theta)\}^L} \sum_{j=1}^L (-1)^{1-\tau_j} \sum_{\bm{v} \in \mathcal{V}} \prod_{\ell \neq j} \expect_{\theta}\left\{ \phi_{\ell,\theta}(\bm{Y}^{(k)}_{v_{\ell}}) \right\} + \cdots \\
  &\quad \cdots + \frac{\sigma_{1,kh}(\theta)  - \sigma_{0,kh}(\theta)}{2 \lvert \mathcal{V} \rvert \{y_k(\theta) - x_k(\theta)\}^L} \sum_{j=1}^L (-1)^{1-\tau_j} \sum_{\bm{v} \in \mathcal{V}} \Anetk{1}_{v_j} \prod_{\ell \neq j} \expect_{\theta}\left\{ \phi_{\ell,\theta}(\bm{Y}^{(k)}_{v_{\ell}}) \right\}
\end{align*}
Similar to the proof of Lemma~\ref{lem:ho_order_theta_unknown}, the first term can be simplified in terms of the leave-one-out subgraph densities $U_H^{(-j)}$.
To simplify the second term, we define {\em add-one} subgraph densities
$$
  W_H^{(j)} = \lvert \mathcal{V} \rvert^{-1} \sum_{\bm{v} \in \mathcal{V}} \Anetk{1}_{v_j} \cdot \prod_{\ell \neq j} \left( \Anetk{1}_{v_{\ell}} \right)^{\tau_{\ell}}\left( 1 - \Anetk{1}_{v_{\ell}}\right)^{1-\tau_{\ell}}.
$$
which enumerates the same configurations as the original subgraph $H$ but requiring the $j$th edge to be present. Then
\begin{align*}
  &\cov_{\theta}\left( S^{(k)}_H,\left[ \binom{n}{2}^{1/2} (\hat{\bm{m}} - \bm{m}(\psi))\right]_h \right) \\
  = &\frac{\sigma_{0,kh}(\theta)}{2 \{y_k(\theta) - x_k(\theta)\}} \sum_{j=1}^L (-1)^{1-\tau_j} U_H^{(-j)} + \frac{\sigma_{1,kh}(\theta)  - \sigma_{0,kh}(\theta)}{2 \{y_k(\theta) - x_k(\theta)\}} \sum_{j=1}^L (-1)^{1-\tau_j} W_H^{(j)}.
\end{align*}
Define these covariance matrices by $\Omega_{H,n}$, and suppose the sequence converges to a limit $\Omega_H \in \real^{K \times p}$.
Then
$$
  \acov(\bm{S}_H,\hat{\theta}) = \Omega_H\tilde{\bm{M}}(\psi)^{\tp} \in \real^{K \times 4}.
$$
Under our regularity conditions, this asymptotic covariance can be consistently estimated by plugging in the GMM estimator $\hat{\psi}$, as well as estimators
\begin{equation*}
  \widehat{W}_H^{(j)} = \frac{1}{\lvert \mathcal{V} \rvert \{ y_1(\hat{\theta}) - x_1(\hat{\theta}) \}^{L}} \sum_{\bm{v} \in \mathcal{V}} \left\{ \Ynetk{1}_{v_{\ell}} - x_1(\hat{\theta}) \right\} \prod_{\ell \neq j} \phi_{\ell,\hat{\theta}}(\Ynetk{1}_{v_{\ell}})
\end{equation*}
for $j=1,\ldots,L$, based on the first noisy snapshot.

\begin{remark}
    Similar to Remark~\ref{rem:loo_subgraphs}, we can explicitly describe the add-one subgraph densities, and the resulting $\Omega_H$ in the case where $H$ is either a triangle or a two-star.
    If $H$ is a triangle, then it is easy to see that $\subtri{W}^{(j)}$ is also the triangle density $\subtri{C}$ for $j=1, 2, 3$.
    If $H$ is a two-star, then $\subtstar{W}^{(j)}$ is also the two-star density $\subtstar{C}$ for $j=2, 3$. If $j=1$, then $\subtstar{W}^{(1)}$ is the triangle density $\subtri{C}$.

    we can derive simplified expressions for the $(k,h)$ entry of $\Omega_H$ as follows.
    \begin{align*}
        \left[ \subtri{\Omega} \right]_{kh} &= \frac{3 \sigma_{1,kh}(\theta) \subtri{C} + 3 \sigma_{0,kh}(\theta)\subtstar{C}}{2 \left\{ y_k(\theta) - x_k(\theta) \right\}}, \\
        \left[ \subtstar{\Omega} \right]_{kh} &= \frac{\sigma_{0,kh}(\theta) (2 \subhchain{C} - \subtchain{C}) + \{\sigma_{1,kh}(\theta)  - \sigma_{0,kh}(\theta) \} (2 \subtstar{C} - \subtri{C})}{2 \{y_k(\theta) - x_k(\theta)\}},
    \end{align*}
    where we use the relationship $\subtchain{C} = \subtri{C} + \subtstar{C}$.
\end{remark}

\begin{remark} \label{rem:ustats}
  In principle, these same types of expansions could be used to evaluate the analytical covariances of $S_H^{(k_1)}$ and $S_H^{(k_2)}$ for any $k_1$ and $k_2$ without the need for an edge bootstrap.
  However, the resulting expressions would involve $(L^2)$-many sums over configurations of edges in $\mathcal{V} \times \mathcal{V}$, and thus will be computationally expensive to evaluate when $n$ is large, even relative to the computational cost of the bootstrap.
  For instance, the variance of triangle density would require sums with $O(n^6)$ terms, relative to sums with $O(n^3)$ terms to compute each bootstrap replicate $\bm{S}_H^{\dagger}$.
\end{remark}

\begin{remark}
  As in \citesupp{chang22estimationB}, similar results will hold for joint inference on multiple subgraph densities, where the covariance among densities can be calculated by estimating both on the same set of bootstrap replicates.
  The covariance between subgraph densities and $\hat{\theta}$ can be estimated using the methods presented in this section.
\end{remark}

\section{Theoretical sketches for the dynamic comparison model} \label{app:comparison_theory}

\subsection{GMM estimation}

We state an analog of Lemma~\ref{lem:mstar_clt} under the comparison model. The proof is identical to the proof of Lemma~\ref{lem:mstar_clt} with some analogous assumptions stated below.

\begin{lemma} \label{lem:mstar_clt_compare}
  Define $m^*_{\mathcal{C}}$ as in Section~\ref{subsec:gmm_comparison}, and suppose $\rho_1, \rho_K, \rho_{1K} \in (0,1)$.
  Then
  \begin{equation} \label{mstar_clt_compare}
    \hspace{-1cm}
    \binom{n}{2}^{-1/2}\left[ \sum_{i < j} m^*_{\mathcal{C}}(\Yvec_{ij}) - \binom{n}{2} \bm{m}^*_{\mathcal{C}}(\rho_1,\rho_K,\rho_{1K},\theta) \right]
    \indist \mathcal{N}\left( \bm{0}_{K+5}, \bm{\Sigma}^*_{\mathcal{C}}(\rho_1,\rho_K,\rho_{1K},\theta)  \right)
  \end{equation}
  as $n \rightarrow \infty$, where
  $$
    \bm{\Sigma}^*_{\mathcal{C}}(\rho_1,\rho_K,\rho_{1K},\theta) = \rho_{1K} \bm{\Sigma}^*_{\mathcal{C},11}(\theta) + \rho_1 \bm{\Sigma}^*_{\mathcal{C},10}(\theta) + \rho_K \bm{\Sigma}^*_{\mathcal{C},01}(\theta) + (1 - \rho_1 - \rho_K - \rho_{1K})\bm{\Sigma}^*_{\mathcal{C},00}(\theta)
  $$
  and
  $$
    \bm{\Sigma}^*_{\mathcal{C},st}(\theta) = \cov_{\theta} \left\{ m^*_{\mathcal{C}}\left( \Yvec_{12} \right) ~\big\vert~ \Anetk{1}_{12}=s, \Anetk{K}_{12}=t \right\} \in \real^{(K+5) \times (K+5)}.
  $$
\end{lemma}

Using Lemma~\ref{lem:mstar_clt_compare}, derivation of the asymptotic distribution for the GMM estimators, including an asymptotic distribution for the linear combinations
$$
  \hat{\delta}_1 = \hat{\rho}_1 + \hat{\rho}_{1K}, \quad \hat{\delta}_K = \hat{\rho}_K + \hat{\rho}_{1K}
$$
will follow the same proof structure as in Section~\ref{app:proofs_gmm}, under analogous assumptions to Assumptions~\ref{assump:global_identification}-\ref{assump:cov_spectrum}.

Define
$$
  \Psi_{\mathcal{C}} = \{(\rho_1,\rho_K,\rho_{1K},\alpha,\beta,\lambda,\mu) \in (0,1)^7 : \alpha + \beta < 1, \lambda + \mu < 1\},
$$
and the compact restriction
$$
  \bar{\Psi}_{\mathcal{C}} = \bar{\Psi}_{\mathcal{C}}(\xi) =  \{(\rho_1,\rho_K,\rho_{1K},\alpha,\beta,\lambda,\mu) \in [\xi,1-\xi]^7 : \alpha + \beta \leq 1-\xi, \lambda + \mu \leq 1-\xi, \rho_1 + \rho_K + \rho_{1K} \leq 1 - \xi\}.
$$

\begin{assumption} \label{assump:global_identification_compare}
The mapping $\bm{m}_{\mathcal{C}}^{(\mathrm{init})}(\rho_1,\rho_K,\rho_{1K},\theta)$ is an injective function on $\Psi_{\mathcal{C}}$.
\end{assumption}

\begin{assumption} \label{assump:delta_seq_compact_compare}
  The true parameters satisfy $(\rho_1(n),\rho_K(n),\rho_{1K}(n),\alpha,\beta,\lambda,\mu) \in \bar{\Psi}_{\mathcal{C}}$ for all $n$, and the sequence is convergent as $n \rightarrow \infty$.
\end{assumption}

\begin{assumption} \label{assump:local_identification_body_compare}
  The derivative matrix
  $$
    \bm{Dm}^*(\rho_1,\rho_K,\rho_{1K},\theta)
  \in \real^{(K+5) \times 7}
  $$
  has full column rank for all $(\rho_1,\rho_K,\rho_{1K},\theta) \in \bar{\Psi}_{\mathcal{C}}$.
\end{assumption}

\begin{assumption} \label{assump:cov_spectrum_compare}
  The covariance matrix $\bm{\Sigma}^*_{\mathcal{C}}(\rho_1,\rho_K,\rho_{1K},\theta)$
  satisfies
  $$
    0 < c \leq \lambda_{\min}\left\{ \bm{\Sigma}^*_{\mathcal{C}}(\rho_1,\rho_K,\rho_{1K},\theta) \right\} \leq \lambda_{\max}\left\{ \bm{\Sigma}^*_{\mathcal{C}}(\rho_1,\rho_K,\rho_{1K},\theta) \right\} \leq C < \infty
  $$
  uniformly over all $(\rho_1,\rho_K,\rho_{1K},\theta) \in \bar{\Psi}_{\mathcal{C}}$ for constants $c$ and $C$ (which may depend on $\xi$).
\end{assumption}

\subsection{Computing bootstrap replicates under the dynamic comparison model}

Under the dynamic comparison model, we simplify the expressions in Section~\ref{sec:comparison}. 
As in that section, we drop ``$ij$'' subscripts and show how to find bootstrap replicates
$$
    \begin{pmatrix}
        \Ynetk{1,\dagger} \\ \Ynetk{K,\dagger}
    \end{pmatrix} \sim p(\cdot , \Ynetk{1},\Ynetk{K}) 
$$
With expected conditional covariance structure which mimics the covariance of the original:
$$
    \hspace{-0.5cm}
    \cov\left\{ \begin{pmatrix}
        \Ynetk{1} \\ \Ynetk{K}
    \end{pmatrix} ~\bigg\vert~ \Anetk{1}=s, \Anetk{K}=t \right\} = \begin{pmatrix}
        (1-s)\alpha(1-\alpha) + s\beta(1-\beta) & 0 \\ 0 & (1-t)\alpha(1-\alpha) + t\beta(1-\beta)
    \end{pmatrix}.
$$
As $\Ynetk{1}$ and $\Ynetk{K}$ are uncorrelated for any underlying state $(s,t)$, it is sufficient to independently bootstrap the two snapshots.

Define 
\begin{align*}
    \nu_0^{\dagger} = \alpha (1 - \beta), \quad \nu_1^{\dagger} = (1-\alpha)\beta.
\end{align*}
We claim that a valid bootstrap procedure will sample
$$
    \Ynetk{1,\dagger} \sim \mathcal{N}\left( 0,\nu_{\Ynetk{1}}^{\dagger} \right), \quad \Ynetk{K,\dagger} \sim \mathcal{N}\left( 0,\nu_{\Ynetk{K}}^{\dagger} \right)
$$
independently. This can be verified by checking that for $s \in \{0,1\}$ and $k \in \{1,K\}$,
\begin{align*}
    &\expect\left\{ \var( \Ynetk{k,\dagger} ~\vert~ \Ynetk{k} ) ~\vert~ \Anetk{k}=s \right\} \\
    &= \expect\{ \nu_1^{\dagger} \Ynetk{k} + \nu_0^{\dagger}(1 - \Ynetk{k}) ~\vert~ \Anetk{k}=s \} \\
    &= \nu_0^{\dagger} + (\nu_1^{\dagger} - \nu_0^{\dagger}) \expect( \Ynetk{k} ~\vert~ \Anetk{k}=s) \\
    &= \alpha - \alpha\beta + (\beta - \alpha)\{(1-s)\alpha + s(1-\beta)\} \\
    &= \alpha - \alpha\beta + \alpha\beta - s\alpha\beta + s\beta - s\beta^2 - \alpha^2 - s\alpha^2 - s\alpha + s\alpha\beta \\
    &= \alpha - \alpha^2 - s\alpha + s\alpha^2 + s \beta - s\beta^2 \\
    &= (1-s)\alpha(1-\alpha) + s \beta(1-\beta) \\
    &= \var( \Ynetk{k} ~\vert~ \Anetk{k} = s).
\end{align*}

\subsection{Higher-order subgraph density estimation with unknown \texorpdfstring{$\theta$}{theta}} \label{app:ho_comparison_theta_unknown}

In this section we develop detailed expressions to estimate the asymptotic covariance of the two plug-in estimators
\[
  \widehat{C}_{H,\mathcal{C}}^{(1)} =  \frac{1}{\lvert \mathcal{V} \rvert} \sum_{\bm{v} \in \mathcal{V}} \prod_{\ell=1}^L \phi_{\ell,\hat{\theta}}(\Ynetk{1}_{v_{\ell}}), \quad \widehat{C}_{H,\mathcal{C}}^{(K)} =  \frac{1}{\lvert \mathcal{V} \rvert} \sum_{\bm{v} \in \mathcal{V}} \prod_{\ell=1}^L \phi_{\ell,\hat{\theta}}(\Ynetk{K}_{v_{\ell}}), 
\]
of subgraph density developed in Section~\ref{subsec:ho_comparison}, where $\hat{\theta}$ is the GMM estimator of $(\alpha,\beta,\lambda,\mu)$ developed in Section~\ref{subsec:ho_comparison}.
Note that $\phi_{\ell,\theta}$ is free of $\lambda$ and $\mu$, and in fact coincides with the edge adjustment used in \citesupp{chang22estimationB} without dynamic observation.

Towards (joint) inference using these two estimators, we state a lemma which gives a simple asymptotic expansion of each of these estimators to isolate the additional variability induced by estimation of $\theta$.
The proof of this lemma omitted, as it is similar to the proof of Lemma~\ref{lem:ho_order_theta_unknown}, where we note that in this case,
$$
     x_k(\theta) = \alpha, \quad y_k(\theta) = 1-\beta, \quad \nabla x_k(\theta) = (1,0,0,0)^{\tp}, \quad \nabla y_k(\theta) = (0,-1,0,0)^{\tp}
$$
for $k \in \{1,K\}$.
\begin{lemma}
    Under the conditions of Lemma~\ref{lem:ho_order_theta_unknown}, for $k \in \{1,K\}$,
    \begin{equation} \label{approx_ho_comparison_unknown_theta}
        \binom{n}{2}^{1/2}(\widehat{C}_{H,\mathcal{C}}^{(k)} - C_H^{(k)}) = S_{H,\mathcal{C}}^{(k)} + \left( \bm{\Xi}_{H}^{(k)} \right)^{\tp} \left\{ \binom{n}{2}^{1/2} \begin{pmatrix} \hat{\alpha} - \alpha \\ \hat{\beta} - \beta \end{pmatrix} \right\} + o_{\prob}(1),
    \end{equation}
    where
    $$
        S_{H,\mathcal{C}}^{(k)} = \binom{n}{2}^{1/2} \cdot \frac{1}{\lvert \mathcal{V} \rvert} \cdot \sum_{j=1}^L (-1)^{1-\tau_j} \sum_{\bm{v} \in \mathcal{V}} \left\{\frac{\bm{Y}^{(k)}_{v_j} - \mathbb{E}_\theta \bm{Y}^{(k)}_{v_j}}{1 - \alpha - \beta} \right\} \prod_{\ell \neq j} \mathbb{E}_{\theta} \left\{ \phi_{\ell,\theta}(\bm{Y}^{(k)}_{v_{\ell}}) \right\}.
    $$
    $$
    \bm{\Xi}_{H}^{(k)} = \frac{1}{1- \alpha -\beta} \begin{pmatrix} L \cdot C_H^{(k)} - \sum_{j : \tau_j=1} U_H^{(-j,k)} \\ 
     L \cdot C_H^{(k)} - \sum_{j : \tau_j=0} U_H^{(-j,k)}
    \end{pmatrix}
  $$
  and 
  $$
    U_H^{(-j,k)} = \lvert \mathcal{V} \rvert^{-1} \sum_{\bm{v} \in \mathcal{V}} \prod_{\ell \neq j} \left( A^{(k)}_{v_{\ell}} \right)^{\tau_{\ell}}\left( 1 - A^{(k)}_{v_{\ell}}\right)^{1-\tau_{\ell}}.
  $$
  As above, $U_H^{(-j,k)}$ is a ``leave-one-out'' subgraph density for the $j$th edge in the subgraph configuration, and the $k$th underlying snapshot.
\end{lemma}

\begin{remark}
    Recall, for triangle density, $\tau_1=\tau_2=\tau_3=1$, and $\subtri{U}^{(-j,k)} = \subtchain{C}^{(k)}$ for $j=1, 2, 3$. Thus,
    $$
    \subtri{\bm{\Xi}}^{(k)} = \frac{3}{1- \alpha -\beta} \begin{pmatrix}  \subtri{C}^{(k)} - \subtchain{C}^{(k)} \\ 
    \subtri{C}^{(k)}
    \end{pmatrix} = \frac{3}{1- \alpha -\beta} \begin{pmatrix}  -\subtstar{C}^{(k)} \\ 
    ~ \subtri{C}^{(k)}
    \end{pmatrix}
    $$
    for $k=1, K$.

    For two-star density, $\tau_1=0$ with $\subtstar{U}^{-1,k)} = \subtchain{C}^{(k)}$, and $\tau_2 = \tau_3 = 1$ with $\subtstar{U}^{-j,k)} = \subhchain{C}^{(k)}$ for $j=2, 3$. Thus,
    $$
        \subtstar{\bm{\Xi}}^{(k)} = \frac{1}{1- \alpha -\beta} \begin{pmatrix} 3 \subtstar{C}^{(k)} - 2 \subhchain{C}^{(k)} \\ 
     3 \subtstar{C}^{(k)} - \subtchain{C}^{(k)}
    \end{pmatrix}
    $$
    for $k=1, K$.
\end{remark}

From the expression in \eqref{approx_ho_comparison_unknown_theta}, a bootstrap procedure, plus analogous derivations to those in Section~\ref{app:ho_theta_unknown} to find the asymptotic variances
$$
    \operatorname{aVar}\left\{ \widehat{C}_{H,\mathcal{C}}^{(k)} \right\} = \operatorname{aVar}(S_{H,\mathcal{C}}^{(k)}) + 2 \cdot \operatorname{aCov}\left\{ S_{H,\mathcal{C}}^{(k)}, \begin{pmatrix}
        \hat{\alpha} \\ \hat{\beta}
    \end{pmatrix} \right\} \bm{\Xi}_H^{(k)} +  \left( \bm{\Xi}_H^{(k)} \right)^{\tp} \operatorname{aCov} \begin{pmatrix}
        \hat{\alpha} \\ \hat{\beta}
    \end{pmatrix} \bm{\Xi}_H^{(k)}
$$
for $k \in \{1,K\}$.
It is also possible to determine the asymptotic covariance between $\widehat{C}_{H,\mathcal{C}}^{(1)}$ and $\widehat{C}_{H,\mathcal{C}}^{(K)}$ which is induced by the common estimates of $\alpha$ and $\beta$. In particular,
$$
    \hspace{-0.5cm}
    \operatorname{aCov}\left\{ \widehat{C}_{H,\mathcal{C}}^{(1)}, \widehat{C}_{H,\mathcal{C}}^{(K)} \right\} = \operatorname{aCov}\left\{ S_{H,\mathcal{C}}^{(1)}, \begin{pmatrix}
        \hat{\alpha} \\ \hat{\beta}
    \end{pmatrix} \right\} \bm{\Xi}_H^{(K)} + \operatorname{aCov}\left\{ S_{H,\mathcal{C}}^{(K)}, \begin{pmatrix}
        \hat{\alpha} \\ \hat{\beta}
    \end{pmatrix} \right\} \bm{\Xi}_H^{(1)} + \left( \bm{\Xi}_H^{(1)} \right)^{\tp} \operatorname{aCov} \begin{pmatrix}
        \hat{\alpha} \\ \hat{\beta}
    \end{pmatrix} \bm{\Xi}_H^{(K)}
$$
where we use the fact that $S_{H,\mathcal{C}}^{(1)}$ and $S_{H,\mathcal{C}}^{(K)}$ are independent. 
Similar expressions can be found for the asymptotic covariance between multiple different higher order subgraphs.

As in Section~\ref{app:ho_theta_unknown} we derive the cross-covariance matrices between the subgraph densities and GMM estimator by a linear decomposition
$$
    \operatorname{aCov}\left\{ \begin{pmatrix} S_{H,\mathcal{C}}^{(1)} \\ S_{H,\mathcal{C}}^{(K)} \end{pmatrix}, \begin{pmatrix}
        \hat{\alpha} \\ \hat{\beta}
    \end{pmatrix} \right\} = \bm{\Lambda}_H \tilde{\bm{M}}(\psi_{\mathcal{C}})^{\tp}
$$
where $\tilde{\bm{M}}(\psi_{\mathcal{C}})$ is computed as part of the comparison model analog to Proposition~\ref{prop:gmm_adaptive_clt}, and $\Lambda_H$ can be found by computing the coordinate-wise covariances between $S^{(k)}_H$ and $\hat{\bm{m}}_h$ for fixed $k=1, K$ and $h=1,\ldots,p$.
Define
\begin{equation*}
  \sigma_{st,kh}(\theta) = \cov_{\theta}\left( \Ynet_{ij}^{(k)}, m_h(\Yvec_{ij}) ~\vert~ \Anetk{1}_{ij}=s, \Anetk{K}_{ij}=t \right)
\end{equation*}
for $s, t \in \{0,1\}$, so that
$$
  \cov_{\theta}(\Ynetk{k}_{ij},m_h(\Yvec_{ij})) = \sum_{s,t \in \{0,1\}} \sigma_{st,kh}(\theta) \mathbb{I}(\Anetk{1}_{ij}=s, \Anetk{K}_{ij}=t).
$$
Then
\begin{align*}
  &\cov_{\theta}\left( S^{(k)}_H,\left[ \binom{n}{2}^{1/2} (\hat{\bm{m}} - \bm{m}(\psi))\right]_h \right) \\
  = &\frac{1}{2 \lvert \mathcal{V} \rvert(1 - \alpha - \beta)} \sum_{j=1}^L (-1)^{1-\tau_j} \sum_{\bm{v} \in \mathcal{V}} \sum_{i_1 \neq i_2} \cov_{\theta}(\Ynetk{k}_{v_j},m_h(\Yvec_{i_1i_2})) \prod_{\ell \neq j} (\Anetk{k}_{v_{\ell}})^{\tau_{\ell}}(1 - \Anetk{k}_{v_{\ell}})^{1 - \tau_{\ell}} \\
  = &\frac{1}{2 \lvert \mathcal{V} \rvert(1 - \alpha - \beta)} \sum_{j=1}^L (-1)^{1-\tau_j} \sum_{\bm{v} \in \mathcal{V}} \cov_{\theta}(\Ynetk{k}_{v_j},m_h(\Yvec_{v_j}))  \prod_{\ell \neq j} (\Anetk{k}_{v_{\ell}})^{\tau_{\ell}}(1 - \Anetk{k}_{v_{\ell}})^{1 - \tau_{\ell}} \\
  = &\frac{1}{2 \lvert \mathcal{V} \rvert(1 - \alpha - \beta)} \sum_{j=1}^L (-1)^{1-\tau_j} \sum_{\bm{v} \in \mathcal{V}} \left\{ \sum_{s,t \in \{0,1\}} \sigma_{st,kh}(\theta) \mathbb{I}(\Anetk{1}_{v_j}=s, \Anetk{K}_{v_j}=t) \right\} \prod_{\ell \neq j}  (\Anetk{k}_{v_{\ell}})^{\tau_{\ell}}(1 - \Anetk{k}_{v_{\ell}})^{1 - \tau_{\ell}} \\
  = &\frac{1}{2 (1 - \alpha - \beta)} \sum_{s,t \in \{0,1\}} \sigma_{st,kh}(\theta) \sum_{j=1}^L (-1)^{1-\tau_j} \left\{ \frac{1}{\lvert \mathcal{V} \rvert} \sum_{\bm{v} \in \mathcal{V}} \mathbb{I}(\Anetk{1}_{v_{j}}=s, \Anetk{K}_{v_{j}}=t) \prod_{\ell \neq j}  (\Anetk{k}_{v_{\ell}})^{\tau_{\ell}}(1 - \Anetk{k}_{v_{\ell}})^{1 - \tau_{\ell}} \right\}
\end{align*}
Notice that for fixed $j$, $s$ and $t$, the final average in braces,
$$
    \frac{1}{\lvert \mathcal{V} \rvert} \sum_{\bm{v} \in \mathcal{V}} \mathbb{I}(\Anetk{1}_{v_{j}}=s, \Anetk{K}_{v_{j}}=t) \prod_{\ell \neq j}  (\Anetk{k}_{v_{\ell}})^{\tau_{\ell}}(1 - \Anetk{k}_{v_{\ell}})^{1 - \tau_{\ell}}
$$
is a subgraph density involving edges from both the 1st and $K$th underlying snapshots. These densities decompose the overall density of the $j$th leave-one-out subgraph density in the $k$th underlying snapshot, such that their sum over $s,t \in \{0,1\}$ is $U_H^{(-j,k)}$.
As in Section~\ref{app:ho_theta_unknown}, the asymptotic cross-covariance can be approximated as a linear combination of the entries of $\sigma_{st,kh}(\theta)$ for $s, t \in \{0,1\}$.

\begin{remark}
    When $H$ is the density of triangles or two-stars, we have
    $$
        U_H^{(-j,k)} \in \{ \subtchain{C}^{(k)}, \subhchain{C}^{(k)} \}
    $$
    for $k=1, K$.
    We denote the decomposition of these 4 subgraph densities by
    $$
        C_H^{(k)} = \sum_{s,t \in \{0,1\}} \bar{C}_{H}^{(k),st} 
    $$
    for $k=1, K$ and $H$ is either two-chains or half-chains.

    Simplifying for triangles, the entries of the cross-covariance are given by
    $$
        \left[ \subtri{\bm{\Lambda}} \right]_{kh} = \frac{3}{2(1-\alpha-\beta)} \sum_{s,t \in \{0,1\}} \sigma_{st,kh}(\theta) \subtchain{\bar{C}}^{(k),st},
    $$
    and for two-stars,
    $$
        \left[ \subtstar{\bm{\Lambda}} \right]_{kh} =  \frac{1}{2(1-\alpha-\beta)} \sum_{s,t \in \{0,1\}} \sigma_{st,kh}(\theta) \left\{ 2 \subhchain{\bar{C}}^{(k),st} - \subtchain{\bar{C}}^{(k),st} \right\}
    $$
    both for $k=1, K$ and $h$ indexing an edge-averaged moment in the GMM computation (see Section~\ref{sec:comparison}).
    As in Section~\ref{app:ho_theta_unknown}, these cross covariances can be consistently estimated by plugging in $\hat{\theta}$, and estimates of each of the 16 subgraph densities $\bar{C}_{H}^{(k),st}$ based on the 1st and $K$th observed snapshots.
\end{remark}

\section{Additional analysis of neuroimaging data} 
\label{app:realdata_2}

\begin{table}[ht]
\centering
\begin{tabular}{c ccc ccc}
\hline
 & \multicolumn{3}{c}{Resting} & \multicolumn{3}{c}{Listening} \\
Subject & $\Delta \delta$ & $\Delta \mathrm{cc}$ & $\Delta \mathrm{nt}$
        & $\Delta \delta$ & $\Delta \mathrm{cc}$ & $\Delta \mathrm{nt}$ \\
\hline
1  &   &   &   & - & - & + \\
2  & + &   &   & + &   & - \\
3  & - & + &   & + &   &   \\
4  &   &   &   & - & + &   \\
5  & - & - & + & - &   & + \\
6  & - &   &   &   &   &   \\
7  & + &   &   & + &   &   \\
8  & + &   &   & + & - & - \\
9  & + &   &   & - &   & + \\
10 & - & - & + & - & - & - \\
11 & - &   & + & + &   &   \\
12 & + &   & - & + &   &   \\
13 & + &   & - & + &   & - \\
14 & - &   &   & + &   & - \\
15 & + &   & + &   &   &   \\
16 & + &   & + &   & - & + \\
17 & + &   &   &   &   & + \\
18 & - &   &   & - &   &   \\
19 & + &   & - & + &   & - \\
20 & + &   & - & + &   & - \\
21 &   &   &   & - &   &   \\
22 & + &   & - & + &   & - \\
23 & - & - &   & - & + & + \\
\hline
\end{tabular}
\caption{Signs of significant changes in network summaries by condition, all subjects. \label{tab:neuro_tests}}
\end{table}

In this section, we present the full testing results for changes in edge density ($\Delta \delta)$, clustering coefficient ($\Delta \mathrm{cc}$), and normalized triangle density ($\Delta\mathrm{nt}$) at the Bonferroni-corrected 5\% level for the 23 subjects analyzed in Section~\ref{sec:real_data}. Table~\ref{tab:neuro_tests} also shows the signs of the significant changes, $+$ for increase and $-$ for decrease.

Due to the heterogeneity across subjects, as a case study we also present the detailed results for subject 1, for both resting state and listening task scans.
Table~\ref{tab:201_estimation} reports estimates of $(\delta_1,\delta_K,\alpha,\beta,\lambda,\mu)$ for both the resting state and listening task network sequences.
Table~\ref{tab:201_inference} reports estimates of $\Delta \delta$, $\Delta \mathrm{cc}$, and $\Delta \mathrm{nt}$.
All estimates are accompanied by asymptotic marginal $95\%$ confidence intervals. Asymptotic standard errors for higher-order summaries are based on a bootstrap sampler with $400$ samples. 
The results are reported in Tables~\ref{tab:201_estimation} and \ref{tab:201_inference}.

\begin{table}[ht]
\centering
\begin{tabular}{lcccc}
\toprule
 & \multicolumn{2}{c}{Rest} & \multicolumn{2}{c}{Listening} \\
\cmidrule(lr){2-3} \cmidrule(lr){4-5}
Parameter 
& Estimate & 95\% CI 
& Estimate & 95\% CI \\
\midrule
$\delta_1$ & 0.1674 & [0.1631, 0.1718] & 0.4623 & [0.4563, 0.4684] \\
$\delta_K$ & 0.1547 & [0.1504, 0.1589] & 0.1637 & [0.1591, 0.1682] \\
$\alpha$   & 0.0298 & [0.0276, 0.0320] & 0.0724 & [0.0686, 0.0761] \\
$\beta$    & 0.1573 & [0.1465, 0.1682] & 0.1086 & [0.1007, 0.1166] \\
$\lambda$  & 0.0078 & [0.0037, 0.0118] & 0.0002 & [-0.1392, 0.1396] \\
$\mu$      & 0.3579 & [0.2869, 0.4290] & 0.0180 & [-0.1361, 0.1721] \\
\bottomrule
\end{tabular}
\caption{Parameter estimates for resting state and listening task networks, subject 1. \label{tab:201_estimation}}
\end{table}

\begin{table}[ht]
\centering
\begin{tabular}{lcccc}
\toprule
 & \multicolumn{2}{c}{Rest} & \multicolumn{2}{c}{Listening} \\
\cmidrule(lr){2-3} \cmidrule(lr){4-5}
Parameter 
& Estimate & 95\% CI 
& Estimate & 95\% CI \\
\midrule
$\Delta\delta$ & -0.0127 & [-0.0175, -0.0080] & -0.2987 & [-0.3053, -0.2920] \\
$\Delta\mathrm{cc}$ & 0.0083 & [-0.0312, 0.0478] & 0.0575 & [-0.0248, 0.1399] \\
$\Delta\mathrm{nt}$ & -0.3490 & [-1.696, 0.9985] & 4.225 & [3.121, 5.330] \\
\bottomrule
\end{tabular}
\caption{Difference in (higher-order) summaries for resting state and listening task networks, subject 1. \label{tab:201_inference}}
\end{table}

We see only minor changes in properties of the resting state network: while the change in density is significant at the 5\% level, it is small in magnitude and does not survive a Bonferroni correction in the multi-subject data. 
Changes in the network coordination, measured by $\Delta \mathrm{cc}$ and $\Delta \mathrm{nt}$, are not significant.
On the other hand, the listening task networks show signficant changes which are large in magnitude in both $\Delta \delta$ and $\Delta \mathrm{nt}$. 
At the end of training, this subject shows a significant decrease edge density, suggesting less overall coordination in fMRI signals; but a significant increase in normalized triangle density, suggesting more tightly clustered regions of brain activity.

\bibliographystylesupp{abbrvnat}
\bibliographysupp{mybib0}

\end{document}